\documentclass[preprint,3p,fleqn]{elsarticle}

\usepackage{latexmake}

\usepackage{amsmath}
\usepackage{amsthm}
\usepackage{txfonts}
\usepackage{multirow}
\usepackage{slashbox}
\usepackage{subfig}
\usepackage{hyperref}
\usepackage{booktabs}
\usepackage{tabls}

\theoremstyle{plain}

\newtheorem{theorem}{Theorem}[section]
\newtheorem{corollary}[theorem]{Corollary}

\newcommand{\RealVect}[1]{{\mathbb R}^{#1}}

\def\RN{\RealVect{N}} 

\DeclareMathOperator{\spn}{span}

\DeclareMathOperator{\argmin}{arg\ min}

\newcommand{\gp}{\ensuremath{g^{\prime}}}
\newcommand{\req}[1]{Eq.~(\ref{#1})}

\begin{document}

\begin{frontmatter}
\title{Nonlinear Krylov Acceleration Applied to a Discrete Ordinates Formulation of the k-Eigenvalue Problem}
\author{Matthew T. Calef\corref{cor1}}
\ead{mcalef@lanl.gov}
\author{Erin D. Fichtl}
\ead{efichtl@lanl.gov}
\author{James S. Warsa} 
\ead{warsa@lanl.gov}
\author{Markus Berndt} 
\ead{berndt@lanl.gov}
\author{Neil N. Carlson}
\ead{nnc@lanl.gov}

\cortext[cor1]{Corresponding author}

\address{Computational Physics and Methods, Los Alamos National
Laboratory, Los Alamos, NM 87545-0001}

\begin{abstract}
We compare a variant of Anderson Mixing with the Jacobian-Free
Newton-Krylov and Broyden methods applied to an instance of the
$k$-eigenvalue formulation of the linear Boltzmann transport equation.
We present evidence that one variant of Anderson Mixing finds
solutions in the fewest number of iterations.  We examine and
strengthen theoretical results of Anderson Mixing applied to linear
problems.
\end{abstract}

\end{frontmatter}

\section{Introduction}

The $k$-eigenvalue formulation of the linear Boltzmann transport
equation is widely used to characterize the criticality of fissioning
systems~\cite{BellGlasstone:1970,LewisMiller:1993}. Physically, the
largest eigenvalue, generally denoted by $k$, is the effective neutron
multiplication factor that, in the equation, scales the fission
production term to achieve a steady-state solution. The corresponding
eigenmode describes the neutron flux profile for that steady-state
(i.e., critical) system and, when the system is close to criticality,
provides useful information about the distribution of the neutron
population in space and
velocity~\cite{BellGlasstone:1970}. Mathematically, the equation is a
standard eigenproblem for which power iteration is well-suited because
the eigenmode of interest is most commonly that with the largest
magnitude~\cite{Henry:1975}. For the deterministic $k$-eigenvalue
problem each step of a true power iteration incurs a heavy
computational cost due to the expense of fully inverting the transport
operator, therefore a nonlinear fixed point iteration is generally
employed in which an \textit{approximate} inversion of this operator
is performed at each iteration. In addition to power iteration, the
Implicitly Restarted Arnoldi Method has been applied to this problem
and has the advantage of being able to compute additional
eigenmodes~\cite{Warsa:2004}. However, the transport matrix
\textit{must} be fully inverted at each iteration, diminishing its
computational efficiency and attractiveness when only the dominant
eigenmode is desired.

Recently, more sophisticated nonlinear iteration methods employing
approximate inversion, predominantly Jacobian-Free Newton-Krylov
(JFNK), have been applied with great
success~\cite{Gill:2011,Gill:2011b,Knoll:2011,Park:2012}. However,
there has not yet been a comprehensive comparison of the performance
of JFNK with other nonlinear solvers.  This paper presents such a
comparison, examining the performance of three nonlinear
solvers\textemdash JFNK, Broyden's Method and Anderson
Mixing\textemdash applied to a particular formulation of the
$k$-eigenvalue problem. A variant of Anderson
Mixing~\cite{Anderson:1965}, first described in~\cite{CarlsonMiller1},
is of particular interest because, in the experience of the authors,
it is frequently computationally more efficient than JFNK and
Broyden's method.

JFNK is an inexact Newton's method in which the inversion of the
Jacobian is performed to arbitrary precision using a Krylov method
(most commonly GMRES) and the Jacobian itself is never formed, but
rather its action is approximated using finite differences of
arbitrarily close state data. JFNK can be expected to converge
quadratically in a neighborhood containing the solution
(cf.~\cite{KnollKeyes1} and the references therein).  Each iteration
of JFNK requires a nested `inner' iteration and the bulk of the
computational effort is expended in this `inner' Krylov inversion of
the Jacobian at each `outer' Newton step.  At the end of each
inversion, the accumulated Krylov space is discarded even though the
Jacobian is expected to change minimally during the final Newton steps
when a similar space will be rebuilt in the next Newton iteration. In
effect, at the end of each iteration, JFNK discards information that
may be of use in successive iterations.

In its standard formulation Broyden's method (cf.~\cite{Kelley1}),
like low memory BFGS (cf.~\cite{NocedalWright1}), uses differences in
state from successive iterations to make low rank updates to the
Jacobian.  The Sherman-Morrison-Woodbury update rule is then used to
compute the action of the inverse of the Jacobian after such an
update.  While Broyden's method is restricted to low-rank updates, it
provides an explicit representation of the Jacobian allowing one to
employ the Dennis-Mor\'e condition~\cite{DennisMore:1974} to show that
it converges super-linearly in a neighborhood containing the solution.
Further, it has been shown to solve linear problems of size $N$ in at
most $2N$ iterations (cf.~\cite{Kelley1} and the references therein.)

Anderson Mixing~\cite{Anderson:1965} uses differences in state from
successive iterations to infer information about the inverse of the
Jacobian, which is assumed to be roughly constant in a neighborhood
containing all the iterates.  The updates can
be of arbitrary rank.  Recent results by Walker and
Ni~\cite{WalkerNi:2010} show that, with mild assumptions, Anderson
Mixing applied to a linear problem performs as well as the generalized
minimum residual method (GMRES)~\cite{SaadSchultz1}.  In this regard,
Anderson Mixing may be thought of as a nonlinear version of GMRES.  In
independent work, Carlson and Miller formulated a so-called nonlinear
Krylov acceleration~\cite{CarlsonMiller1} method which we show to be a
variant of Anderson Mixing.  Further, we examine the hypotheses of a
central theorem presented by Walker and Ni and argue that they will,
with high probability, be met and that they can be omitted entirely if
one is willing to accept a small performance penalty.  While Anderson
Mixing performs best for our numerical experiments, there is no theory
that the authors of this paper know of that can characterize its
performance for nonlinear problems.

In this paper we provide theoretical and computational examinations of
Carlson and Miller's nonlinear Krylov acceleration and strengthen
theory about Anderson Mixing in general, comparing theoretical results
for Anderson Mixing with those for the Broyden and Jacobian-Free
Newton-Krylov methods.  In our computational investigations, we
compare the performance of these methods on computational physics
problems derived from problems described in~\cite{Warsa:2004}
and~\cite{C5G7}; these are instances of the $k$-eigenvalue formulation
of the Boltzmann neutron transport equation, which is commonly used to
characterize the neutron multiplication of fissile systems.  The rest
of this paper is organized as follows. Section~\ref{sec:nka-theory}
reviews and strengthens the theory regarding Anderson Mixing, and
provides an overview about what is known theoretically about the
Broyden and JFNK methods.  Section~\ref{sec:transport} describes the
formulation of the $k$-eigenvalue problem.  Section~\ref{sec:results}
provides the results of our numerical experiments and in
Section~\ref{sec:concl} we conclude with a summary of this work.

\section{Background of Anderson Mixing and nonlinear Krylov acceleration}\label{sec:nka-theory}

\subsection{Nonlinear Krylov acceleration}

In~\cite{CarlsonMiller1,CarlsonMiller2} Carlson and Miller outlined an
iterative method, dubbed nonlinear Krylov acceleration or NKA, for
accelerating convergence of fixed-point iteration by using information
gained over successive iterations.  The problem they consider is to
find a root of the function $\RN \ni {\bf x} \to f({\bf x}) \in \RN$.
One approach is to apply a fixed-point iteration of the form
\begin{equation}\label{mtc:eq:1}
{\bf x}_{n+1} = {\bf x}_n  - f({\bf x}_n).
\end{equation}
It is assumed that this basic iteration converges in the usual geometric
manner, albeit more slowly than would be acceptable.  In particular, it
converges when $\| {\bf I} - Df \| < 1$, where ${\bf I}$ is the identity and $Df$ denotes
the derivative of $f$.  The closer $Df$ is to the identity the more rapid the
convergence.  The question at hand is how this basic iteration
can be modified so as to accelerate the convergence.  Observing that Eq.~\eqref{mtc:eq:1}
can be viewed as an approximate Newton iteration where $Df^{-1}$ has been
replaced by $I$, the motivation behind NKA is instead to 
approximate $Df^{-1}$ using information from previous iterates, improving
that approximation over successive iterations, and in cases where no
applicable approximation is available, revert to a fixed-point
iteration where $Df^{-1}$ is replaced with $I$.

NKA requires an initial guess ${\bf x}_0$ and at the $n^\text{th}$
invocation provides an update ${\bf v}_{n+1}$ that is used to derive
the $n+1^\text{st}$ iterate from the $n^\text{th}$. This method may be
written as
\begin{eqnarray*}
{\bf v}_{n+1} =& \text{NKA}[f({\bf x}_n),\ldots] \\
{\bf x}_{n+1} =& {\bf x}_n - {\bf v}_{n+1},
\end{eqnarray*}
where NKA$[f({\bf x}_n),\ldots]$ is the update computed by the
nonlinear Krylov accelerator.  We use the brackets and ellipsis to
indicate that NKA is stateful and draws on previous information when
computing an update.

On its first invocation ($n=0$) NKA has no information and
simply returns $f({\bf x}_0)$. At iteration $n>0$ it has access to the
$M$ vectors of differences for some natural number $M\in (0, n)$,
\begin{equation*}
{\bf v}_i = {\bf x}_{i-1} - {\bf x}_i \quad\text{and}\quad {\bf w}_i = f({\bf
  x}_{i-1}) - f({\bf x}_i) \quad \text{for $i=n-M+1,\ldots,n$,}
\end{equation*}
and where, for convenience, we shall define ${\mathcal W}_n$ to be the
span of the ${\bf w}_i$ vectors available at iteration $n$.  If $f$
has a constant and invertible derivative, $Df$, then we would have
\begin{equation}\label{mtc:eq:1a}
Df {\bf v}_i = {\bf w}_i\qquad\text{and}\qquad Df^{-1} {\bf w}_i =
{\bf v}_i.
\end{equation}  
We denote by ${\mathcal P}_{{\mathcal W}_n}$ the operator that
projects onto the subspace ${\mathcal W}_n$ and write the
identity
\begin{equation*}
f({\bf x}_n) = 
{\mathcal P}_{{\mathcal W}_n} f({\bf x}_n) + (f({\bf  x}_n) - 
{\mathcal P}_{{\mathcal W}_n} f({\bf x}_n)).
\end{equation*}
Note that $f({\bf x}_n) - {\mathcal P}_{{\mathcal W}_n} f({\bf x}_n)$
is orthogonal to ${\mathcal W}_n$. If the ${\bf w}_i$ vectors are
linearly independent, then there is a unique set of coefficients ${\bf
  z}^{(n)} := (z_1^{(n)}, z_2^{(n)},\ldots, z_n^{(n)}) \in
\RealVect{n}$ so that
\begin{equation*}
{\mathcal P}_{{\mathcal W}_n} f({\bf x}_n) = \sum_{i=n-M+1}^n
z_i^{(n)}{\bf w}_i,
\end{equation*}
and hence by Eq.~\eqref{mtc:eq:1a} 
\begin{equation*}
Df^{-1}{\mathcal P}_{{\mathcal W}_n} f({\bf x}_n) = 
Df^{-1}\sum_{i=n-M+1}^n z_i^{(n)}{\bf w}_i = 
\sum_{i=n-M+1}^n z_i^{(n)}{\bf v}_i.
\end{equation*}
The idea motivating Carlson and Miller is to project $f({\bf x}_n)$
onto ${\mathcal W}_n$ (the space where the action of $Df^{-1}$ is
known), compute that action on the projection and, for lack of
information, apply a fixed-point update given in Eq.~\eqref{mtc:eq:1}
on the portion of $f({\bf x}_n)$ that is orthogonal to ${\mathcal
W}_n$.  The resulting formula for ${\bf x}_{n+1}$ is
\begin{equation}\label{nnc:eq:1}
\begin{array}{cc}
{\bf v}_{n+1} =&   
\left[\sum_{i=n-M+1}^nz_i^{(n)}{\bf v}_i  + \left(f({\bf x}_n) - \sum_{i=n-M+1}^n z_i^{(n)}{\bf
    w}_i \right)\right],\\
{\bf x}_{n+1} =& {\bf x}_n - {\bf v}_{n+1},
\end{array}
\end{equation}
where the vector of coefficients in the orthogonal projection, ${\bf
  z}^{(n)}$, is the solution to the projection, alternatively
minimization, problem
\begin{equation*}
{\bf z}^{(n)} = \argmin_{{\bf y}\in \RealVect{M}}\left\| f({\bf x}_n)
- \sum_{i=n-M+1}^n y_i{\bf w}_i \right\|_2.
\end{equation*} 

\subsection{NKA as Anderson mixing}

Carlson and Miller had essentially rediscovered, albeit in a slightly
different form, the iterative method presented much earlier by
Anderson in~\cite{Anderson:1965}, which is now commonly referred to as
Anderson Mixing or Anderson Acceleration.  Anderson was studying the
problem of finding a fixed point of some function $\RN \ni {\bf x} \to
G({\bf x}) \in\RN$.  In Section 4 of~\cite{Anderson:1965} he defines a
fixed point residual for iteration $i$ as
\begin{equation*} 
{\bf r}_i = G({\bf x}_i) - {\bf x}_i,
\end{equation*}
which can be used to measure how ${\bf x}_i$ fails to be a fixed
point.  With this Anderson proposed an updating scheme of the
form\footnote{Anderson uses $\theta_i^n$ to denote the $\text{i}^{th}$
  update coefficient of the $\text{n}^{th}$ iterate, where we have
  written $\tilde z_i^{(n)}$ in keeping with the presentation of NKA.}
\begin{equation*}
{\bf x}_{n+1} = {\bf x}_n - \left[
\sum_{i=1}^M\tilde z_i^{(n)} ({\bf x}_n - {\bf x}_{n-i}) - \beta_n\left(
{\bf r}_n - \sum_{i=1}^M\tilde z_i^{(n)}[{\bf r}_n - {\bf r}_{n-i}]
\right)
\right]
\end{equation*}
where the vector of coefficients, $\tilde {\bf z}^{(n)} \in
\RealVect{M}$, is chosen to minimize the quantity
\begin{equation*}
\left\|{\bf r}_n - \sum_{i=1}^M\tilde z_i^{(n)}[{\bf r}_n - {\bf r}_{n-i}] \right\|_2,
\end{equation*}
and where Anderson requires that $\beta_n >0$.  Here Anderson
considers a depth of $M$ difference-vectors.

The problem of finding a root of $f$ may be recast as a fixed point
problem by defining $G({\bf x}) = {\bf x} - f({\bf x})$, and then
${\bf r}_i = -f({\bf x}_i)$.  If, for each $n$, one defines the vectors
\begin{equation*}
\tilde {\bf v}_i^{(n)} = {\bf x}_n - {\bf x}_{n-i} 
\quad\text{and}\quad 
\tilde {\bf w}_i^{(n)} = f({\bf x}_n) - f({\bf x}_{n-i})
\quad\text{for $i=1,\ldots, M$},
\end{equation*}
then the Anderson Mixing update becomes
\begin{equation}\label{nnc:eq:2}
{\bf x}_{n+1} = {\bf x}_n - \left[
\sum_{i=1}^M\tilde z_i^{(n)} \tilde {\bf v}_i^{(n)} + \beta_n\left(
f({\bf x}_n) - \sum_{i=1}^M\tilde z_i^{(n)}\tilde {\bf w}_i^{(n)}
\right)
\right],
\end{equation}
where 
\begin{equation}\nonumber
\tilde {\bf z}^{(n)} 
= \argmin_{{\bf y} \in \RealVect{M}} 
\left\|f({\bf x}_n) - \sum_{i=1}^M y_i\tilde {\bf w}_i^{(n)} \right\|_2.
\end{equation}

Note that the ${\bf v}_i$ and ${\bf w}_i$ vectors are the differences of successive iterates and differences of function evaluations of successive iterates respectively.  In contrast the $\tilde {\bf v}_i^{(n)}$ vectors are the differences between the most recent iterate and all of the previous iterates.  Similarly the $\tilde {\bf w}_i^{(n)}$ vectors are the differences between the most recent function evaluation and all the previous function evaluations. We append the superscript $(n)$ to the vectors marked with a tilde to indicate that this quantity is only used for the $n^\text{th}$ iterate.  It is important also to note that $\spn \{\tilde {\bf w}_1^{(n)},\ldots,\tilde {\bf
  w}_M^{(n)}\} = \spn \{{\bf w}_{n-M+1},\ldots,{\bf w}_n\} = {\mathcal
  W}_n$. In particular
the update coefficients of NKA and Anderson Mixing are related by the
formula
\begin{equation}\label{mtc:eq:6}
z_i^{(n)} = -\sum_{j=n-i+1}^{M}\tilde z_{j}^{(n)}.
\end{equation} Further, for a constant and invertible derivative,
$Df^{-1}\tilde{\bf w}_i^{(n)} = \tilde{\bf v}_i^{(n)}$.  With this it
is clear that NKA is equivalent to Anderson
Mixing when $\beta_n = 1$; the difference is only in the choice of
basis vectors that span ${\mathcal W}_n$ resulting in a change of
variables given by Eq.~\eqref{mtc:eq:6}. One advantage of the NKA
formulation is that, because it uses differences of successive function
evaluations to form a basis for ${\mathcal W}_n$, these differences
can be used for $M$ successive iterations.  In contrast, Anderson's
formulation forms a basis for ${\mathcal W}_n$ using differences of
the most recent function evaluation with all previous evaluations,
which are be recomputed at every iteration.

\subsection{Analytic examinations of Anderson mixing applied to a linear problem}

In~\cite{WalkerNi:2010}, Walker and Ni consider the behavior of a
generalized Anderson Mixing algorithm when it is applied to linear
problems and when all previous vectors are stored, i.e. when $M=n$.
They prove that, under modest assumptions, a class of variants of
Anderson Mixing are equivalent to GMRES.  They consider Anderson
Mixing as a means to find a fixed-point of $G$.  Like Anderson they
compute ${\bf r}_i = G({\bf x}_i) - {\bf x}_i$ at each step.

In their presentation of Anderson Mixing, though, they compute the
update coefficients\footnote{Walker and Ni used $f_j$ and
  $\alpha_n^{(j)}$ where we, for consistency, are using ${\bf r}_j$
  and $ \tilde z_j^{(n)}$.}
\begin{equation*}
\tilde {\bf  z}^{(n)} = \argmin_{y \in \RealVect{n+1}} 
\left\|\sum_{i=0}^ny_i{\bf r}_i \right\|_2 
\qquad\text{subject to}\qquad
\sum_{i=0}^n  \tilde z_i^{(n)} = 1,
\end{equation*}
which are then used to form the update 
\begin{equation*}
{\bf x}_{n+1} = \sum_{i=0}^n \tilde z_i^{(n)} G({\bf x}_i).
\end{equation*}
Again choosing $G({\bf x}) = {\bf x} - f({\bf x})$ (and hence ${\bf r}_i =
-f({\bf x}_i)$) and noting that the constraint requires that
\begin{equation*}
\tilde z_n^{(n)} = 1 - \sum_{i=0}^{n-1} \tilde z_i^{(n)},
\end{equation*}
one has, as noted by Walker and Ni, that the above is equivalent to
Anderson's original formulation
\begin{equation*}
\tilde {\bf  z}^{(n)} = \argmin_{y \in \RealVect{n}} 
\left\|f({\bf x}_n) - \sum_{i=0}^{n-1} y_i(f ({\bf x}_n) - f({\bf x}_i))\right\|_2
\end{equation*}
and
\begin{equation*}
{\bf x}_{n+1} = {\bf x}_n - 
\left[
\sum_{i=0}^{n-1} \tilde z_i^{(n)}({\bf x}_n - {\bf x}_i) +
\left(
f({\bf x}_n) - 
\sum_{i=0}^{n-1} \tilde z_i^{(n)}(f({\bf x}_n) - f({\bf x}_i)) 
\right)
\right].
\end{equation*}
Because Anderson Mixing and NKA are equivalent to each other and to
the Walker and Ni formulation, we may restate one of Walker's and Ni's
theorems as follows:
\begin{theorem}[Walker and Ni~\cite{WalkerNi:2010}]\label{mtc:th:1}
Suppose that we search for a root of the function $f({\bf x}) = {\bf
  Ax} - {\bf b}$ starting with ${\bf x}_0$ using either the NKA update
or Anderson's update with $\beta_n =1$.  Suppose further that ${\bf
  A}$ is non-singular.  Let ${\bf x}_n^{GMRES}$ denote the
$n^\text{th}$ iterate generated by GMRES applied to the problem ${\bf
  Ax} = {\bf b}$ with starting point ${\bf x}_0$ and let ${\bf
  r}_n^{GMRES} = {\bf A}{\bf x}_n^{GMRES} - {\bf b}$ denote the
associated residual. If for some $n >0$, ${\bf r}_{n-1}^{GMRES} \ne 0$
and $\| {\bf r}_j^{GMRES} \|_2 < \|{\bf r}_{j-1}^{GMRES}\|_2$ for all
$0<j<n$, then the $n+1$ iterate generated by either update rule is
given by ${\bf x}_{n+1} = ({\bf I} - {\bf A}){\bf x}_n^{GMRES} + b$.
\end{theorem}

It should be noted that the original presentation of the theorem given
in~\cite{WalkerNi:2010} applies to a class of orthogonal projection
methods where we have presented the theorem as applied to two members
of this class: Anderson's original method and the NKA variant.  It
should also be noted that this theorem of Walker and Ni shows that 
\begin{equation*}
{\mathcal W}_n$ = ${\bf A}{\mathcal K}_n,
\end{equation*} 
where ${\mathcal K}_n = \spn\{{\bf r}_0, {\bf A}{\bf r}_0,\ldots,{\bf
  A}^{n-1}{\bf r}_0\}$ is the $n^\text{th}$ Krylov space associated with ${\bf A}$ and ${\bf r}_0$.

An immediate corollary to Theorem~\ref{mtc:th:1} is
\begin{corollary}\label{mtc:cor:1}
Let ${\bf r}_n = {\bf A}{\bf x}_n - {\bf b}$ denote the residual
associated with the $n^\text{th}$ iterate generated by either Anderson
Mixing or NKA.  Under the assumptions of Theorem~\ref{mtc:th:1}
\begin{equation*}
\|{\bf r}_{n+1} \|_2 \le \|{\bf I} - {\bf A}\|\,\|{\bf r}_n^{GMRES}\|_2,
\end{equation*}
where $\|\cdot\|$ is the $L^2$ operator norm.
\end{corollary}
\begin{proof}
\begin{eqnarray*}
{\bf r}_{n+1} =& {\bf Ax}_{n+1} - {\bf b}\\
=& {\bf A}\left[
({\bf I}
- {\bf A}){\bf x}_n^{GMRES} + b
\right] - {\bf b}\\
=& {\bf Ax}_n^{GMRES} - {\bf b} - {\bf A}({\bf Ax}_n^{GMRES} -
{\bf b})\\
=& ({\bf I} - {\bf A})({\bf Ax}_n^{GMRES} -
{\bf b}),
\end{eqnarray*}
from which the claim follows.
\end{proof}

Recall that 
\begin{equation*}
\|{\bf r}_n^{GMRES}\|_2 = 
\min_{{\bf y} \in {\mathcal K}_n}\|{\bf Ay} - {\bf b}\|_2.
\end{equation*}
From this we see the value of Corollary~\ref{mtc:cor:1}. In the linear case,
convergence of the non-truncated versions of both NKA and Anderson
Mixing have the same characterizations as GMRES, i.e. when considering
a normal matrix ${\bf A}$, the residual is controlled by the spectrum
of ${\bf A}$, provided that GMRES does not stagnate.

\subsection{Non-stagnation of GMRES}

When considering a linear problem, the coefficients $\tilde {\bf
  z}^{(n)}$ (Anderson Mixing) and ${\bf z}^{(n)}$ (NKA) do two things.
Through a minimization process, they select the best approximation
within a given subspace, and they also select the subspace for the
next iteration.  The failure mode is that the best approximation
within a given subspace does not use the information from the most
recent iteration, in which case the next subspace will be the same as
the current one and Anderson Mixing and NKA will become trapped.
Anderson briefly discusses this in his presentation of the
method\footnote{One must choose $\beta_n \ne 0$ otherwise the
  subspaces over which Anderson Mixing and NKA search will not
  expand.}.  Lemma 2.1.5 from Ni's dissertation~\cite{Ni:2009}
addresses this more directly.  What Walker and Ni recognize
in~\cite{WalkerNi:2010} is that this failure mode corresponds to the
stagnation of GMRES.

There have been several examinations of when GMRES stagnates.
Zavorian, O'Leary and Elman in~\cite{ZOE1} and Greenbaum, Ptak and
Strakos in~\cite{GPS1} present examples where GMRES does not decrease
the norm of the residual on several successive iterations, i.e. it
stagnates.  While GMRES converges for such cases, Anderson Mixing and
NKA will not.  Greenbaum and Strakos show in~\cite{GreenbaumStrakos1}
that the residual for GMRES applied to a matrix ${\bf A}$ is strictly
decreasing, i.e. no stagnation, if $\langle {\bf r} , {\bf A}{\bf
  r}\rangle \ne 0$ for all ${\bf r}$ satisfying $\|{\bf r}\|_2 \ne 0$.
This in turn will ensure the convergence of Anderson Mixing and NKA.

It is important to bear in mind that, in practice, the subspace
${\mathcal W}_n$ is likely to have modest dimension, say $10$, and is
a subspace in $\RN$ where $N$ is much larger, often on the order of
thousands or millions.  The slightest perturbation of a vector will,
with probability one, push it off of this low dimensional space.  Even
so, there is a simple change to the update step that provides a
theoretical guarantee that NKA will still converge even when GMRES
stagnates.  This guarantee comes at a slight performance cost and is
accomplished by modifying the update coefficients to ensure that
${\mathcal W}_{n-1} \subsetneq {\mathcal W}_n$.

We consider a modification to the NKA update rule\footnote{The
  following could be adapted to many formulations of Anderson Mixing}:
\begin{equation*}
{\bf x}_{n+1} = {\bf x}_n - 
\left[\sum_{i=1}^n{\bf v}_i z_i^{(n)} + \left(f({\bf x}_n) - \sum_{i=1}^n {\bf
    w}_i z_i^{(n)}\right)\right],
\end{equation*}
where ${\bf z}^{(n)}$ is chosen to minimize
\begin{equation*}
\left\| f({\bf x}_n) - \sum_{i=1}^n {\bf w}_i z_i^{(n)}\right\|_2.
\end{equation*} 
We now add the following safety check: if $\|{\bf w}_n\|_2 \ne 0$,
then perform the following update
\begin{equation*}
z_n^{(n)} \leftarrow z_n^{(n)} \pm \varepsilon 
\frac{\left \|f({\bf x}_n) - \sum_{i=1}^n {\bf w}_i z_i^{(n)}\right\|_2}
{\|{\bf w}_n\|_2}
\end{equation*}
for some positive $\varepsilon$, where one chooses to add or subtract
based on which option will maximize $|z_n^{(n)} + 1|$.  As will be shown in the proof of
Theorem~\ref{mtc:th:2}, the numerator of the modification is zero only
when NKA has found the root.

With this we can strengthen Corollary~\ref{mtc:cor:1} as it applies to
NKA as follows:

\begin{theorem}\label{mtc:th:2}
Let ${\bf A}$ be non-singular square matrix and suppose that we search
for a root of the function $f({\bf x}) = {\bf Ax} - {\bf b}$ starting
with ${\bf x}_0$ using the modified version of NKA. Let ${\bf
  r}_n^{GMRES} = {\bf A}{\bf x}_n^{GMRES} - {\bf b}$ denote the
$n^\text{th}$ GMRES residual and let ${\bf r}_n$ denote the
$n^\text{th}$ NKA residual, then
\begin{equation*}
\|{\bf r}_{n+1} \|_2 \le (1+\varepsilon)\|{\bf I} - {\bf A}\|\,\|{\bf r}_n^{GMRES}\|_2.
\end{equation*}
\end{theorem}
The proof can be found in Appendix~\ref{apdx:1}. The limitation of this result
is that the orthogonal projection will become more poorly conditioned
as $\varepsilon$ decreases.  Note that for the same reason both
Anderson Mixing and NKA can develop arbitrarily poorly conditioned
projection problems.

\subsection{Relationship between Anderson mixing and Broyden and the choice of mixing parameter}

For the numerical experiments we performed, we used a variant of
Broyden's method where, at the $n^\text{th}$ iteration, the
approximate derivative of $f$ maps ${\bf v}_n$ to ${\bf w}_n$.  There
are generalizations of Broyden's method, sometimes known as ``bad''
Broyden, where the approximate inverse of the derivative of $f$ maps
${\bf w}_i$ onto ${\bf v}_i$ for some range of $0 \le i \le n$.
In~\cite{Eyert:1996} Eyert recognized the equivalence of ``bad''
Broyden and Anderson Mixing.

Fang and Saad, in~\cite{FangSaad:2009}, extend this comparison and
present a range of multisecant methods including generalizations of
Broyden, Anderson Mixing and Eirola-Nevanlinna-like methods.  In their
presentation of Anderson Mixing they form a basis for ${\mathcal W}_n$
by using differences of successive function evaluations, rather than
differences of the most recent function evaluation and the previous
function evaluations, as is done in this work.  Adapting Eq. (24)
of~\cite{FangSaad:2009} to the notation used here\footnote{The mapping
  between Fang and Saad's notation and ours is ${\bf v}_{k} = -\Delta
  x_{k-1}$, ${\bf w}_{k} = -\Delta f_{k-1}$ and the coefficients are
  given by ${\bf z}_k = -\gamma_k$}, they consider Anderson Mixing in
the form

\begin{equation*}
{\bf x}_{n+1} = {\bf x}_n - 
\left[\sum_{i=n-M+1}^n{\bf v}_i z_i^{(n)} + (-\beta) \left(f({\bf x}_n) - \sum_{i=n-M+1}^n {\bf
    w}_i z_i^{(n)}\right)\right],
\end{equation*}
where
\begin{equation*}
{\bf z}^{(n)} = \argmin_{{\bf y}\in \RealVect{M}}\left\| f({\bf x}_n)
- \sum_{i=n-M+1}^n {\bf w}_i y_i\right\|_2,
\end{equation*} 
and where $\beta$ is a mixing parameter.  Comparing this description
of Anderson mixing to Eqs.~\eqref{nnc:eq:1} and~\eqref{nnc:eq:2} we find that $-\beta$ corresponds
to Anderson's $\beta_n$ relaxation parameters.\footnote{
  We believe that Fang and Saad implicitly made a modification to Anderson Mixing in the first line of Eq. (24) of [24],
  and that the next Anderson iterate is
  $x_{k+1} = \tilde{x}_k - \beta\tilde{f}_k$.
}
However, where Anderson takes the $\beta_n > 0$ (and NKA takes $\beta_n=1$)
Fang and Saad only consider positive $\beta$ in their numerical
experiments; that is, \emph{negative} $\beta_n$.  The assumption
behind the derivation of both Anderson mixing and NKA is that the
basic fixed point iteration given in Eq.~\eqref{mtc:eq:1} is convergent, and from this it
follows that the $\beta_n$ should be positive, with $\beta_n=1$ often
the most appropriate choice.  

The effective $f({\bf x})$ in our formulation
of the $k$-eigenvalue problem described in Section~\ref{sec:transport} is specifically
designed (via the transport sweeps) so that~\eqref{mtc:eq:1} is convergent.
Simply replacing $f({\bf x})$ with $-f({\bf x})$, for example, gives a very different
fixed point iteration where negative $\beta_n$ would be appropriate.
To illustrate what happens when a $\beta_n$ of the ``wrong'' sign is used,
our experiments include a variant of Anderson mixing
with $\beta_n = -1$, which we denote by $NKA_{-1}$.  The poor performance of this method on
our test problems underscores the importance of understanding the character
of $f({\bf x})$ in the context of (1), and perhaps modifying it through scaling
or more elaborate preconditioning.

\subsection{NKA in practice}

For our experiments we shall use the NKA and NKA$_{-1}$ formulation of Anderson
Mixing.  Each ${\bf w}_i$ vector is rescaled to have
unit-norm, and the associated ${\bf v}_i$ vector is scaled by the same
factor. In our implementation we use a Cholesky Factorization to solve
the least-squares problem associated with the orthogonal projection.
This requires that the Gramian of the ${\bf w}_i$ vectors be positive
definite.  We enforce this by choosing a linearly independent subset
of the ${\bf w}_i$ vectors as follows: At the beginning of iteration
$n$ we start with ${\bf w}_n$ and consider ${\bf w}_{n-1}$.  If the
angle between ${\bf w}_n$ and ${\bf w}_{n-1}$ is less than some
tolerance we discard ${\bf w}_{n-1}$.  We iterate in this manner over
successively older vectors, keeping them if the angle they make with
the space spanned by the kept ${\bf w}_i$ vectors is greater than the
tolerance.  This ensures that we may use Cholesky factorization and
that the condition number of the problem is bounded.

Memory constraints require that $M<n$, forcing us to use NKA,
NKA$_{-1}$ and Broyden in a setting for which there are no theoretical
results that the authors of this paper know of.  One important
distinction between Broyden and NKA is that for each iteration NKA
stores a pair of state vectors, while Broyden only stores one.
Consequently the memory requirements for Broyden are half that of NKA
for the same depth of stored vectors.

Depending on the tolerance chosen at each linear solve for JFNK, one
can reasonably expect theoretical guarantees of quadratic convergence
to hold in some neighborhood of the solution, however identifying that
neighborhood is often considerably harder than solving the original
problem.  In summary, for the numerical experiments we present, there
is little theory regarding performance, making the following numerical
results of value.

\section{The $k$-eigenvalue formulation of the Boltzmann transport equation}
\label{sec:transport}

The problem that we use to compare the efficiency of these various
nonlinear methods is the $k$-eigenvalue formulation of the Boltzmann
neutron transport equation.
\begin{equation}
\label{Boltzmann}
\left[ \hat{\Omega}\cdot\vec{\nabla} + \Sigma_t(\vec{r},E)
\right]\psi(\vec{r},\hat{\Omega},E) = \int dE^\prime\int
d\Omega^\prime\Sigma_s(\vec{r},E^\prime\rightarrow
E,\hat{\Omega}^\prime\cdot\hat{\Omega})\psi(\vec{r},\hat{\Omega},E^\prime)
+ \frac{1}{k}\int dE^\prime\chi(E^\prime\rightarrow
E)\bar{\nu}\Sigma_f(\vec{r},E')\phi(\vec{r},E^\prime).
\end{equation}
The unknown $\psi$ describes the neutron flux in space,
$\vec{r}\in\RealVect{3}$, angle $\hat{\Omega} \in {\mathbb S}^2$ and
energy $E$.  $\Sigma_t$ is the total cross section, $\Sigma_s$ is the
scattering cross section and $\Sigma_f$ is the fission cross section.
The quantities $\bar\nu$ and $\chi$ characterize the rate and energy
spectrum of neutrons emitted in the fission process.  Integrating the
flux over angle gives $\phi$, the scalar flux at a given position and
energy. For a thorough examination of the mathematical models of
fission, including the Boltzmann transport equation,
cf.~\cite{BellGlasstone:1970,LewisMiller:1993}.

Discretization of \req{Boltzmann} is accomplished using
\begin{enumerate}
\item S$_N$ (discrete ordinates) in angle: We use $S_4$ and $S_6$ level symmetric
  quadrature sets and the solution is computed at the
  abscissas of that set and then integrated over angle using the
  quadrature weights.
\item Multigroup in energy: Cross sections can be very noisy and it is
  therefore not practical to discretize the energy variable as one
  would a continuous function. Instead, the energy space is divided up
  into groups and it is assumed that the energy component for a cross
  section $\sigma$ can be separated out:
  \begin{eqnarray*} \sigma(\vec{r},E) &\approx&
    f(E)\sigma_g(\vec{r}),\qquad E_g<E\leq E_{g-1}, \\
    \sigma_g &=& \frac{\int_{E_g}^{E_{g-1}}dE\sigma(\vec{r},E)}
    {\int_{E_g}^{E_{g-1}}dE f(E)}.
  \end{eqnarray*}
\item Finite element or difference in space: The spatial variable can
  be treated in a variety of different ways. Here we explore results
  from two different discretization strategies as employed by the Los
  Alamos National Laboratory production transport codes
  PARTISN~\cite{PARTISN} and Capsaicin. PARTISN is used to generate
  results on a structured mesh with diamond (central) difference and
  Capsaicin to generate results on an unstructured polygonal mesh using
  discontinuous finite elements.
\end{enumerate}
Applying the S$_N$ and multigroup approximations, the $k$-eigenvalue
equation takes the form
\begin{equation}\label{BT eqn}
  \left( \hat{\Omega}_m \cdot \nabla + \sigma_{t,g}(\vec{r})
  \right) \psi_{g,m} (\vec{r}) = \dfrac{1}{4 \pi} \sum_{\gp=1}^{G}
  \sigma_{s,\gp \rightarrow g} (\vec{r}) \phi_{\gp}(\vec{r}) +
  \dfrac{1}{k} \sum_{\gp=1}^{G} \bar\nu
  \sigma_{f,\gp}(\vec{r}) \dfrac{\chi_{g^\prime\rightarrow g}(\vec{r})}{4 \pi} 
  \phi_{\gp}(\vec{r}).
\end{equation}
Here, 
\begin{itemize}
\item $\psi_{g,m}$ represents the angular flux in direction
$\hat{\Omega}_m$ in energy group $g$, which is the number of neutrons
passing through some unit area per unit time
$\left(\frac{\#}{cm^2\cdot Mev\cdot ster\cdot sec}\right)$,
\item $\phi_{g}$ is the scalar flux, or angle-integrated angular
flux. The S$_N$ quadrature used to define angular abscissas
($\hat{\Omega}_m$) can be used to integrate $\psi$ over all angle: 
$\phi_g = \sum_{m=1}^M w_m \psi_{g,m}$ where $w_m$ are the quadrature
weights $\left(\frac{\#}{cm^2\cdot Mev\cdot sec}\right)$,
\item $\sigma_{t,g}$ is the total cross section, or interaction
probability per area, for group $g$ $\left(\frac{\#}{cm^2}\right)$,
\item $\sigma_{s,\gp \rightarrow g}$ is the `inscatter' cross section,
which is the probability per area that a neutron will scatter from
group $\gp$ into group $g$ $\left(\frac{\#}{cm^2}\right)$,
\item $\sigma_{f,g}$ is the fission cross section for group $g$
$\left(\frac{\#}{cm^2}\right)$,
\item $\bar\nu$ is the mean number of neutrons produced per fission
event, and
\item $\chi_{g^\prime\rightarrow g}$ describes the energy spectrum of
  emitted fission neutrons in group $g$ produced by neutrons absorbed
  from group $g^\prime$.
\end{itemize}
Application of the spatial discretization then yields a matrix
equation. For convenience, this equation can be
expressed in operator notation as
\begin{equation}\label{BT op}
\mathbf{L} \psi = \mathbf{M} \mathbf{S} \mathbf{D} \psi + \frac{1}{k}
\mathbf{M} \mathbf{F} \mathbf{D} \psi.
\end{equation}
where $\mathbf{L}$ is the streaming and removal operator, $\mathbf{S}$
is the scattering operator, $\mathbf{F}$ is the fission operator, and
$\mathbf{M}$ and $\mathbf{D}$ are the moment-to-discrete and
discrete--to--moment operators, respectively, and $\mathbf{D}\psi =
\phi$. 

If the fluxes are arranged by energy group, the form of $\mathbf{L}$
is block diagonal and it can be inverted onto a constant
right-hand-side exactly by performing a so-called sweep for each
energy group and angular abscissa. The sweep is an exact application
of $\mathbf{L}^{-1}$ and can be thought of as tracking the movement of
particles through the mesh along a single direction beginning at the
incident boundaries, which vary by angle, so that all necessary
information about neutrons streaming into a cell from other `upwind'
cells is known before it is necessary to make calculations for that
cell.  Most transport codes contain the mechanism to perform the sweep,
but could not apply $\mathbf{L}$ directly without significant
modification, therefore \req{BT op} is generally thought of as a
standard eigenproblem of the form
\begin{equation}\nonumber
k \phi = \left(\mathbf{I}-\mathbf{DL}^{-1} \mathbf{M}
\mathbf{S}\right)^{-1} \mathbf{DL}^{-1}\mathbf{MF} \phi.
\end{equation}
Because we seek the mode with the largest magnitude eigenvalue, a power
iteration is one possible iterative technique:
\begin{subequations}
\nonumber
\begin{eqnarray}
\nonumber
\phi_{z+1} &=& \left(\mathbf{I}-\mathbf{DL}^{-1} \mathbf{M}
\mathbf{S}\right)^{-1} \mathbf{DL}^{-1}\mathbf{M}\frac 1 {k_z}\mathbf{F} \phi_z,\\
\nonumber
k_{z+1} &=& k_z\frac{W^T \mathbf{F}\phi_{z+1}}{W^T \mathbf{F}\phi_z} .
\end{eqnarray}
\end{subequations}
Note that $k_z$ could be removed from both of these equations and the
iteration would be unchanged -- this is generally how power iteration is
presented in the mathematical literature. Our motivation for writing
it this way will become apparent shortly. The update step for the
eigenvalue is physically motivated by the fact that $k$ represents the
time-change in the neutron population, which is either a positive or
negative trend depending on the strength of the only source of
neutrons in the system, i.e. the fission source. The system is made
steady state by adjusting the fission source globally, therefore $W$
is typically taken to be a vector of cell volumes so that each
cell-wise contribution is weighted according to its contribution to
the total. The dot products therefore produce total fission sources,
i.e., volume integrals of the spatially dependent fission sources, and
the ratio between successive fission sources, which are produced by
solving the steady-state equation, is indicative of the direction in
which the neutron population is changing.

Full inversion of $\left(\mathbf{I}-\mathbf{DL}^{-1} \mathbf{M}
\mathbf{S}\right)$ is expensive, however, so it is generally more
efficient to use a nonlinear fixed point iteration (FPI) in which the
operator $\left(\mathbf{I}-\mathbf{DL^{-1}} \mathbf{M}
\mathbf{S}\right)$ is inverted only approximately using a nested
inner-outer iteration scheme to converge the scattering
term. Commonly, one or more `outer iterations' are conducted in which
the fission source is lagged. Each outer iteration consists of a
series of inner `within-group' iterations (one or more per energy
group) in which the inscatter is lagged so that the within-group
scattering can be converged. In general, the iteration can be written
as
\begin{subequations}
\begin{eqnarray}
\phi_{z+1} &=& \mathbf{P}(k_z) \phi_z,\\ 
k_{z+1} &=&
k_z\frac{W^T\mathbf{F}\phi_{z+1}}{W^T\mathbf{F}\phi_{z}} \label{k-eig}.
\end{eqnarray}
\end{subequations}
If $\mathbf{P}(k_z)=\left(\mathbf{I}-\mathbf{DL}^{-1} \mathbf{M}\mathbf{S}\right)^{-1} \mathbf{DL}^{-1}\mathbf{M}\frac{1}{k_z}\mathbf{F}$
we recover a true power iteration, but there are numerous possible
forms for this operator, and it is typically more efficient to choose
a form of $\mathbf{P}(k_z)$ that requires a minimal number of sweeps.

A fixed point problem, which admits the same solution as the power
iteration, but does not require that one invert
$\left(\mathbf{I}-\mathbf{DL}^{-1} \mathbf{M} \mathbf{S}\right)$, and
is not strictly in the form of an eigenvalue problem is
\begin{subequations}\label{k FP}
\begin{eqnarray}
\label{k FPa}
\phi_{z+1} &=& \mathbf{DL}^{-1}
\mathbf{M}\left(\mathbf{S}+\frac{1}{k_z}\mathbf{F}\right) \phi_z,\\
\label{k FPb}
k_{z+1} &=&
\frac{W^T\mathbf{F}\phi_{z+1}}{W^T\left(\frac{1}{k_z}\mathbf{F}\phi_z-\mathbf{S}
(\phi_{z+1}-\phi_z)\right)}.
\end{eqnarray}
\end{subequations}
Here the update for $k_{z+1}$ follows from assuming the difference of
the right hand side of~\req{k FPa} at iterations $z$ and $z+1$ is
zero.  This scheme, referred to as ``flattened'' in~\cite{Gill:2011b},
was shown to be, in most cases, the most efficient formulation
compared to the other formulations that involve additional,
intermediate levels of iteration, and JFNK with this formulation was
shown to be the most efficient solution method compared to fixed-point
iteration, even without employing a preconditioner for the inner GMRES
iterations.

However, if we assume that the scattering is converged, then the
second term in the denominator of \req{k FPb}
is zero and we recover \req{k-eig}. For a converged system, this term
will indeed be zero and the ratio of the fission sources will go to
one because the fission source has been suitably adjusted by
$\frac{1}{k}$ so that the net neutron production is zero.

From this we have that a fixed point of the iteration in \req{k FP}
is also a root of the residual function
\begin{equation}\label{k FP 2}
F\left(
\begin{array}{c}
\phi \\
k
\end{array}
\right) = 
\left(
\begin{array}{c}
\left(\mathbf{I} - \mathbf{P}(k)\right)\phi \\
\left(1 - \frac{W^T\mathbf{F}\mathbf{P}(k)\phi}{W^T\mathbf{F}\phi}\right)k
\end{array}
\right),
\end{equation}
where
\begin{equation*}
\mathbf{P}(k) = 
\mathbf{DL}^{-1}\mathbf{M}\left(
\mathbf{S}+\frac{1}{k}\mathbf{F}\right).
\end{equation*}
In order to initialize the iteration, a single sweep is performed on a
vector of ones, $E=[1\,1\ldots 1]^T$, and the scaled two-norm of the
flux is then normalized to one:
\begin{subequations}
\begin{eqnarray*}
  \phi_0 &=& \frac {\mathbf{DL}^{-1}
  \mathbf{M}\left(\mathbf{S}+\mathbf{F}\right)E}
{\left\|\mathbf{DL}^{-1}
  \mathbf{M}\left(\mathbf{S}+\mathbf{F}\right)E,
\right\|_{2,s}}, \\
  k_0 &=& 1.
\end{eqnarray*}
\end{subequations}
Here $\|\cdot\|_{2,s}$ indicates the $L^2$-norm normalized by the
square root of the number of degrees of freedom. Once the residual is
formulated, it is possible to apply any of the nonlinear solvers
discussed to seek a root of $F$ given in Eq.~\eqref{k FP 2}. The
following section contains a comparison of JFNK, Broyden, NKA$_{-1}$
and NKA for the $k$-eigenvalue problem.

For true power iteration, when the initial guess is not orthogonal to
the first eigenmode, the first eigenmode emerges as the other modes
are reduced at each iteration by successive powers of the ratios of
the corresponding eigenvalues to the dominant eigenvalue.  The
asymptotic rate of convergence of true power iteration is therefore
controlled by the dominance ratio, or the ratio of the second
eigenvalue to the first, therefore it is a common measure for the
difficulty of the problem. We note that the term `power iteration' is
used loosely in the literature, often referring to an iterative
process that is actually a form of FPI. While experience has shown
that FPI, and accelerated variants of FPI, can be trusted to converge
to the dominant eigenmode, there has been no rigorous mathematical
proof that this will always be the case.  Furthermore, the
relationship between the dominance ratio and the asymptotic rates of
convergence is not clear for FPI as it is in the case of true power
iteration, although in our experience problems with larger dominance
ratios are more computationally demanding for FPI than those with
smaller dominance ratios as for power iteration.

\section{Results} \label{sec:results}

Results are given for two test problems. The first is a cylinder with
a 3.5 cm radius and a height of 9 cm modeled in two-dimensional
cylindrical coordinates, similar to the cylindrical problem that was
studied in~\cite{Warsa:2004} in three-dimensional Cartesian
coordinates.  The problem consists of a central 5-cm layer of Boron-10
with 1 cm thick water layers on either side and 1 cm layers of highly
enriched uranium on the ends.  The top, bottom and radial boundaries
are either all vacuum or all reflective (the inner radial boundary is
a symmetry condition in cylindrical coordinates).  A 16-group
Hansen-Roach cross section data set is used to generate the results
presented here (the 16 group data was collapsed to 5 groups in the
original paper~\cite{Warsa:2004}). 

The second problem is the
well-known C5G7-MOX problem~\cite{C5G7}, which was formulated as a
benchmark for deterministic neutron transport codes. It is a
3-dimensional mockup of a 16 assembly nuclear reactor with
quarter-core symmetry surrounded on all sides by a moderator. There
are three variants of the problem that vary in the extent to which the
control rods are inserted. The problem specification includes
seven-group cross sections with upscatter in the lower energy groups
and we use an S$_6$ quadrature set as in~\cite{Gill:2011b}. The geometry
is complicated by the fact that there are cylindrical fuel pins that
must be modeled on an orthogonal mesh in PARTISN, so material
homogenization was necessary on cells that contained fuel pin
boundaries.  PARTISN contributed results for the C5G7-MOX benchmark
in~\cite{C5G7}, so the original input and geometry files were
used to do this study. The PARTISN results and a detailed
description of how the geometry was set up are available
in~\cite{Dahl:2006}. Capsaisin was not used to test this problem.

The three nonlinear solvers that we consider in this paper are
implemented in the NOX nonlinear solver package that is part of the
larger software package Trilinos 10.6~\cite{Trilinos}.  We use JFNK as
it is implemented in NOX and have written implementations of both NKA
and Broyden's method in the NOX framework as user supplied direction
classes. These two direction classes are available for download on
Sourceforge at \url{http://sourceforge.net/projects/nlkain/}. In
addition to the existing JFNK interface, interfaces to the NOX package
were developed for the Broyden and NKA methods so that all methods
except fixed-point iteration are accessed through the same Trilinos
solver.  JFNK can be extremely sensitive to the choice of forcing
parameter, $\eta$, therefore we explored variations of the three
possibilities implemented in NOX:
\begin{subequations}
\begin{align}
& \mbox{Constant:} & \eta_z &= \eta_0 \\ & \mbox{Type 1 (EW1):} &
  \eta_z &=
  \left|\frac{\|F_z\|-\|J_{z-1}\delta_{z-1}+F_{z-1}\|}{\|F_{z-1}\|}\right|.
  &&\text{If}\quad\eta_{z-1}^{\frac{1+\sqrt{5}}{2}} >
  .1,\quad\text{then}\quad\eta_z \leftarrow
  \mbox{max}\left\{\eta_{z},\,\eta_{z-1}^{\frac{1+\sqrt{5}}{2}}\right\}.\\ &
  \mbox{Type 2 (EW2):} & \eta_z &=
  \gamma\left(\frac{\|F_z\|}{\|F_{z-1}\|}\right)^\alpha.
  &&\text{If}\quad \gamma\eta_{z-1}^\alpha\quad > .1,
  \quad\text{then}\quad \eta_z \leftarrow \mbox{max}\{\eta_z ,
  \gamma\eta_{z-1}^\alpha\}.
\end{align}
\end{subequations}
Types 1 and 2 were developed by Eisenstat and
Walker~\cite{Eisenstat:1996}, therefore they are denoted EW1 and EW2
in the results that follow. The NOX default parameters were also used
for EW1 and EW2: $\eta_0 = 10^{-1}$, $\eta_{min} = 10^{-6}$,
$\eta_{max} = 10^{-2}$, $\alpha = 1.5$ and $\gamma = 0.9$.  

Note that for the results below, the convergence criterion is $\|F\|_{2,s}\leq
10^{-9}$ for the reflected cylinder and $\|F\|_{2,s}\leq 10^{-8}$ for the C5G7-MOX problem.

\subsection{2-D cylinder}\label{cyl test}
A 175 (r-axis) by 450 (z-axis) mesh of equally sized squares is used
in PARTISN for both the reflected and unreflected problems.  The
Capsaicin results are computed on an unstructured mesh comprising
roughly the same number of cells in the $r$ and $z$ axes as the
PARTISN computation, for a total of 79855 (possibly non-convex)
polygons, each of which has 3 to 6 vertexes on a cell.  These codes
use similar methods, but there are several differences that lead to
different iterative behavior and slightly different results.  First,
PARTISN is an orthogonal mesh code that utilizes a diamond (central)
difference (DD) spatial discretization, requiring the storage of only
one cell-centered value per energy group~\cite{LewisMiller:1993}.
Capsaicin uses an unstructured mesh, which has the advantage of being
able to model geometrically complex problems with higher fidelity, but
it comes at a cost.  Finding an efficient sweep schedule on an
unstructured mesh that minimizes the latency in parallel computations
is a difficult problem in itself \cite{Hendrickson:2005, Pautz:2002,
  Kumar:2006}.  In contrast, a parallel sweep schedule that is nearly
optimal is easy to implement for structured meshes \cite{Baker:1998}.
Furthermore, a discontinuous finite element method (DFEM) spatial
discretization is employed in Capsaicin such that the number of
unknowns per cell is equal to the number of nodes~\cite{Warsa:2008}.
While the DFEM has better accuracy than DD, potentially enabling the
use of commensurately fewer mesh cells for the same solution accuracy,
it is more costly to compute the solution to the DFEM equations than
the DD equations, and the DFEM is thus less efficient than DD when
calculations are performed on meshes of similar size.

The second difference between the two codes is that reflecting
boundary conditions can cause instabilities with the DD spatial
discretization and the angular differencing used in PARTISN so, in
order to achieve stability, the reflected flux is `relaxed'. This is
done using a linear combination of the new and old reflected fluxes as
the boundary source for the next iteration:
\[
\psi_{relaxed}^{z+1} = r\psi_{refl}^z+(1-r)\psi_{refl}^{z+1}.
\] 
The relaxation parameter, $r$, is, by default, $\frac12$ for this
problem.  This relaxation is not necessary in Capsaicin because it
uses a more accurate spatial and angular discretization, including the
use of starting directions for calculations in cylindrical
coordinates~\cite{LewisMiller:1993}, as well the use of reflected
starting directions when reflection conditions are specified on the
radial boundary.  The eigenvalues computed by PARTISN are
$k=0.1923165$ and $k=0.8144675$ for the unreflected and reflected
cases, respectively.  Because the codes employ different
discretizations, the eigenvalues calculated by Capsaicin were
$k=0.191714$ for the unreflected case and $k=0.814461$ for the
reflected case. The Capsaicin discretization results in roughly $8$
million degrees of freedom while the PARTISN discretization results in
a little over $1.2$ million degrees of freedom. 


\subsubsection{Unreflected cylinder}

\paragraph{PARTISN results\label{section:PART-no_refl}}

Tables~\ref{table:JFNK-norefl}-\ref{table:perc-norefl} show the number
of JFNK and total inner GMRES iterations, the number of sweeps, the
CPU time and the percentage of that time that was spent computing the
residual, respectively.  As can been seen in
Tables~\ref{table:JFNK-norefl} and \ref{table:sweep-norefl}, the GMRES
subspace size does not affect the number of Newton iterations for the
most part and has little effect on the total number of GMRES
iterations or sweeps.  As the forcing parameter decreases, however,
more sweeps are required because more GMRES iterations are required,
both overall and per Newton step. Smaller subspace sizes require more
restarts, which in turn each require one additional sweep. In general,
for NKA, NKA$_{-1}$ and Broyden,
decreasing the subspace size leads to an increase in sweep count, but
the effect is negligible for NKA in this case, while it is significant
for NKA$_{-1}$ and Broyden. Broyden(10) also requires manytotal  more
iterations than Broyden(5) in contradiction to the general
trend. NKA$_{-1}$ also failed to converge for a subspace size of
5. Overall, NKA requires the fewest sweeps and displays a predictable
trend of decreasing sweep count with increasing subspace size. As
Table~\ref{table:time-norefl} shows, the runtimes are consistent with
the sweep count. NKA is more than 3 times faster than FPI and between
1.6 and 2.3 times as fast as JFNK, while NKA$_{-1}$ is, at best,
comparable to JFNK and, at worst, more than an order of magnitude
slower than FPI. The behavior of Broyden is rather chaotic\textemdash
at best, it is almost as efficient as NKA, but at worst it too is
slower than FPI.

The percentage of the total time spent evaluating the residual is
shown in Table~\ref{table:perc-norefl} because, for this particular
code and problem, the solver time requires a non-trivial portion of
the run-time. As can be seen, this percentage is slightly smaller for
NKA and NKA$_{-1}$ than Broyden and noticeably smaller for Broyden than for
JFNK.  There is also a slight decrease as the subspace size increases
in all cases, which is expected since the respective solvers must do
more work for a larger subspace than a smaller one.  JFNK requires the
least amount of time in the solver with approximately 95\% of the
run-time spent in the sweep regardless of GMRES subspace size. Despite
these numbers, we note that NKA still requires the least run-time of
all of the methods.

The plots in Fig.~\ref{fig:no_refl} show the scaled $L^2$-norm of
the residual as a function of the number of sweeps. For JFNK where
there are a few outer Newton iterations with multiple sweeps required
to do the inner inversion of the Jacobian, the $L^2$-norm is given at
each Newton step plotted at the cumulative sweep count up to that
point. Fig.~\ref{fig:no_refl_nka} shows that the behavior of NKA is
similar for all subspace sizes and much more efficient than
FPI. Fig.~\ref{fig:no_refl_AM} shows NKA$_{-1}$, which converges
fairly quickly for the larger restarts, but then has trouble reducing
the residual for subspaces of 5 and 10. While 10 eventually converges,
5 seems to stagnate. Fig.~\ref{fig:no_refl_broy} shows the behavior
of Broyden. As can be seen, the norm fluctuates quite erratically for
every subspace size except 30 and cannot compete with NKA with a
subspace of 5. Fig.~\ref{fig:no_refl_jfnk} shows some of the more
efficient JFNK results plotted at each Newton iteration for the
current sweep count. And finally, Fig.~\ref{fig:no_refl_comp} shows
a comparison of the most efficient results for each method.

\begin{table}[h!]
\centering

\caption{PARTISN unreflected cylinder: Number of outer and inner
JFNK iterations, sweeps, run-time and percentage of the run-time spent
computing the residual to an accuracy of $\|F\|_{2,s}\leq 10^{-9}$ for
the various methods.}
\subfloat[Outer JFNK/total inner GMRES iterations]
{
\begin{tabular}{cccccc}  \toprule[1pt]
  & \multicolumn{5}{c}{$\eta$} \\ \cline{2-6}
\raisebox{1.5ex}[0cm][0cm]{subspace}
        &	0.1		&	0.01		&	0.001		&	EW1		&	EW2		\\\hline
30	&	8	(30)	&	6	(42)	&	5	(45)	&	6	(39)	&	5	(38)	\\
20	&	8	(30)	&	6	(42)	&	5	(45)	&	6	(39)	&	5	(38)	\\
10	&	8	(30)	&	6	(42)	&	5	(48)	&	6	(41)	&	5	(42)	\\
5	&	8	(30)	&	6	(48)	&	5	(50)	&	5	(43)	&	5	(44)	\\\bottomrule[1pt]
\label{table:JFNK-norefl}
\end{tabular}}

\subfloat[Number of sweeps (FPI converged in 99)]
{
\begin{tabular}{ccccccccc}  \toprule[1pt]
	&		&		&		&		&		& JFNK $\eta$  &		&		\\\cline{5-9}
\raisebox{1.5ex}[0cm][0cm]{subspace}
        &	\raisebox{1.5ex}[0cm][0cm]{NKA}
                        &	\raisebox{1.5ex}[0cm][0cm]{NKA$_{-1}$}
                                                        &     \raisebox{1.5ex}[0cm][0cm]{Broyden}
                                                      	&	0.1	&	0.01	&	0.001	&	EW1	&	EW2	\\\hline
30	&	26	&	43	&	31	&	47	&	55	&	56	&	57	&	49	\\
20	&	26	&	69	&	43	&	47	&	55	&	56	&	57	&	49	\\
10	&	27	&	946	&	123	&	47	&	55	&	60	&	60	&	55	\\
5	&	28	&	--	&	70	&	47	&	66	&	68	&	64	&	61	\\\bottomrule[1pt]

\label{table:sweep-norefl}
\end{tabular}}

\subfloat[CPU time (s) (FPI converged in 245.8 s)]
{
\begin{tabular}{ccccccccc}  \toprule[1pt]
	&		&		&		&		&		& JFNK $\eta$  &		&		\\\cline{5-9}
\raisebox{1.5ex}[0cm][0cm]{subspace}
        &	\raisebox{1.5ex}[0cm][0cm]{NKA}
                        &	\raisebox{1.5ex}[0cm][0cm]{NKA$_{-1}$}
                                                        &     \raisebox{1.5ex}[0cm][0cm]{Broyden}
                                                      	&	0.1	&	0.01	&	0.001	&	EW1	&	EW2	\\\hline
30	&	74.36	&	129.10	&	85.37	&	121.09	&	142.05	&	147.25	&	147.48	&	127.32	\\
20	&	74.36	&	203.84	&	118.08	&	120.32	&	141.85	&	147.35	&	147.10	&	127.94	\\
10	&	75.29	&	2657.78	&	328.29	&	121.40	&	141.84	&	156.61	&	154.53	&	141.82	\\
5	&	75.81	&	--	&	182.20	&	121.03	&	169.41	&	174.05	&	164.40	&	158.09	\\\bottomrule[1pt]
\label{table:time-norefl}
\end{tabular}} 

\subfloat[Percentage of CPU time spent in the residual evaluation]
{
\begin{tabular}{ccccccccc}  \toprule[1pt]
	&		&		&		&		&		& JFNK $\eta$  &		&		\\\cline{5-9}
\raisebox{1.5ex}[0cm][0cm]{subspace}
        &	\raisebox{1.5ex}[0cm][0cm]{NKA}
                        &	\raisebox{1.5ex}[0cm][0cm]{NKA$_{-1}$}
                                                        &     \raisebox{1.5ex}[0cm][0cm]{Broyden}
                                                      	&	0.1	&	0.01	&	0.001	&	EW1	&	EW2	\\\hline
30	&	85.79	&	81.56	&	89.67	&	95.70	&	95.02	&	94.75	&	94.93	&	94.65	\\
20	&	86.30	&	82.78	&	89.47	&	95.71	&	95.03	&	94.76	&	94.93	&	94.66	\\
10	&	89.05	&	87.39	&	92.05	&	95.74	&	95.02	&	94.88	&	95.20	&	95.12	\\
5	&	91.77	&	--	&	94.36	&	95.74	&	95.70	&	95.75	&	95.80	&	95.79	\\\bottomrule[1pt]
\label{table:perc-norefl}
\end{tabular}} 
\end{table}

\begin{figure}[h!]
\centering
\caption{PARTISN unreflected cylinder: Scaled $L^2$-norm of the residual
as a function of number of sweeps for the various methods and subspace
sizes. Each of the methods is plotted on the same scale to simplify 
comparisons between panels. Note that, in panel (d),  points
plotted on the lines indicate when a JFNK iteration starts (JFNK requires
multiple sweeps per iteration). In panels (a), (b), and (c) we plot one point
per two iterations of the method. In panel (e) the convention for iterations
per plotted points is the same as in panels (a) through (d).}  
\begin{tabular}{cc}
\subfloat[NKA and FPI]{
\resizebox{80mm}{!}{
\begingroup%
\makeatletter%
\newcommand{\GNUPLOTspecial}{%
  \@sanitize\catcode`\%=14\relax\special}%
\setlength{\unitlength}{0.0500bp}%
\begin{picture}(7200,5040)(0,0)%
  \put(5936,4036){\makebox(0,0)[r]{\strut{}NKA(20)}}%
  \put(5936,4236){\makebox(0,0)[r]{\strut{}NKA(10)}}%
  \put(5936,4436){\makebox(0,0)[r]{\strut{}NKA(5)}}%
  \put(5936,4636){\makebox(0,0)[r]{\strut{}FPI}}%
  \put(4089,140){\makebox(0,0){\strut{}Sweeps}}%
  \put(160,2719){%
\rotatebox{-270}{%
  \makebox(0,0){\strut{}$\| {F} \|_{2,s}$}%
}}%
  \put(6619,440){\makebox(0,0){\strut{} 120}}%
  \put(5739,440){\makebox(0,0){\strut{} 100}}%
  \put(4859,440){\makebox(0,0){\strut{} 80}}%
  \put(3980,440){\makebox(0,0){\strut{} 60}}%
  \put(3100,440){\makebox(0,0){\strut{} 40}}%
  \put(2220,440){\makebox(0,0){\strut{} 20}}%
  \put(1340,440){\makebox(0,0){\strut{} 0}}%
  \put(1220,4799){\makebox(0,0)[r]{\strut{} 1}}%
  \put(1220,3967){\makebox(0,0)[r]{\strut{} 0.01}}%
  \put(1220,3135){\makebox(0,0)[r]{\strut{} 0.0001}}%
  \put(1220,2304){\makebox(0,0)[r]{\strut{} 1e-06}}%
  \put(1220,1472){\makebox(0,0)[r]{\strut{} 1e-08}}%
  \put(1220,640){\makebox(0,0)[r]{\strut{} 1e-10}}%
\includegraphics{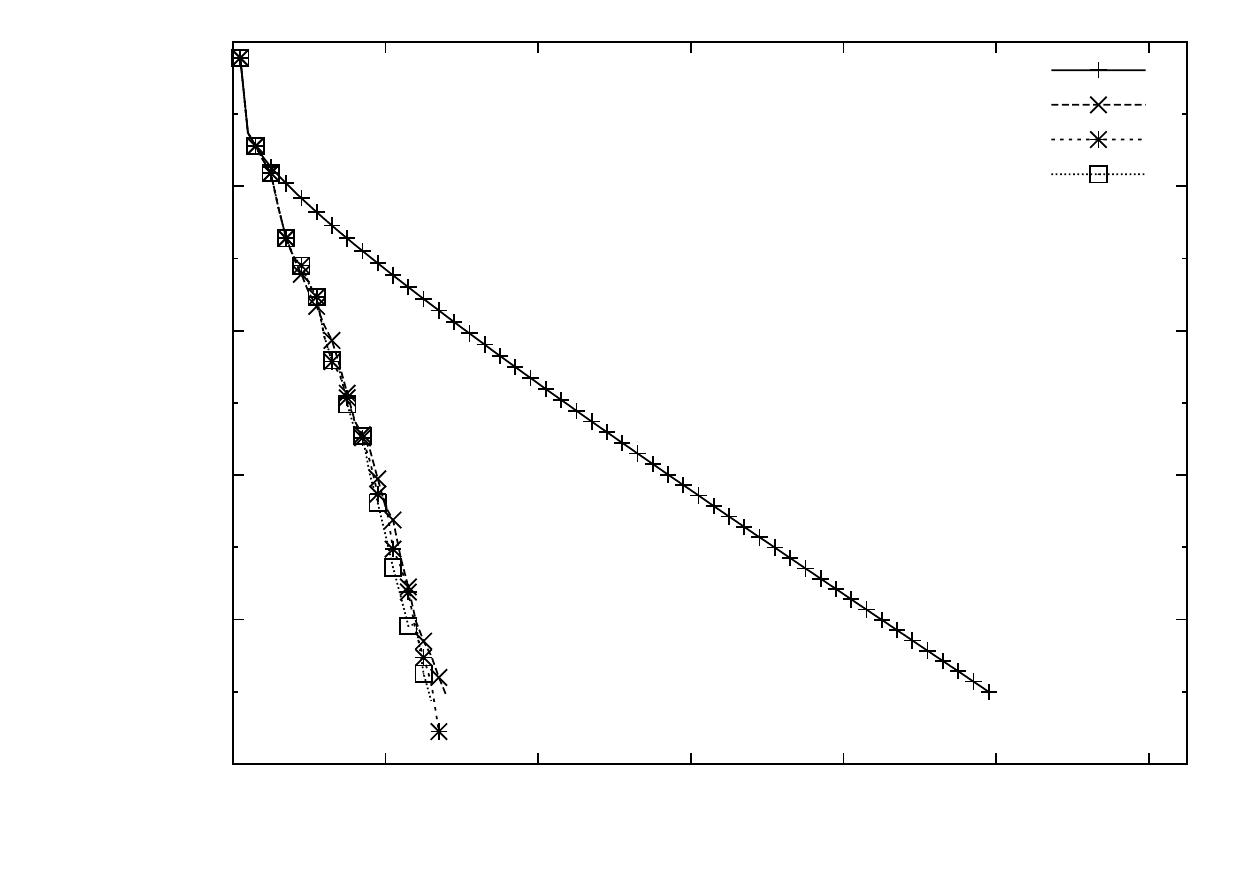}%
\end{picture}%
\endgroup

}
 \label{fig:no_refl_nka}}
&
\subfloat[NKA$_{-1}$]{
 \resizebox{80mm}{!}{
\begingroup%
\makeatletter%
\newcommand{\GNUPLOTspecial}{%
  \@sanitize\catcode`\%=14\relax\special}%
\setlength{\unitlength}{0.0500bp}%
\begin{picture}(7200,5040)(0,0)%
  \put(5936,4036){\makebox(0,0)[r]{\strut{}NKA$_{-1}$(30)}}%
  \put(5936,4236){\makebox(0,0)[r]{\strut{}NKA$_{-1}$(20)}}%
  \put(5936,4436){\makebox(0,0)[r]{\strut{}NKA$_{-1}$(10)}}%
  \put(5936,4636){\makebox(0,0)[r]{\strut{}NKA$_{-1}$(5)}}%
  \put(4089,140){\makebox(0,0){\strut{}Sweeps}}%
  \put(160,2719){%
\rotatebox{-270}{%
  \makebox(0,0){\strut{}$\| {F} \|_{2,s}$}%
}}%
  \put(6619,440){\makebox(0,0){\strut{} 120}}%
  \put(5739,440){\makebox(0,0){\strut{} 100}}%
  \put(4859,440){\makebox(0,0){\strut{} 80}}%
  \put(3980,440){\makebox(0,0){\strut{} 60}}%
  \put(3100,440){\makebox(0,0){\strut{} 40}}%
  \put(2220,440){\makebox(0,0){\strut{} 20}}%
  \put(1340,440){\makebox(0,0){\strut{} 0}}%
  \put(1220,4799){\makebox(0,0)[r]{\strut{} 1}}%
  \put(1220,3967){\makebox(0,0)[r]{\strut{} 0.01}}%
  \put(1220,3135){\makebox(0,0)[r]{\strut{} 0.0001}}%
  \put(1220,2304){\makebox(0,0)[r]{\strut{} 1e-06}}%
  \put(1220,1472){\makebox(0,0)[r]{\strut{} 1e-08}}%
  \put(1220,640){\makebox(0,0)[r]{\strut{} 1e-10}}%
\includegraphics{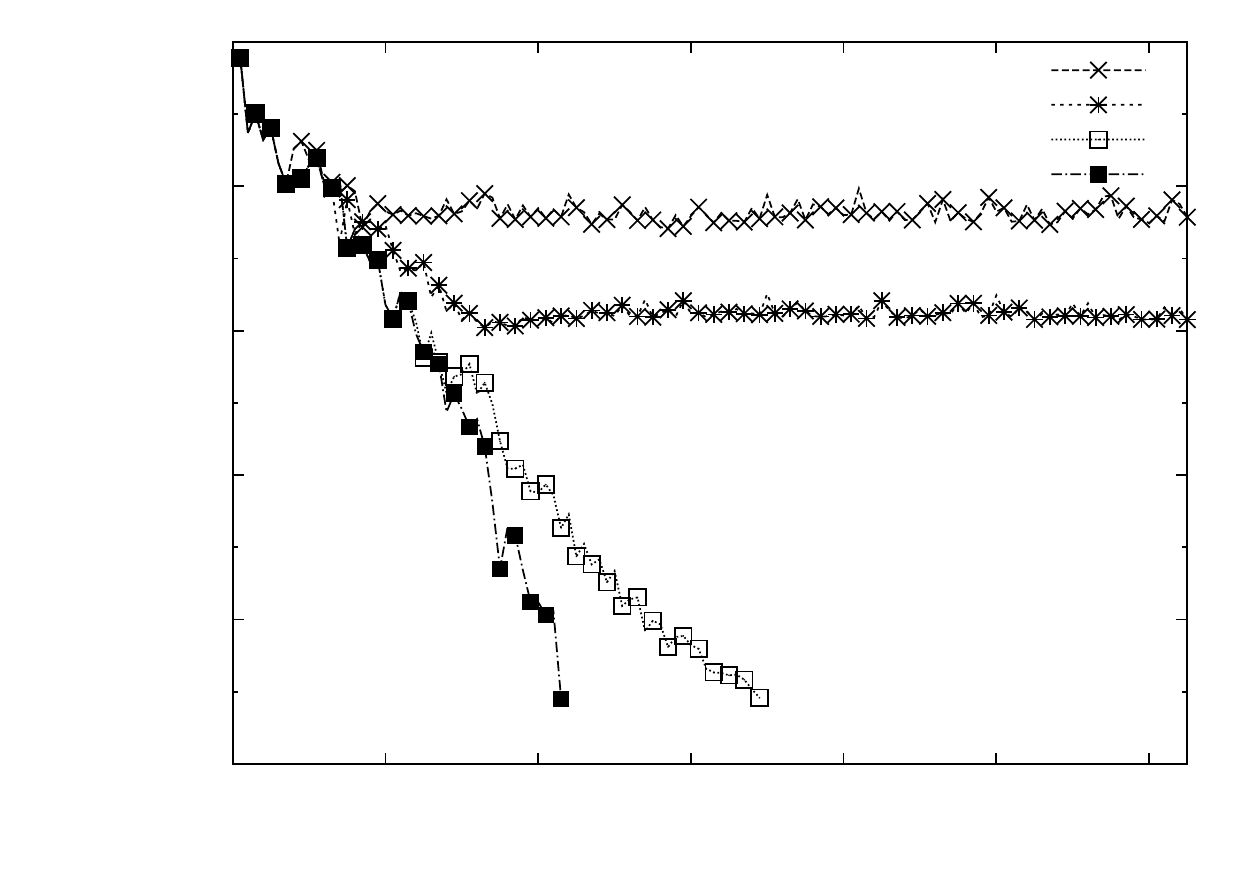}%
\end{picture}%
\endgroup

}
 \label{fig:no_refl_AM}}
\\
\subfloat[Broyden]{
 \resizebox{80mm}{!}{
\begingroup%
\makeatletter%
\newcommand{\GNUPLOTspecial}{%
  \@sanitize\catcode`\%=14\relax\special}%
\setlength{\unitlength}{0.0500bp}%
\begin{picture}(7200,5040)(0,0)%
  \put(5936,4036){\makebox(0,0)[r]{\strut{}Broyden(30)}}%
  \put(5936,4236){\makebox(0,0)[r]{\strut{}Broyden(20)}}%
  \put(5936,4436){\makebox(0,0)[r]{\strut{}Broyden(10)}}%
  \put(5936,4636){\makebox(0,0)[r]{\strut{}Broyden(5)}}%
  \put(4089,140){\makebox(0,0){\strut{}Sweeps}}%
  \put(160,2719){%
\rotatebox{-270}{%
  \makebox(0,0){\strut{}$\| {F} \|_{2,s}$}%
}}%
  \put(6619,440){\makebox(0,0){\strut{} 120}}%
  \put(5739,440){\makebox(0,0){\strut{} 100}}%
  \put(4859,440){\makebox(0,0){\strut{} 80}}%
  \put(3980,440){\makebox(0,0){\strut{} 60}}%
  \put(3100,440){\makebox(0,0){\strut{} 40}}%
  \put(2220,440){\makebox(0,0){\strut{} 20}}%
  \put(1340,440){\makebox(0,0){\strut{} 0}}%
  \put(1220,4799){\makebox(0,0)[r]{\strut{} 1}}%
  \put(1220,3967){\makebox(0,0)[r]{\strut{} 0.01}}%
  \put(1220,3135){\makebox(0,0)[r]{\strut{} 0.0001}}%
  \put(1220,2304){\makebox(0,0)[r]{\strut{} 1e-06}}%
  \put(1220,1472){\makebox(0,0)[r]{\strut{} 1e-08}}%
  \put(1220,640){\makebox(0,0)[r]{\strut{} 1e-10}}%
\includegraphics{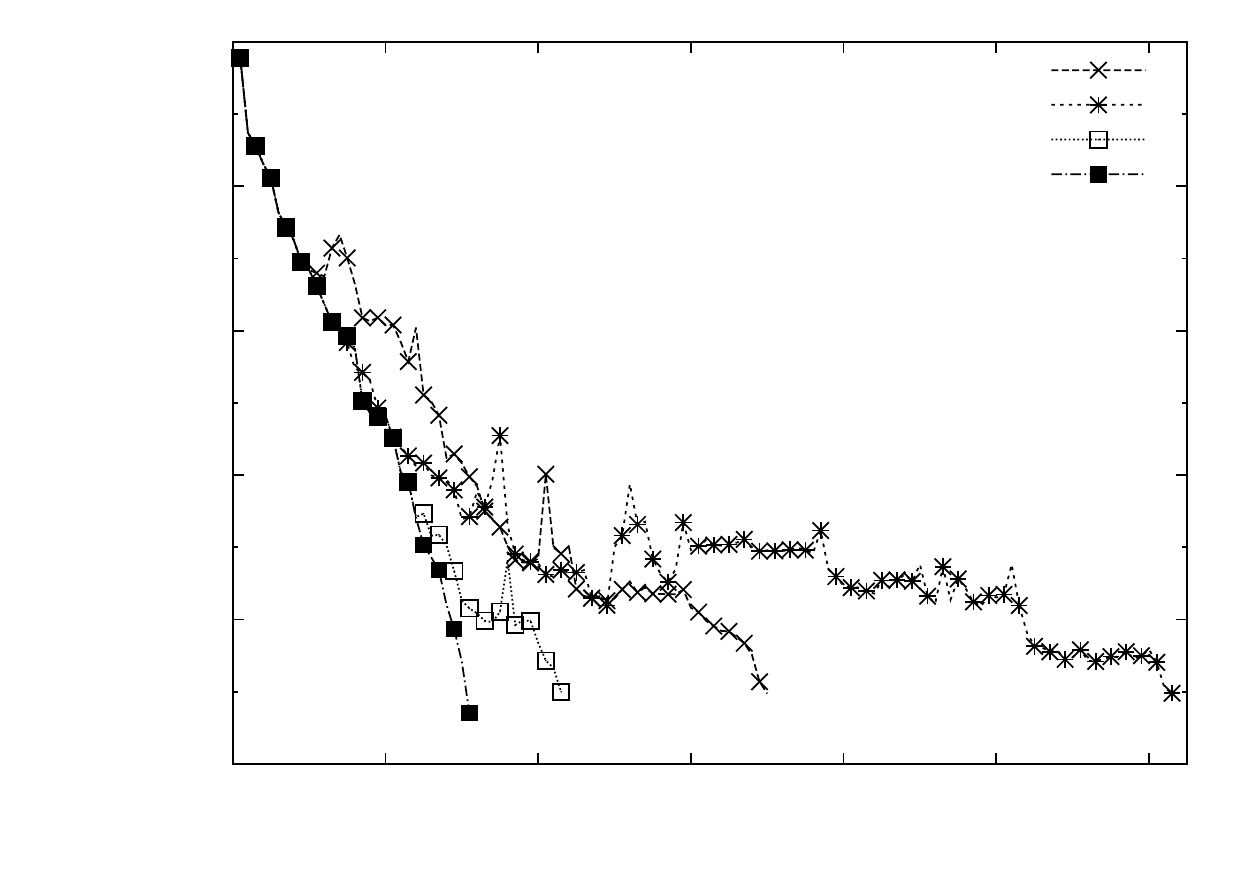}%
\end{picture}%
\endgroup

}
 \label{fig:no_refl_broy}}
&
\subfloat[JFNK]{
 \resizebox{80mm}{!}{
\begingroup%
\makeatletter%
\newcommand{\GNUPLOTspecial}{%
  \@sanitize\catcode`\%=14\relax\special}%
\setlength{\unitlength}{0.0500bp}%
\begin{picture}(7200,5040)(0,0)%
  \put(5936,4036){\makebox(0,0)[r]{\strut{}GMRES(10), $\eta_{EW2}$}}%
  \put(5936,4236){\makebox(0,0)[r]{\strut{}GMRES(5), $\eta_{EW2}$}}%
  \put(5936,4436){\makebox(0,0)[r]{\strut{}GMRES(10), $\eta=0.1$}}%
  \put(5936,4636){\makebox(0,0)[r]{\strut{}GMRES(5), $\eta=0.1$}}%
  \put(4089,140){\makebox(0,0){\strut{}Sweeps}}%
  \put(160,2719){%
\rotatebox{-270}{%
  \makebox(0,0){\strut{}$\| {F} \|_{2,s}$}%
}}%
  \put(6619,440){\makebox(0,0){\strut{} 120}}%
  \put(5739,440){\makebox(0,0){\strut{} 100}}%
  \put(4859,440){\makebox(0,0){\strut{} 80}}%
  \put(3980,440){\makebox(0,0){\strut{} 60}}%
  \put(3100,440){\makebox(0,0){\strut{} 40}}%
  \put(2220,440){\makebox(0,0){\strut{} 20}}%
  \put(1340,440){\makebox(0,0){\strut{} 0}}%
  \put(1220,4799){\makebox(0,0)[r]{\strut{} 1}}%
  \put(1220,3967){\makebox(0,0)[r]{\strut{} 0.01}}%
  \put(1220,3135){\makebox(0,0)[r]{\strut{} 0.0001}}%
  \put(1220,2304){\makebox(0,0)[r]{\strut{} 1e-06}}%
  \put(1220,1472){\makebox(0,0)[r]{\strut{} 1e-08}}%
  \put(1220,640){\makebox(0,0)[r]{\strut{} 1e-10}}%
\includegraphics{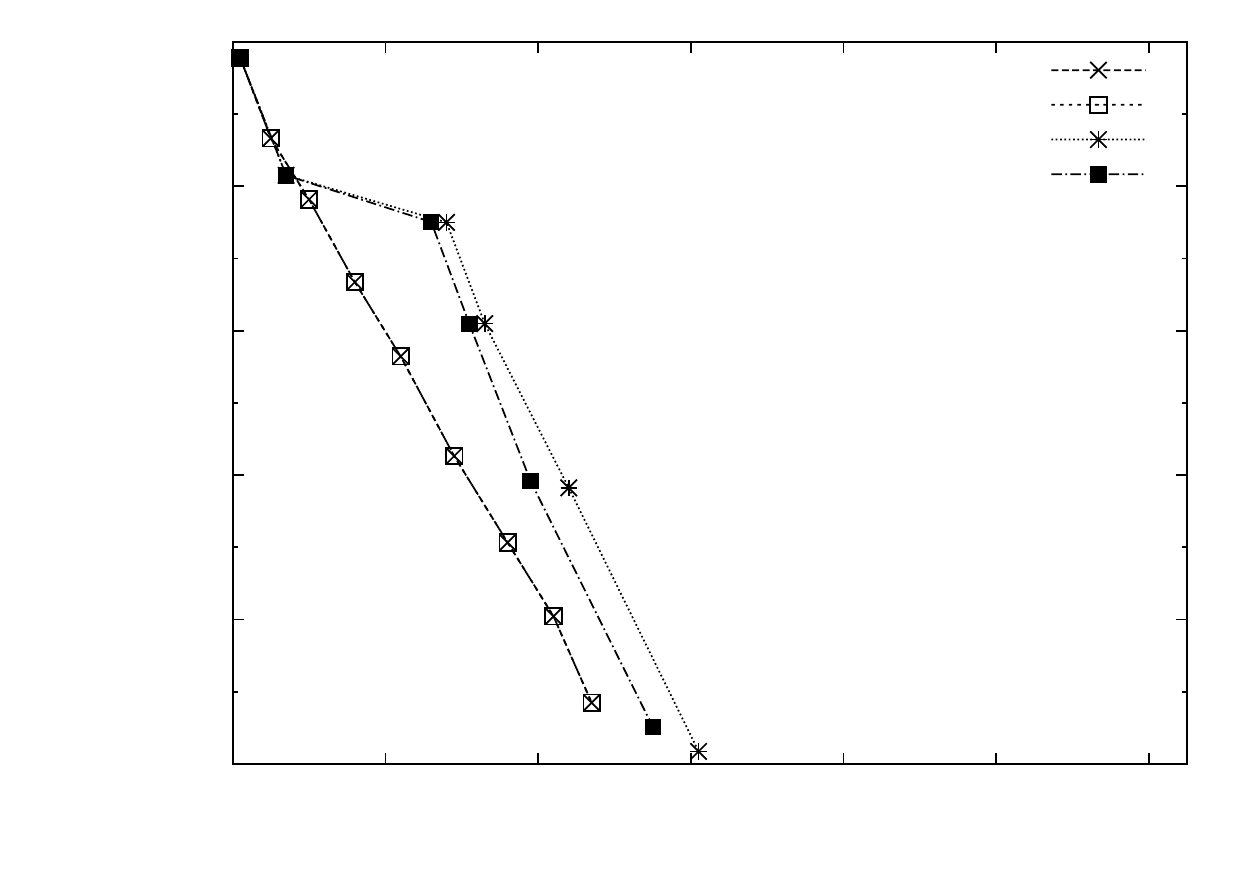}%
\end{picture}%
\endgroup

}
 \label{fig:no_refl_jfnk}}
\\
\multicolumn{2}{c}{
\subfloat[Comparison of the best results for each method]{
 \resizebox{80mm}{!}{
\begingroup%
\makeatletter%
\newcommand{\GNUPLOTspecial}{%
  \@sanitize\catcode`\%=14\relax\special}%
\setlength{\unitlength}{0.0500bp}%
\begin{picture}(7200,5040)(0,0)%
  \put(5936,3836){\makebox(0,0)[r]{\strut{}JFNK (GMRES(5), $\eta=0.1$)}}%
  \put(5936,4036){\makebox(0,0)[r]{\strut{}Broyden(30)}}%
  \put(5936,4236){\makebox(0,0)[r]{\strut{}NKA$_{-1}$(30)}}%
  \put(5936,4436){\makebox(0,0)[r]{\strut{}NKA(20)}}%
  \put(5936,4636){\makebox(0,0)[r]{\strut{}FPI}}%
  \put(4089,140){\makebox(0,0){\strut{}Sweeps}}%
  \put(160,2719){%
\rotatebox{-270}{%
  \makebox(0,0){\strut{}$\| {F} \|_{2,s}$}%
}}%
  \put(6619,440){\makebox(0,0){\strut{} 120}}%
  \put(5739,440){\makebox(0,0){\strut{} 100}}%
  \put(4859,440){\makebox(0,0){\strut{} 80}}%
  \put(3980,440){\makebox(0,0){\strut{} 60}}%
  \put(3100,440){\makebox(0,0){\strut{} 40}}%
  \put(2220,440){\makebox(0,0){\strut{} 20}}%
  \put(1340,440){\makebox(0,0){\strut{} 0}}%
  \put(1220,4799){\makebox(0,0)[r]{\strut{} 1}}%
  \put(1220,3967){\makebox(0,0)[r]{\strut{} 0.01}}%
  \put(1220,3135){\makebox(0,0)[r]{\strut{} 0.0001}}%
  \put(1220,2304){\makebox(0,0)[r]{\strut{} 1e-06}}%
  \put(1220,1472){\makebox(0,0)[r]{\strut{} 1e-08}}%
  \put(1220,640){\makebox(0,0)[r]{\strut{} 1e-10}}%
\includegraphics{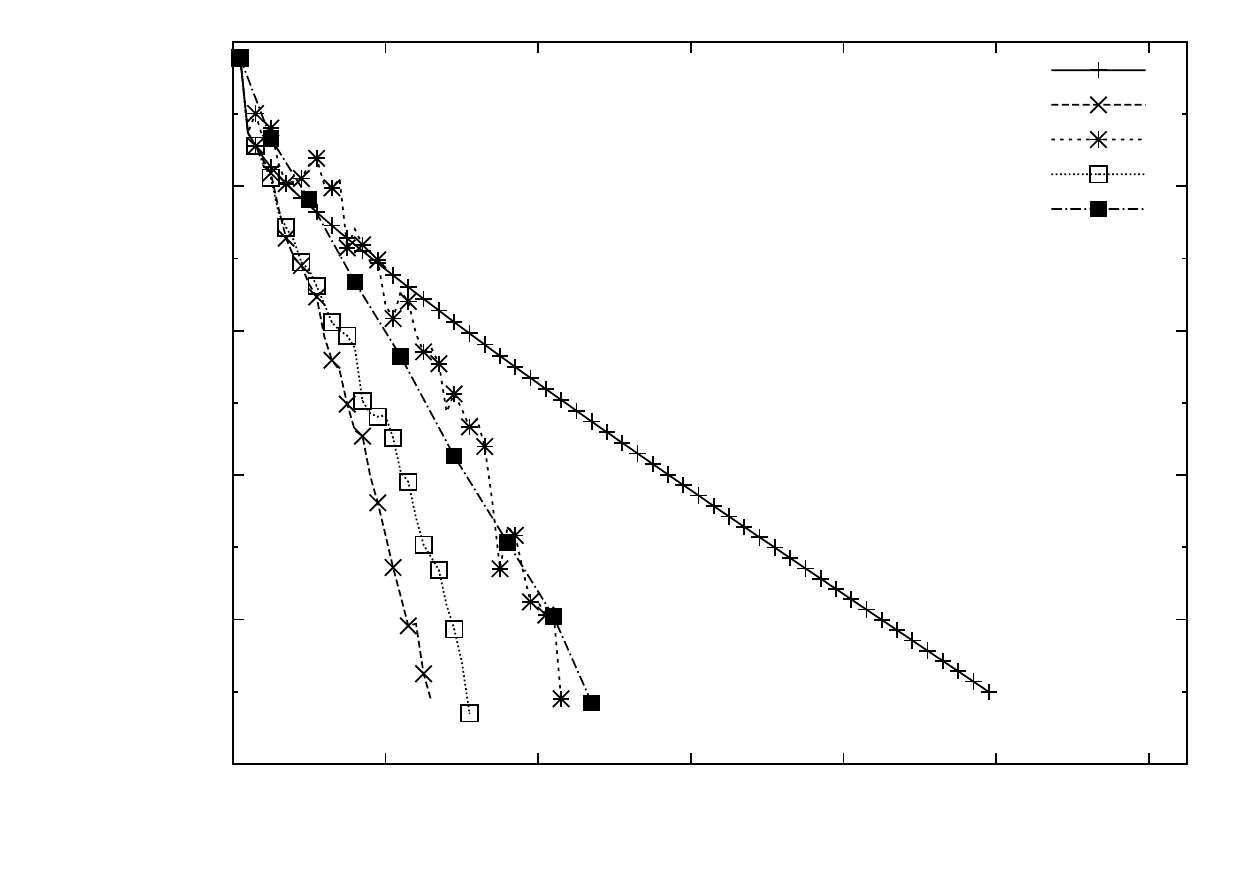}%
\end{picture}%
\endgroup

}
 \label{fig:no_refl_comp}}}
\\
\end{tabular}
\label{fig:no_refl}
\end{figure}
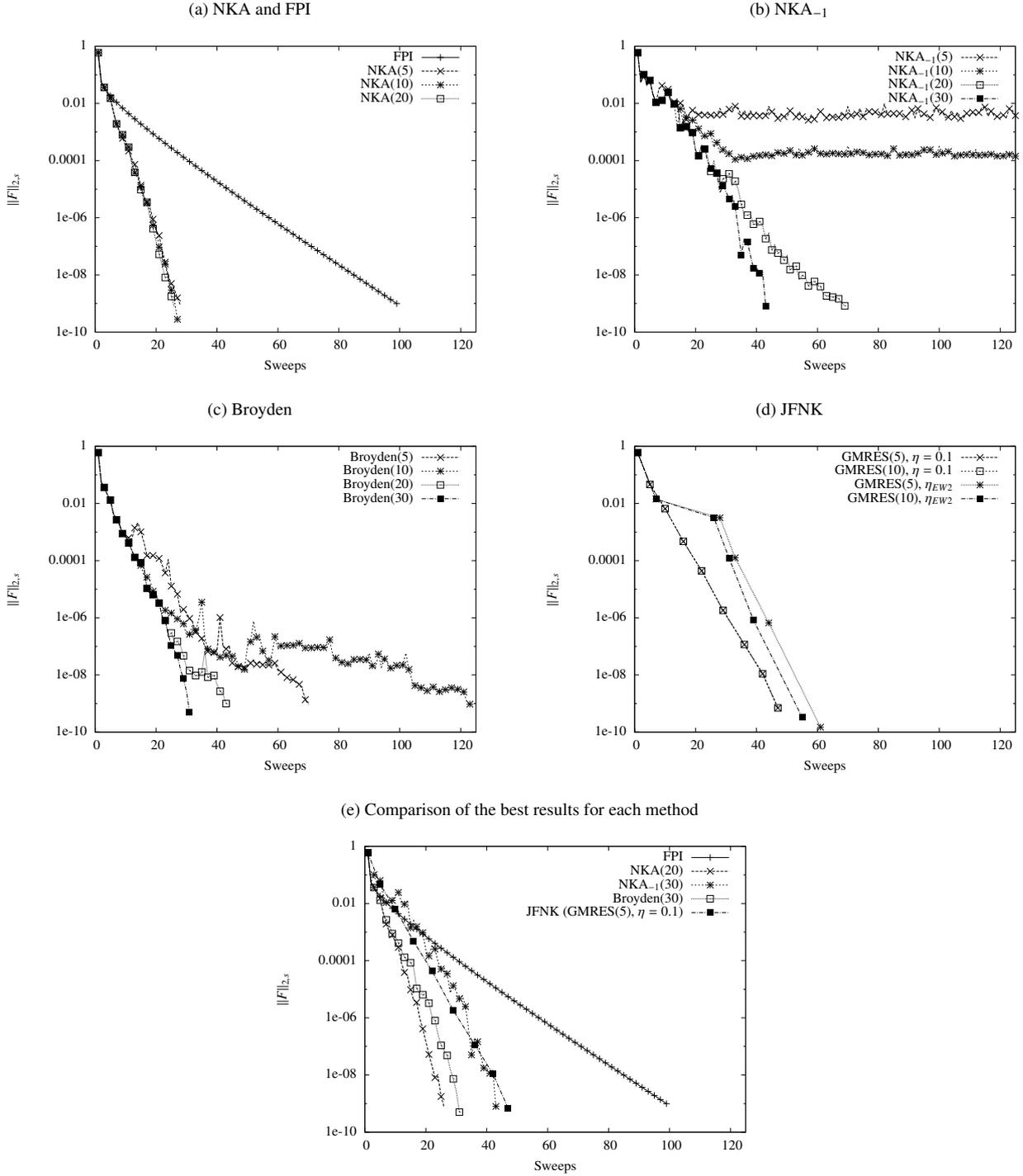

\clearpage
\paragraph{Capsaicin results\label{section:cap-no_refl}}

The different spatial and angular discretization methods used in
Capsaicin alter the numerical properties of the operators associated
with the $k$-eigenvalue problem, affecting the convergence rates and
curves of the various solution methods.  While the methods in
Capsaicin have better accuracy than those in PARTISN, they are also
more costly due to the additional degrees of freedom associated with
the DFEM spatial discretization and the use of starting directions in
the angular discretization.  This is evident in the relatively long
run times associated with the Capsaicin calculations, even though they
were computed with a parallel decomposition in energy, in addition to
the parallel mesh decomposition, on 24 processors.  Because of other
differences in implementation between PARTISN and Capsaicin, a minimum
of 99\% of the computation time is spent in the evaluation of the
residual (for all solution methods) and therefore is not reported in
the tables.

Tables~\ref{table:Cap-JFNK-norefl}-\ref{table:Cap-time-norefl} show
the number of JFNK and total inner GMRES iterations, the number of
sweeps and the CPU time that was spent computing the residual,
respectively.  As can been seen in Tables~\ref{table:Cap-JFNK-norefl}
and \ref{table:Cap-sweep-norefl}, as for PARTISN the GMRES subspace
size does not affect the number of Newton iterations for the most part
and as the forcing parameter decreases more sweeps are
required. However, the subspace size does affect the total number of
GMRES iterations, hence the sweep count, for the smaller and varying
forcing parameters as was not seen with PARTISN.  Once again, for NKA,
NKA$_{-1}$ and Broyden, decreasing the subspace size leads to an
overall increase in sweep count and Broyden(10) requires many more
iterations than Broyden(5) in contradiction to the general
trend. NKA$_{-1}$ failed to converge for subspace sizes of 5 and 10
where NKA$_{-1}$ in PARTISN only failed for a subspace size of 5. Once again, NKA requires the fewest
sweeps and displays a predictable trend of decreasing sweep count with
increasing subspace size. Table~\ref{table:Cap-time-norefl}
demonstates that once again the runtimes are consistent with the sweep
count. NKA is more than 6 times faster than FPI and between 1.8 and
3.6 times as fast as JFNK for comparable subspace sizes, while
NKA$_{-1}$ is comparable to JFNK \textit{when it converges}. The
behavior of Broyden is, once again, chaotic, at times almost as
efficient as NKA, but at worst slower than JFNK, although in this case
it is always more efficient than FPI.

The plots in Fig.~\ref{fig:no_refl_cap} show the scaled $L^2$-norm of
the residual as a function of the number of sweeps. The behavior is
similar to that seen in Fig.~\ref{fig:no_refl}, FPI experiences an
abrupt change in slope between $\|F\|_{2,s}=10^{-8}$ and $10^{-9}$
that is not seen in the PARTISN results and leads to a much larger
iteration count. It is also interesting to note that in the NKA$_{-1}$
results, shown in Fig.~\ref{fig:no_refl_And_cap}, subspaces of 5 and
10 seem to stagnate for large numbers of iterations instead of
consistently dropping to zero as they do for NKA in
Fig.~\ref{fig:no_refl_nka_cap}.

\begin{table}[h!]
\centering

\caption{Capsaicin unreflected cylinder: Number of outer and inner
JFNK iterations, sweeps and run-time spent computing the residual to an
accuracy of $\|F\|_{2,s}\leq 10^{-9}$ for the various methods.}

\subfloat[Outer JFNK/total inner GMRES iterations]
{
\begin{tabular}{cccccc}  \toprule[1pt]
  & \multicolumn{5}{c}{$\eta$} \\ \cline{2-6}
\raisebox{1.5ex}[0cm][0cm]{subspace}
        &	0.1		&	0.01		&	0.001		&	EW1		&	EW2		\\\hline
30       &  8 (35)  &  6 (47)  &  5 (49)  &  5 (37)  &  5 (39)  \\
20       &  8 (35)  &  6 (47)  &  5 (49)  &  5 (37)  &  5 (39)  \\
10       &  8 (35)  &  6 (60)  &  5 (62)  &  5 (51)  &  5 (52)  \\
5        &  9 (62)  &  6 (66)  &  5 (68)  &  6 (85)  &  5 (53)  \\ \bottomrule[1pt]
\label{table:Cap-JFNK-norefl}
\end{tabular}}

\subfloat[Number of sweeps (FPI converged in 202)]
{
\begin{tabular}{ccccccccc}  \toprule[1pt]
	&		&		&		&		&		& JFNK $\eta$  &		&		\\\cline{5-9}
\raisebox{1.5ex}[0cm][0cm]{subspace}
        &	\raisebox{1.5ex}[0cm][0cm]{NKA}
                        &	\raisebox{1.5ex}[0cm][0cm]{NKA$_{-1}$}
                                                        &     \raisebox{1.5ex}[0cm][0cm]{Broyden}
                                                      	&	0.1	&	0.01	&	0.001	&	EW1	&	EW2	\\\hline
30       & 31 & 52 & 39  &   53   &    61  &  61  &  53  &  51  \\
20       & 31 & 55 & 53  &   53   &    61  &  61  &  53  &  51  \\
10       & 32 & -- & 127 &   53   &    76  &  76  &  70  &  66  \\
5        & 34 & -- & 73  &   88   &    89  &  90  & 116  &  72  \\ \bottomrule[1pt]
\label{table:Cap-sweep-norefl}
\end{tabular}}

\subfloat[CPU time (ks) (FPI converged in 12.71 ks)]
{
\begin{tabular}{ccccccccc}  \toprule[1pt]
	&		&		&		&		&		& JFNK $\eta$  &		&		\\\cline{5-9}
\raisebox{1.5ex}[0cm][0cm]{subspace}
        &	\raisebox{1.5ex}[0cm][0cm]{NKA}
                        &	\raisebox{1.5ex}[0cm][0cm]{NKA$_{-1}$}
                                                        &     \raisebox{1.5ex}[0cm][0cm]{Broyden}
                                                      	&	0.1	&	0.01	&	0.001	&	EW1	&	EW2	\\\hline
30       & 1.9 & 3.3 & 2.4  &   3.4   &    4.2  &  3.9  &  3.4  &  3.4  \\
20       & 1.9 & 3.6 & 3.4  &   3.4   &    4.1  &  4.0  &  3.5  &  3.3  \\
10       & 2.1 & --  & 8.2  &   3.5   &    5.0  &  5.1  &  4.7  &  4.3  \\
5        & 2.1 & --  & 4.5  &   5.7   &    5.9  &  5.9  &  7.6  &  4.7  \\ \bottomrule[1pt]
\label{table:Cap-time-norefl}
\end{tabular}}
\end{table}

\begin{figure}[h!]
\centering
\caption{Capsaicin unreflected cylinder: Scaled $L^2$-norm of the
residual as a function of number of sweeps for the various methods and
subspace sizes. Each of the methods is plotted on the same scale to simplify 
comparisons between panels. Note that, in panel (d),  points
plotted on the lines indicate when a JFNK iteration starts (JFNK requires
multiple sweeps per iteration). In panels (a), (b), and (c) we plot one point
per two iterations of the method. In panel (e) the convention for iterations
per plotted points is the same as in panels (a) through (d).}
\begin{tabular}{cc}
\subfloat[NKA and FPI]{ 
 \resizebox{80mm}{!}{
\begingroup%
\makeatletter%
\newcommand{\GNUPLOTspecial}{%
  \@sanitize\catcode`\%=14\relax\special}%
\setlength{\unitlength}{0.0500bp}%
\begin{picture}(7200,5040)(0,0)%
  \put(5936,3836){\makebox(0,0)[r]{\strut{}NKA(30)}}%
  \put(5936,4036){\makebox(0,0)[r]{\strut{}NKA(20)}}%
  \put(5936,4236){\makebox(0,0)[r]{\strut{}NKA(10)}}%
  \put(5936,4436){\makebox(0,0)[r]{\strut{}NKA(5)}}%
  \put(5936,4636){\makebox(0,0)[r]{\strut{}FPI}}%
  \put(4089,140){\makebox(0,0){\strut{}Sweeps}}%
  \put(160,2719){%
\rotatebox{-270}{%
  \makebox(0,0){\strut{}$\| {F} \|_{2,s}$}%
}}%
  \put(6619,440){\makebox(0,0){\strut{} 120}}%
  \put(5739,440){\makebox(0,0){\strut{} 100}}%
  \put(4859,440){\makebox(0,0){\strut{} 80}}%
  \put(3980,440){\makebox(0,0){\strut{} 60}}%
  \put(3100,440){\makebox(0,0){\strut{} 40}}%
  \put(2220,440){\makebox(0,0){\strut{} 20}}%
  \put(1340,440){\makebox(0,0){\strut{} 0}}%
  \put(1220,4799){\makebox(0,0)[r]{\strut{} 1}}%
  \put(1220,3967){\makebox(0,0)[r]{\strut{} 0.01}}%
  \put(1220,3135){\makebox(0,0)[r]{\strut{} 0.0001}}%
  \put(1220,2304){\makebox(0,0)[r]{\strut{} 1e-06}}%
  \put(1220,1472){\makebox(0,0)[r]{\strut{} 1e-08}}%
  \put(1220,640){\makebox(0,0)[r]{\strut{} 1e-10}}%
\includegraphics{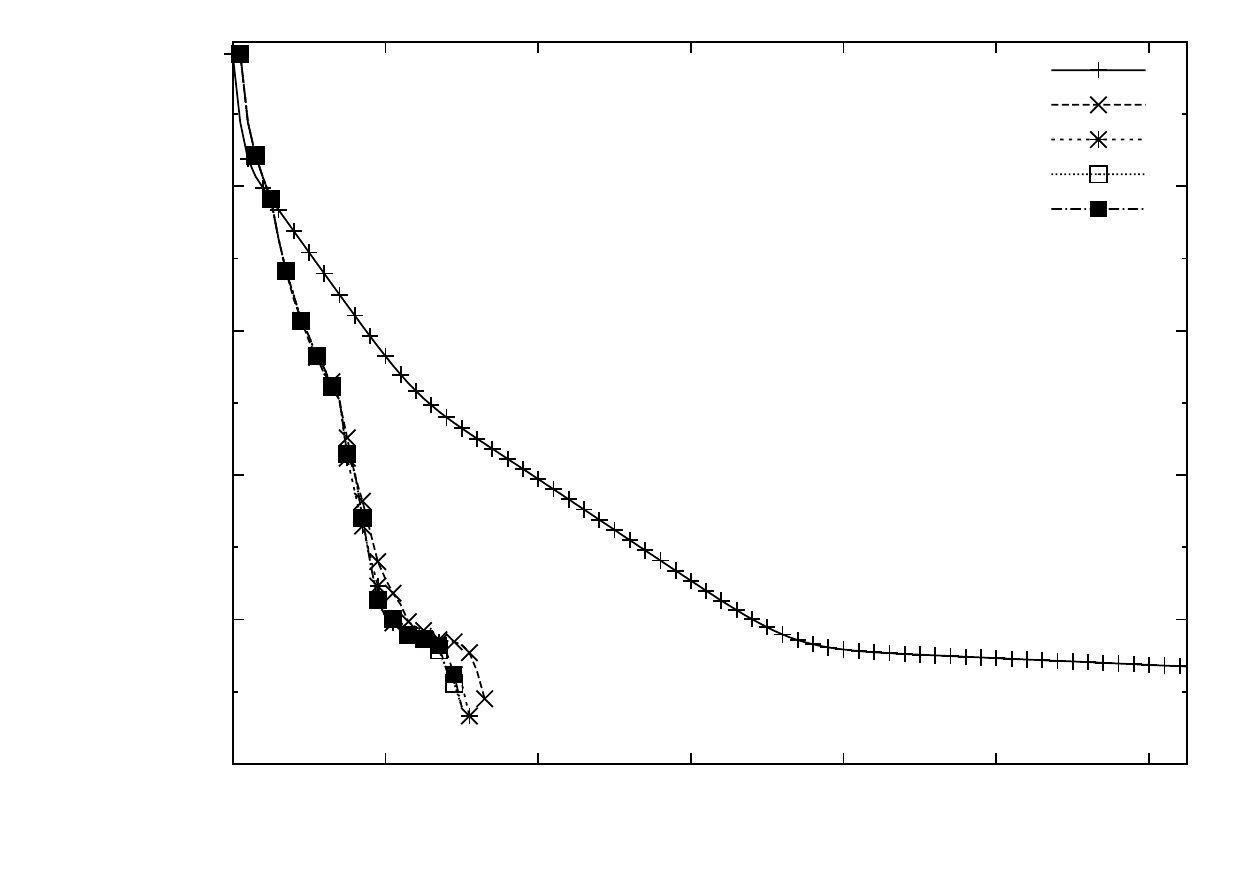}%
\end{picture}%
\endgroup

}
 \label{fig:no_refl_nka_cap}}
&
\subfloat[NKA$_{-1}$]{
 \resizebox{80mm}{!}{
\begingroup%
\makeatletter%
\newcommand{\GNUPLOTspecial}{%
  \@sanitize\catcode`\%=14\relax\special}%
\setlength{\unitlength}{0.0500bp}%
\begin{picture}(7200,5040)(0,0)%
  \put(5936,4036){\makebox(0,0)[r]{\strut{}NKA$_{-1}$(30)}}%
  \put(5936,4236){\makebox(0,0)[r]{\strut{}NKA$_{-1}$(20)}}%
  \put(5936,4436){\makebox(0,0)[r]{\strut{}NKA$_{-1}$(10)}}%
  \put(5936,4636){\makebox(0,0)[r]{\strut{}NKA$_{-1}$(5)}}%
  \put(4089,140){\makebox(0,0){\strut{}Sweeps}}%
  \put(160,2719){%
\rotatebox{-270}{%
  \makebox(0,0){\strut{}$\| {F} \|_{2,s}$}%
}}%
  \put(6619,440){\makebox(0,0){\strut{} 120}}%
  \put(5739,440){\makebox(0,0){\strut{} 100}}%
  \put(4859,440){\makebox(0,0){\strut{} 80}}%
  \put(3980,440){\makebox(0,0){\strut{} 60}}%
  \put(3100,440){\makebox(0,0){\strut{} 40}}%
  \put(2220,440){\makebox(0,0){\strut{} 20}}%
  \put(1340,440){\makebox(0,0){\strut{} 0}}%
  \put(1220,4799){\makebox(0,0)[r]{\strut{} 1}}%
  \put(1220,3967){\makebox(0,0)[r]{\strut{} 0.01}}%
  \put(1220,3135){\makebox(0,0)[r]{\strut{} 0.0001}}%
  \put(1220,2304){\makebox(0,0)[r]{\strut{} 1e-06}}%
  \put(1220,1472){\makebox(0,0)[r]{\strut{} 1e-08}}%
  \put(1220,640){\makebox(0,0)[r]{\strut{} 1e-10}}%
\includegraphics{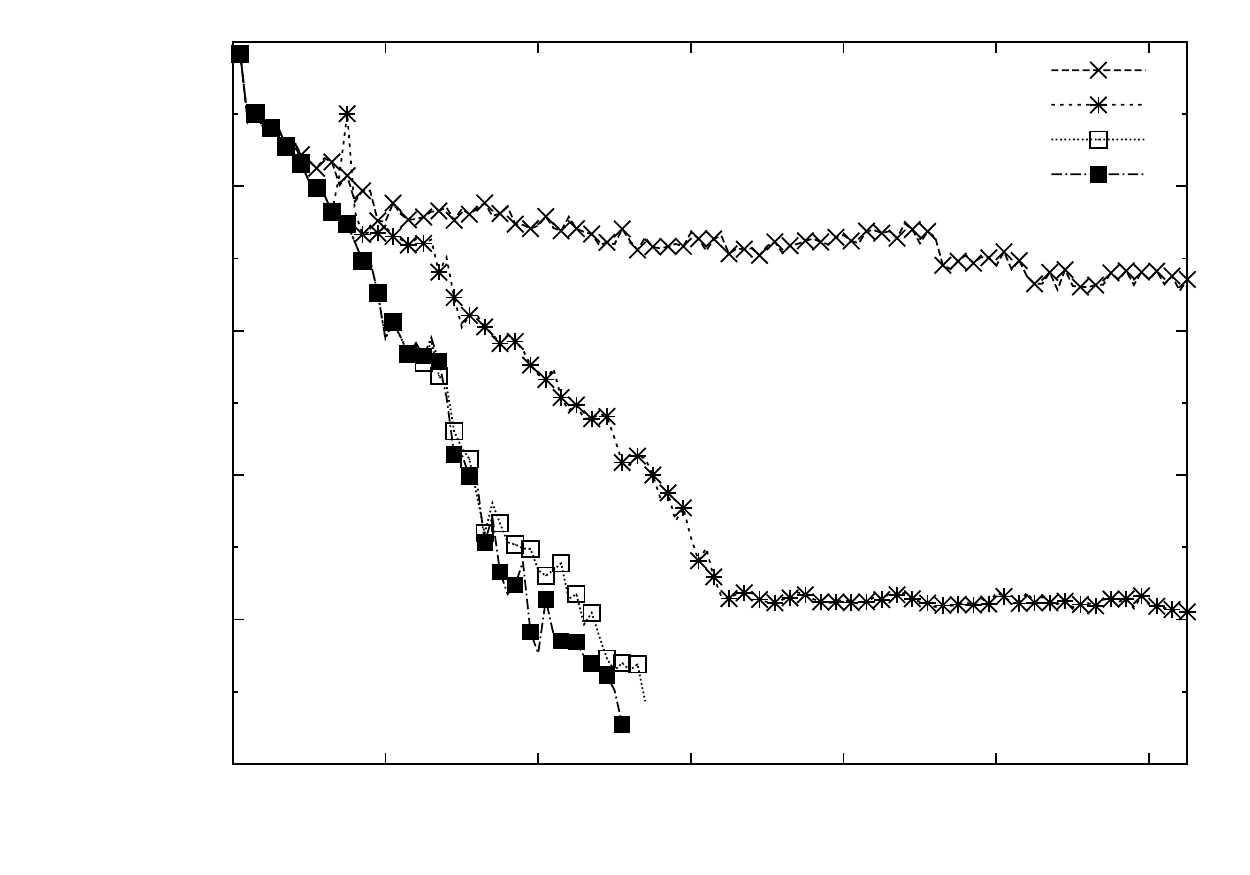}%
\end{picture}%
\endgroup

}
 \label{fig:no_refl_And_cap}}
\\
\subfloat[Broyden]{
 \resizebox{80mm}{!}{
\begingroup%
\makeatletter%
\newcommand{\GNUPLOTspecial}{%
  \@sanitize\catcode`\%=14\relax\special}%
\setlength{\unitlength}{0.0500bp}%
\begin{picture}(7200,5040)(0,0)%
  \put(5936,4036){\makebox(0,0)[r]{\strut{}Broyden(30)}}%
  \put(5936,4236){\makebox(0,0)[r]{\strut{}Broyden(20)}}%
  \put(5936,4436){\makebox(0,0)[r]{\strut{}Broyden(10)}}%
  \put(5936,4636){\makebox(0,0)[r]{\strut{}Broyden(5)}}%
  \put(4089,140){\makebox(0,0){\strut{}Sweeps}}%
  \put(160,2719){%
\rotatebox{-270}{%
  \makebox(0,0){\strut{}$\| {F} \|_{2,s}$}%
}}%
  \put(6619,440){\makebox(0,0){\strut{} 120}}%
  \put(5739,440){\makebox(0,0){\strut{} 100}}%
  \put(4859,440){\makebox(0,0){\strut{} 80}}%
  \put(3980,440){\makebox(0,0){\strut{} 60}}%
  \put(3100,440){\makebox(0,0){\strut{} 40}}%
  \put(2220,440){\makebox(0,0){\strut{} 20}}%
  \put(1340,440){\makebox(0,0){\strut{} 0}}%
  \put(1220,4799){\makebox(0,0)[r]{\strut{} 1}}%
  \put(1220,3967){\makebox(0,0)[r]{\strut{} 0.01}}%
  \put(1220,3135){\makebox(0,0)[r]{\strut{} 0.0001}}%
  \put(1220,2304){\makebox(0,0)[r]{\strut{} 1e-06}}%
  \put(1220,1472){\makebox(0,0)[r]{\strut{} 1e-08}}%
  \put(1220,640){\makebox(0,0)[r]{\strut{} 1e-10}}%
\includegraphics{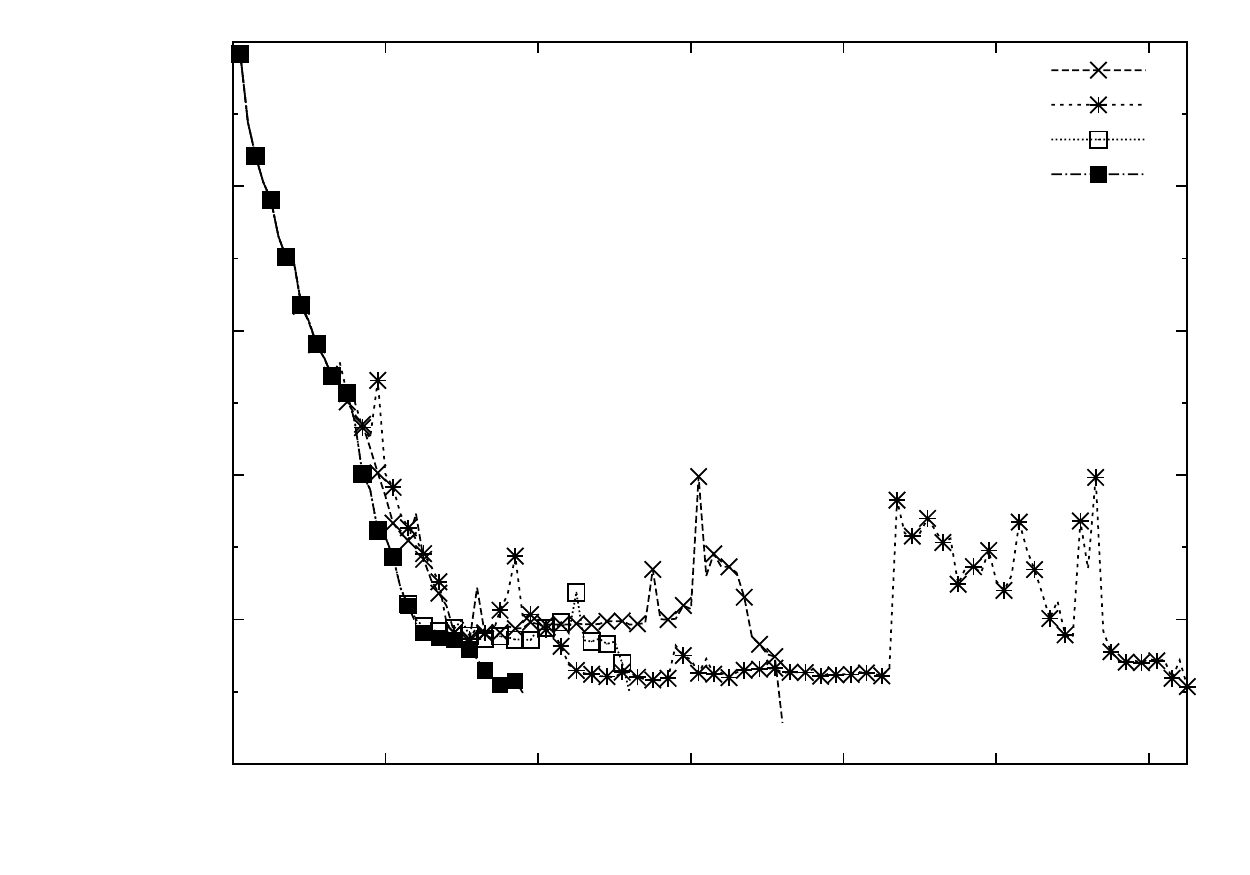}%
\end{picture}%
\endgroup

}
 \label{fig:no_refl_broy_cap}}
&
\subfloat[JFNK]{
 \resizebox{80mm}{!}{
\begingroup%
\makeatletter%
\newcommand{\GNUPLOTspecial}{%
  \@sanitize\catcode`\%=14\relax\special}%
\setlength{\unitlength}{0.0500bp}%
\begin{picture}(7200,5040)(0,0)%
  \put(5936,4036){\makebox(0,0)[r]{\strut{}GMRES(10), $\eta_{EW2}$}}%
  \put(5936,4236){\makebox(0,0)[r]{\strut{}GMRES(5), $\eta_{EW2}$}}%
  \put(5936,4436){\makebox(0,0)[r]{\strut{}GMRES(10), $\eta=0.1$}}%
  \put(5936,4636){\makebox(0,0)[r]{\strut{}GMRES(5), $\eta=0.1$}}%
  \put(4089,140){\makebox(0,0){\strut{}Sweeps}}%
  \put(160,2719){%
\rotatebox{-270}{%
  \makebox(0,0){\strut{}$\| {F} \|_{2,s}$}%
}}%
  \put(6619,440){\makebox(0,0){\strut{} 120}}%
  \put(5739,440){\makebox(0,0){\strut{} 100}}%
  \put(4859,440){\makebox(0,0){\strut{} 80}}%
  \put(3980,440){\makebox(0,0){\strut{} 60}}%
  \put(3100,440){\makebox(0,0){\strut{} 40}}%
  \put(2220,440){\makebox(0,0){\strut{} 20}}%
  \put(1340,440){\makebox(0,0){\strut{} 0}}%
  \put(1220,4799){\makebox(0,0)[r]{\strut{} 1}}%
  \put(1220,3967){\makebox(0,0)[r]{\strut{} 0.01}}%
  \put(1220,3135){\makebox(0,0)[r]{\strut{} 0.0001}}%
  \put(1220,2304){\makebox(0,0)[r]{\strut{} 1e-06}}%
  \put(1220,1472){\makebox(0,0)[r]{\strut{} 1e-08}}%
  \put(1220,640){\makebox(0,0)[r]{\strut{} 1e-10}}%
\includegraphics{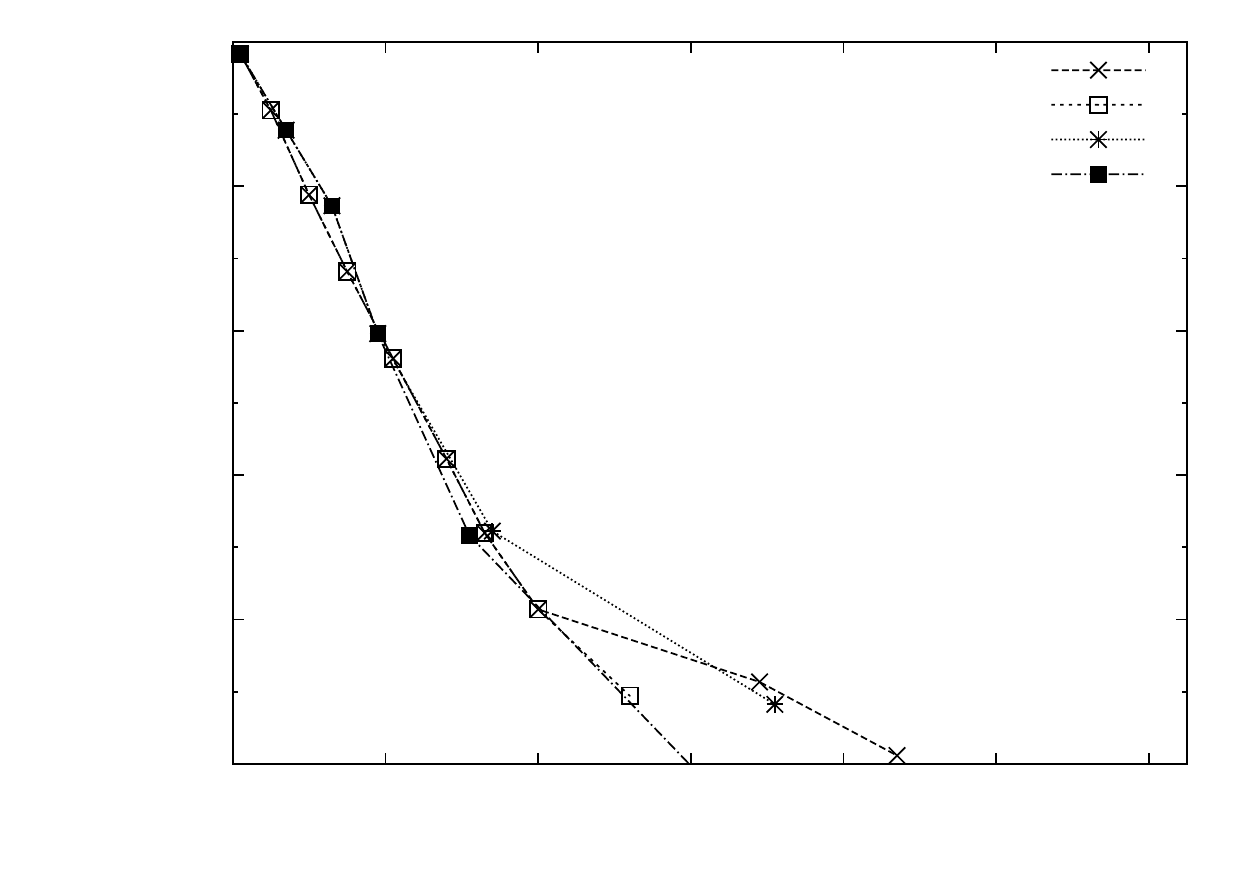}%
\end{picture}%
\endgroup

}
 \label{fig:no_refl_jfnk_cap}}
\\
\multicolumn{2}{c}{\subfloat[Comparison of the best results for each method]{
 \resizebox{80mm}{!}{
\begingroup%
\makeatletter%
\newcommand{\GNUPLOTspecial}{%
  \@sanitize\catcode`\%=14\relax\special}%
\setlength{\unitlength}{0.0500bp}%
\begin{picture}(7200,5040)(0,0)%
  \put(5936,3836){\makebox(0,0)[r]{\strut{}JFNK (GMRES(10), $\eta=0.1$)}}%
  \put(5936,4036){\makebox(0,0)[r]{\strut{}Broyden(30)}}%
  \put(5936,4236){\makebox(0,0)[r]{\strut{}NKA$_{-1}$(30)}}%
  \put(5936,4436){\makebox(0,0)[r]{\strut{}NKA(10)}}%
  \put(5936,4636){\makebox(0,0)[r]{\strut{}FPI}}%
  \put(4089,140){\makebox(0,0){\strut{}Sweeps}}%
  \put(160,2719){%
\rotatebox{-270}{%
  \makebox(0,0){\strut{}$\| {F} \|_{2,s}$}%
}}%
  \put(6619,440){\makebox(0,0){\strut{} 120}}%
  \put(5739,440){\makebox(0,0){\strut{} 100}}%
  \put(4859,440){\makebox(0,0){\strut{} 80}}%
  \put(3980,440){\makebox(0,0){\strut{} 60}}%
  \put(3100,440){\makebox(0,0){\strut{} 40}}%
  \put(2220,440){\makebox(0,0){\strut{} 20}}%
  \put(1340,440){\makebox(0,0){\strut{} 0}}%
  \put(1220,4799){\makebox(0,0)[r]{\strut{} 1}}%
  \put(1220,3967){\makebox(0,0)[r]{\strut{} 0.01}}%
  \put(1220,3135){\makebox(0,0)[r]{\strut{} 0.0001}}%
  \put(1220,2304){\makebox(0,0)[r]{\strut{} 1e-06}}%
  \put(1220,1472){\makebox(0,0)[r]{\strut{} 1e-08}}%
  \put(1220,640){\makebox(0,0)[r]{\strut{} 1e-10}}%
\includegraphics{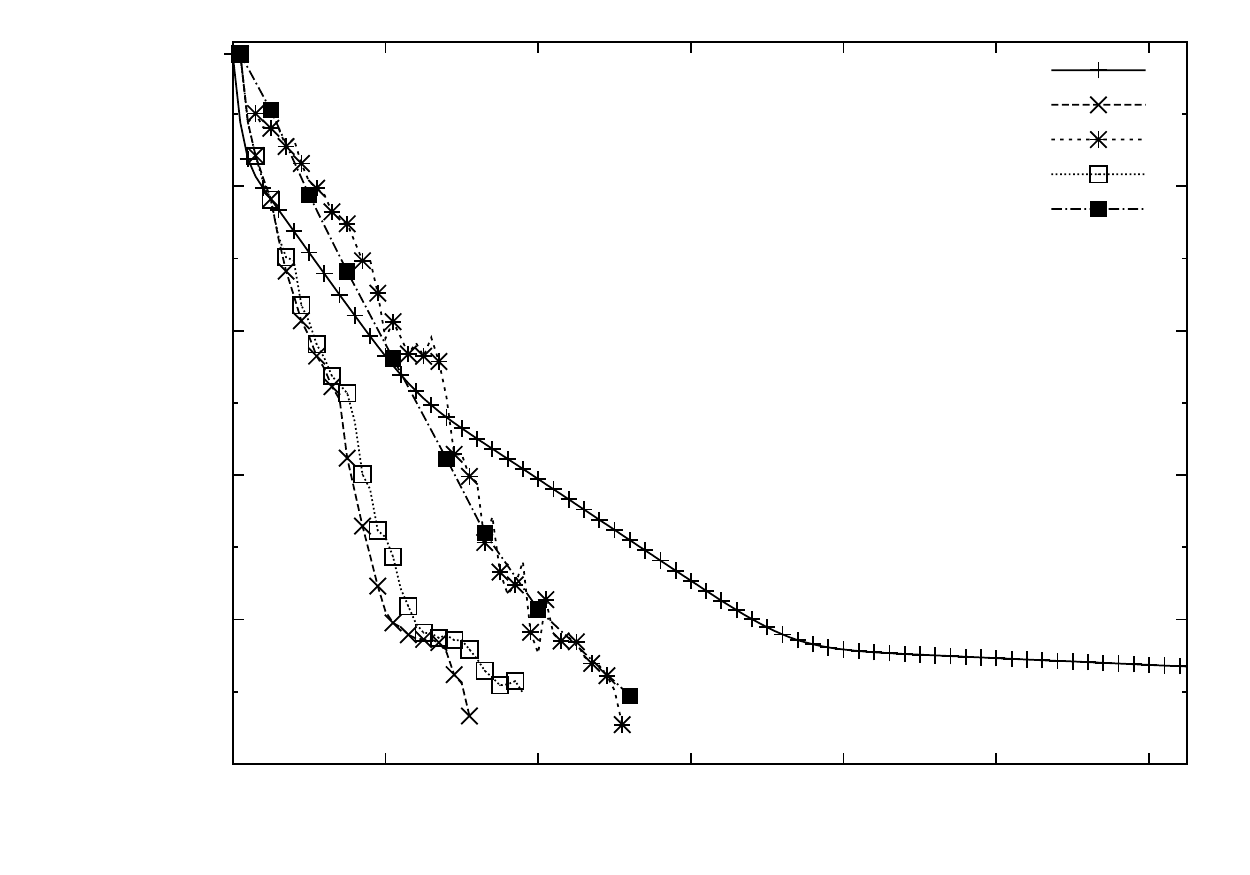}%
\end{picture}%
\endgroup

}
 \label{fig:no_refl_comp_cap}}}
\\
\end{tabular}
\label{fig:no_refl_cap}
\end{figure}
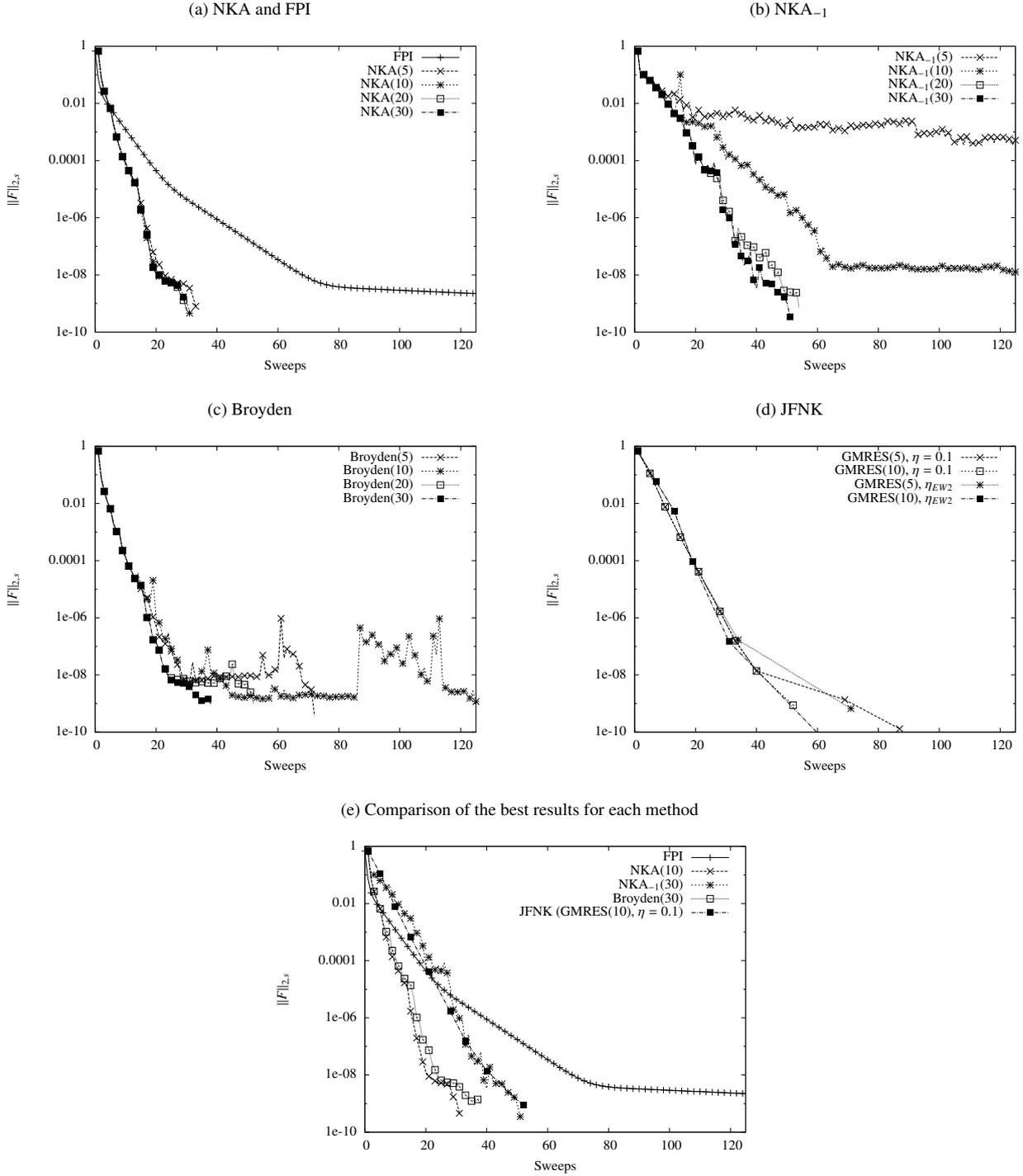

\clearpage
\subsubsection{Fully reflected cylinder}

\paragraph{PARTISN results}
Here, we duplicate the results shown in~\ref{section:PART-no_refl},
only this time, reflection is added to the top, bottom and outer edges
of the cylinder. Tables~\ref{table:JFNK-refl}-\ref{table:perc-refl}
show the number of JFNK and total inner GMRES iterations, the number
of sweeps, the CPU time and the percentage of that time that was spent
computing the residual, respectively. As can be seen in
Tables~\ref{table:JFNK-refl} and \ref{table:sweep-refl}, increasing
the subspace size has a non-negligible effect on Newton iteration
count, and a significant effect on GMRES iteration count and sweep
count. One additional parameter that comes into play for this problem
is the maximum number of inner GMRES iterations allowed. This
parameter was set to the NOX default of 30 for these computations and
in many cases the inner iteration did not converge to the specified
tolerance because it exceeded this limit. We note that, in our
experience, increasing this upper bound tends to degrade the
performance of JFNK, just as decreasing the forcing parameter often
degrades the performance. For NKA and Broyden, decreasing the subspace
size also leads to an increase in sweep/iteration count, quite
noticeably in this case. The only subspace size that converged for
NKA$_{-1}$ was 30, which requires more sweeps than any other iterative
method including FPI. For smaller subspace sizes, Broyden also fails
to converge. When comparing the runtimes shown in
Table~\ref{table:time-refl} for all of the iterative methods, we see
that JFNK with GMRES(5) is only slightly more efficient than FPI,
while JFNK with GMRES(30) is comparable to NKA(5). Broyden, when it
converges, is slower than NKA, but comparable to JFNK. The runtime for
NKA$_{-1}$ is an order of magnitude greater than any other
method. Table~\ref{table:perc-refl} shows that the percentage of time
spent in the residual is ranges from approximately 95\% for JFNK with
GMRES(5) to approximately 90\% for GMRES(30). The percentage is
smallest for NKA and NKA$_{-1}$, and slightly larger for Broyden. Once
again, we note that despite the fact that a larger percentage of time
is spent in the NKA solver than in JFNK or Broyden, it is still more
efficient in terms of CPU time.

The plots in Fig.~\ref{fig:refl} show the scaled $L^2$-norm of the
residual as a function of the number of
sweeps. Fig.~\ref{fig:refl_nka} shows that the behavior of NKA is
similar for all subspace sizes and much more efficient than FPI in all
Fig.~\ref{fig:refl_AM} shows that NKA$_{-1}$ behaves erratically for
subspaces of 5 and 10, and at 20 seems to stagnate before reaching the
specified convergence criterion. Even when it does converge for a
subspace size of 30, it is slower even than
FPI. Fig.~\ref{fig:refl_broy} shows the behavior of Broyden. As can be
seen, the norm fluctuates erratically and, although it appears to be
converging initially for subspaces of 5 and 10, it never achieves an
error norm below $10^{-8}$ and eventually diverges.  Broyden(30) and
Broyden(20), which converge to the desired norm, cannot compete with
NKA(20). Fig.~\ref{fig:refl_jfnk} shows some of the more efficient
JFNK results plotted at each Newton iteration for the current sweep
count. And finally, Fig.~\ref{fig:refl_comp} shows a comparison of
the most efficient results for each method. Clearly, NKA(20) is
significantly more efficient than the other iterative methods,
although we note that NKA(10) is comparable to JFNK with GMRES(20).

\begin{table}[h!]
\centering

\caption{PARTISN reflected cylinder: Number of outer and inner JFNK
iterations, sweeps, run-time and percentage of the run-time spent
computing the residual to an accuracy of $\|F\|_{2,s}\leq 10^{-9}$ for
the various methods.}
\subfloat[Outer JFNK/total inner GMRES iterations]
{
\begin{tabular}{cccccc}  \toprule[1pt]
  & \multicolumn{5}{c}{$\eta$} \\ \cline{2-6}
\raisebox{1.5ex}[0cm][0cm]{subspace}
	&	0.1		&	0.01		&	0.001		&	EW1		&	EW2	\\\hline
30	&	10	(177)	&	7	(167)	&	7	(192)	&	7	(165)	&	7	(175)	\\
20	&	10	(179)	&	8	(198)	&	7	(197)	&	8	(195)	&	7	(177)	\\
10	&	10	(186)	&	8	(202)	&	9	(257)	&	8	(196)	&	8	(212)	\\
5	&	12	(263)	&	10	(271)	&	10	(289)	&	10	(252)	&	10	(279)	\\\bottomrule[1pt]
\label{table:JFNK-refl}
\end{tabular}}

\subfloat[Number of sweeps (FPI converged in 389)]
{
\begin{tabular}{ccccccccc}  \toprule[1pt]
	&		&		&		&		&		& JFNK $\eta$  &		&		\\\cline{5-9}
\raisebox{1.5ex}[0cm][0cm]{subspace}
        &	\raisebox{1.5ex}[0cm][0cm]{NKA}
                        &	\raisebox{1.5ex}[0cm][0cm]{NKA$_{-1}$}
                                                        &     \raisebox{1.5ex}[0cm][0cm]{Broyden}
                                                      	&	0.1	&	0.01	&	0.001	&	EW1	&	EW2	\\\hline
30	&	115	&	594	&	218	&	198	&	182	&	207	&	186	&	190	\\
20	&	131	&	--	&	243	&	200	&	216	&	212	&	220	&	193	\\
10	&	173	&	--	&	--	&	218	&	232	&	293	&	233	&	243	\\
5	&	197	&	--	&	--	&	331	&	337	&	358	&	323	&	346	\\\bottomrule[1pt]

\label{table:sweep-refl}
\end{tabular}}

\subfloat[CPU time (s) (FPI converged in 971.7 s)]
{
\begin{tabular}{ccccccccc}  \toprule[1pt]
	&		&		&		&		&		& JFNK $\eta$  &		&		\\\cline{5-9}
\raisebox{1.5ex}[0cm][0cm]{subspace}
        &	\raisebox{1.5ex}[0cm][0cm]{NKA}
                        &	\raisebox{1.5ex}[0cm][0cm]{NKA$_{-1}$}
                                                        &     \raisebox{1.5ex}[0cm][0cm]{Broyden}
                                                      	&	0.1	&	0.01	&	0.001	&	EW1	&	EW2	\\\hline
30	&	372.2	&	1975.1	&	645.7	&	536.8	&	498.6	&	569.8	&	509.8	&	523.3	\\
20	&	398.0	&	--	&	686.6	&	529.7	&	577.9	&	565.4	&	591.2	&	516.0	\\
10	&	486.3	&	--	&	--	&	572.1	&	613.4	&	769.9	&	613.2	&	639.9	\\
5	&	535.7	&	--	&	--	&	855.8	&	895.0	&	932.4	&	842.2	&	906.0	\\\bottomrule[1pt]
\label{table:time-refl}
\end{tabular}} 

\subfloat[Percentage of CPU time spent in the residual evaluation]
{
\begin{tabular}{ccccccccc}  \toprule[1pt]
	&		&		&		&		&		& JFNK $\eta$  &		&		\\\cline{5-9}
\raisebox{1.5ex}[0cm][0cm]{subspace}
        &	\raisebox{1.5ex}[0cm][0cm]{NKA}
                        &	\raisebox{1.5ex}[0cm][0cm]{NKA$_{-1}$}
                                                        &     \raisebox{1.5ex}[0cm][0cm]{Broyden}
                                                      	&	0.1	&	0.01	&	0.001	&	EW1	&	EW2	\\\hline
30	&	76.44	&	74.51	&	83.48	&	91.40	&	90.68	&	90.36	&	90.94	&	90.59	\\
20	&	81.05	&	--	&	87.25	&	93.06	&	92.83	&	92.64	&	92.97	&	92.74	\\
10	&	87.02	&	--	&	--	&	94.55	&	94.40	&	94.35	&	94.48	&	94.36	\\
5	&	91.04	&	--	&	--	&	95.61	&	95.66	&	95.57	&	95.62	&	95.58	\\\bottomrule[1pt]
\label{table:perc-refl}
\end{tabular}} 
\end{table}

\begin{figure}
\caption{PARTISN reflected cylinder: Scaled $L^2$-norm of the residual
  as a function of number of sweeps for the various methods and
  subspace sizes. Each of the methods is plotted on the same scale to simplify 
comparisons between panels. Note that, in panel (d),  points
plotted on the lines indicate when a JFNK iteration starts (JFNK requires
multiple sweeps per iteration). In panels (a), (b), and (c) we plot one point
per six iterations of the method. In panel (e) the convention for iterations
per plotted points is the same as in panels (a) through (d).}  
\centering
\begin{tabular}{cc}
\subfloat[NKA and FPI]{
\resizebox{80mm}{!}{
\begingroup%
\makeatletter%
\newcommand{\GNUPLOTspecial}{%
  \@sanitize\catcode`\%=14\relax\special}%
\setlength{\unitlength}{0.0500bp}%
\begin{picture}(7200,5040)(0,0)%
  \put(5936,3836){\makebox(0,0)[r]{\strut{}NKA(30)}}%
  \put(5936,4036){\makebox(0,0)[r]{\strut{}NKA(20)}}%
  \put(5936,4236){\makebox(0,0)[r]{\strut{}NKA(10)}}%
  \put(5936,4436){\makebox(0,0)[r]{\strut{}NKA(5)}}%
  \put(5936,4636){\makebox(0,0)[r]{\strut{}FPI}}%
  \put(4089,140){\makebox(0,0){\strut{}Sweeps}}%
  \put(160,2719){%
\rotatebox{-270}{%
  \makebox(0,0){\strut{}$\| {F} \|_{2,s}$}%
}}%
  \put(6839,440){\makebox(0,0){\strut{} 350}}%
  \put(6053,440){\makebox(0,0){\strut{} 300}}%
  \put(5268,440){\makebox(0,0){\strut{} 250}}%
  \put(4482,440){\makebox(0,0){\strut{} 200}}%
  \put(3697,440){\makebox(0,0){\strut{} 150}}%
  \put(2911,440){\makebox(0,0){\strut{} 100}}%
  \put(2126,440){\makebox(0,0){\strut{} 50}}%
  \put(1340,440){\makebox(0,0){\strut{} 0}}%
  \put(1220,4799){\makebox(0,0)[r]{\strut{} 1}}%
  \put(1220,3967){\makebox(0,0)[r]{\strut{} 0.01}}%
  \put(1220,3135){\makebox(0,0)[r]{\strut{} 0.0001}}%
  \put(1220,2304){\makebox(0,0)[r]{\strut{} 1e-06}}%
  \put(1220,1472){\makebox(0,0)[r]{\strut{} 1e-08}}%
  \put(1220,640){\makebox(0,0)[r]{\strut{} 1e-10}}%
\includegraphics{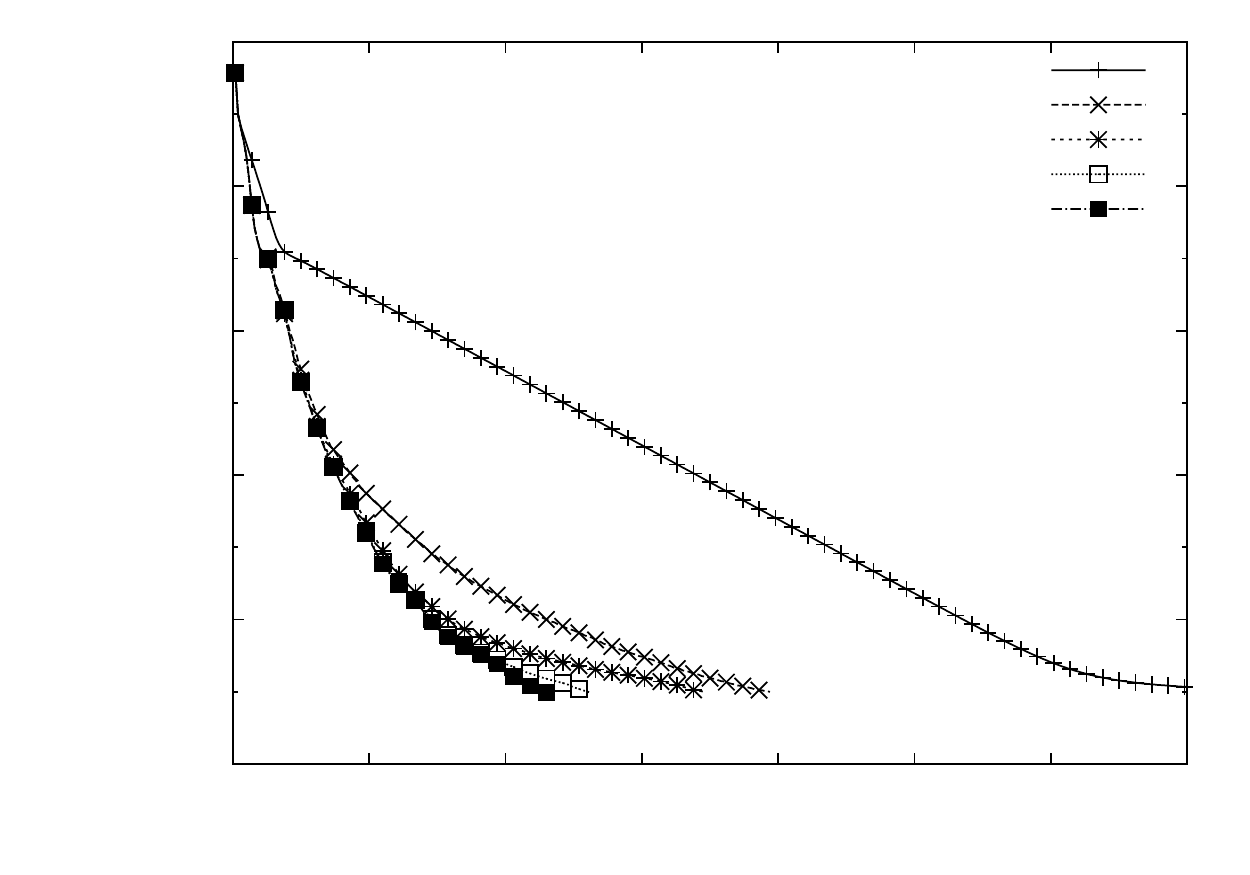}%
\end{picture}%
\endgroup

}
 \label{fig:refl_nka}}
&
\subfloat[NKA$_{-1}$]{
 \resizebox{80mm}{!}{
\begingroup%
\makeatletter%
\newcommand{\GNUPLOTspecial}{%
  \@sanitize\catcode`\%=14\relax\special}%
\setlength{\unitlength}{0.0500bp}%
\begin{picture}(7200,5040)(0,0)%
  \put(5936,4036){\makebox(0,0)[r]{\strut{}NKA$_{-1}$(30)}}%
  \put(5936,4236){\makebox(0,0)[r]{\strut{}NKA$_{-1}$(20)}}%
  \put(5936,4436){\makebox(0,0)[r]{\strut{}NKA$_{-1}$(10)}}%
  \put(5936,4636){\makebox(0,0)[r]{\strut{}NKA$_{-1}$(5)}}%
  \put(4089,140){\makebox(0,0){\strut{}Sweeps}}%
  \put(160,2719){%
\rotatebox{-270}{%
  \makebox(0,0){\strut{}$\| {F} \|_{2,s}$}%
}}%
  \put(6839,440){\makebox(0,0){\strut{} 350}}%
  \put(6053,440){\makebox(0,0){\strut{} 300}}%
  \put(5268,440){\makebox(0,0){\strut{} 250}}%
  \put(4482,440){\makebox(0,0){\strut{} 200}}%
  \put(3697,440){\makebox(0,0){\strut{} 150}}%
  \put(2911,440){\makebox(0,0){\strut{} 100}}%
  \put(2126,440){\makebox(0,0){\strut{} 50}}%
  \put(1340,440){\makebox(0,0){\strut{} 0}}%
  \put(1220,4799){\makebox(0,0)[r]{\strut{} 1}}%
  \put(1220,3967){\makebox(0,0)[r]{\strut{} 0.01}}%
  \put(1220,3135){\makebox(0,0)[r]{\strut{} 0.0001}}%
  \put(1220,2304){\makebox(0,0)[r]{\strut{} 1e-06}}%
  \put(1220,1472){\makebox(0,0)[r]{\strut{} 1e-08}}%
  \put(1220,640){\makebox(0,0)[r]{\strut{} 1e-10}}%
\includegraphics{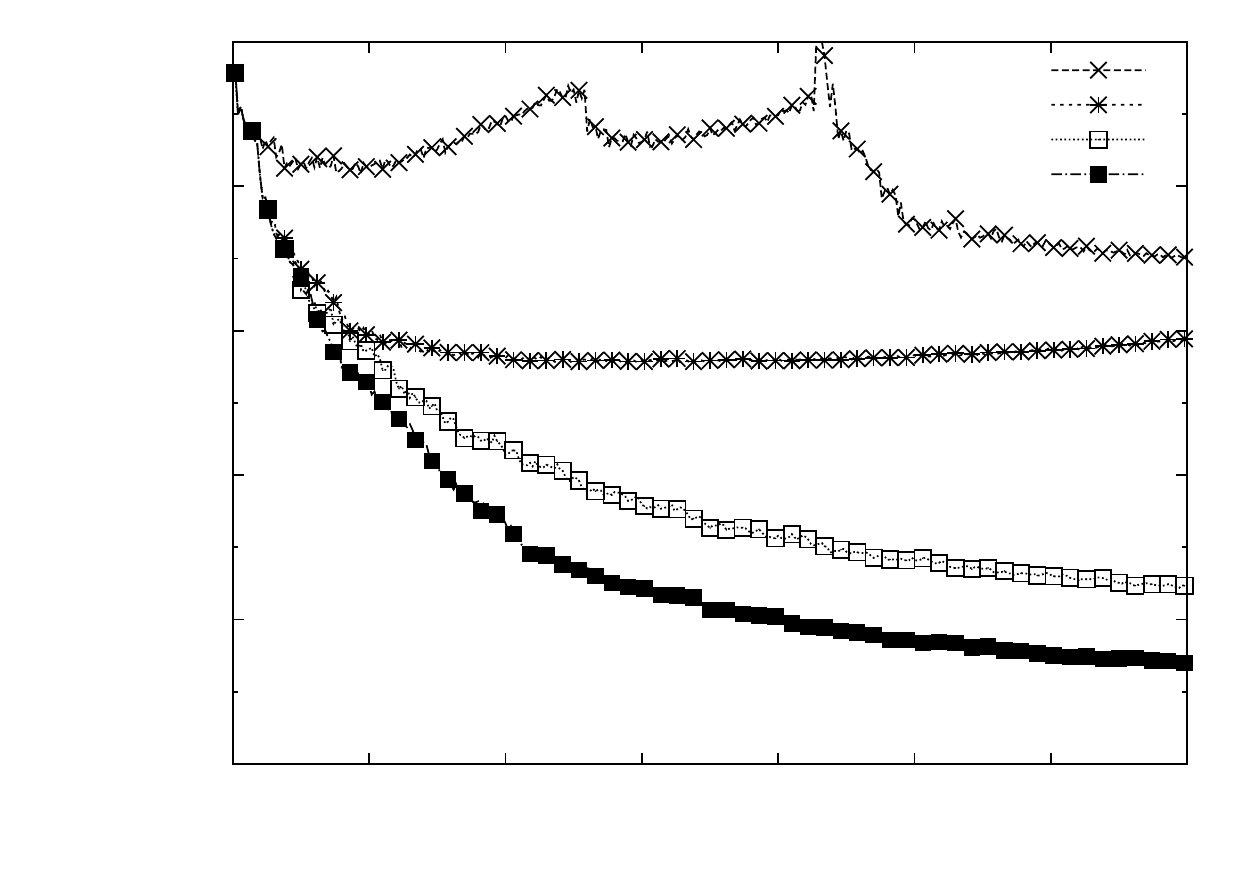}%
\end{picture}%
\endgroup

}
 \label{fig:refl_AM}}
\\
\subfloat[Broyden]{
 \resizebox{80mm}{!}{
\begingroup%
\makeatletter%
\newcommand{\GNUPLOTspecial}{%
  \@sanitize\catcode`\%=14\relax\special}%
\setlength{\unitlength}{0.0500bp}%
\begin{picture}(7200,5040)(0,0)%
  \put(5936,4036){\makebox(0,0)[r]{\strut{}Broyden(30)}}%
  \put(5936,4236){\makebox(0,0)[r]{\strut{}Broyden(20)}}%
  \put(5936,4436){\makebox(0,0)[r]{\strut{}Broyden(10)}}%
  \put(5936,4636){\makebox(0,0)[r]{\strut{}Broyden(5)}}%
  \put(4089,140){\makebox(0,0){\strut{}Sweeps}}%
  \put(160,2719){%
\rotatebox{-270}{%
  \makebox(0,0){\strut{}$\| {F} \|_{2,s}$}%
}}%
  \put(6839,440){\makebox(0,0){\strut{} 350}}%
  \put(6053,440){\makebox(0,0){\strut{} 300}}%
  \put(5268,440){\makebox(0,0){\strut{} 250}}%
  \put(4482,440){\makebox(0,0){\strut{} 200}}%
  \put(3697,440){\makebox(0,0){\strut{} 150}}%
  \put(2911,440){\makebox(0,0){\strut{} 100}}%
  \put(2126,440){\makebox(0,0){\strut{} 50}}%
  \put(1340,440){\makebox(0,0){\strut{} 0}}%
  \put(1220,4799){\makebox(0,0)[r]{\strut{} 1}}%
  \put(1220,3967){\makebox(0,0)[r]{\strut{} 0.01}}%
  \put(1220,3135){\makebox(0,0)[r]{\strut{} 0.0001}}%
  \put(1220,2304){\makebox(0,0)[r]{\strut{} 1e-06}}%
  \put(1220,1472){\makebox(0,0)[r]{\strut{} 1e-08}}%
  \put(1220,640){\makebox(0,0)[r]{\strut{} 1e-10}}%
\includegraphics{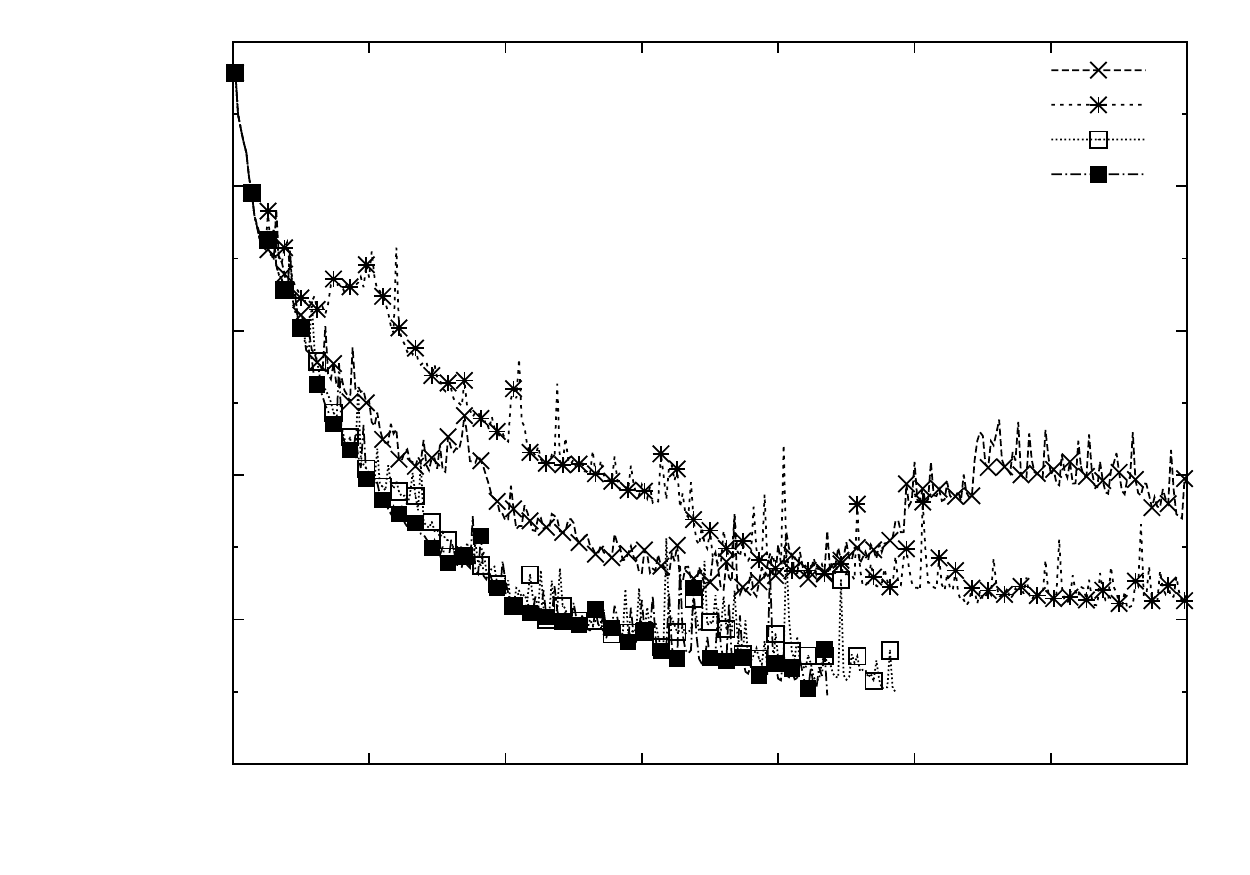}%
\end{picture}%
\endgroup

}
 \label{fig:refl_broy}}
&
\subfloat[JFNK]{
 \resizebox{80mm}{!}{
\begingroup%
\makeatletter%
\newcommand{\GNUPLOTspecial}{%
  \@sanitize\catcode`\%=14\relax\special}%
\setlength{\unitlength}{0.0500bp}%
\begin{picture}(7200,5040)(0,0)%
  \put(5936,4036){\makebox(0,0)[r]{\strut{}GMRES(20), $\eta_{EW2}$}}%
  \put(5936,4236){\makebox(0,0)[r]{\strut{}GMRES(5), $\eta_{EW2}$}}%
  \put(5936,4436){\makebox(0,0)[r]{\strut{}GMRES(20), $\eta=0.1$}}%
  \put(5936,4636){\makebox(0,0)[r]{\strut{}GMRES(5), $\eta=0.1$}}%
  \put(4089,140){\makebox(0,0){\strut{}Sweeps}}%
  \put(160,2719){%
\rotatebox{-270}{%
  \makebox(0,0){\strut{}$\| {F} \|_{2,s}$}%
}}%
  \put(6839,440){\makebox(0,0){\strut{} 350}}%
  \put(6053,440){\makebox(0,0){\strut{} 300}}%
  \put(5268,440){\makebox(0,0){\strut{} 250}}%
  \put(4482,440){\makebox(0,0){\strut{} 200}}%
  \put(3697,440){\makebox(0,0){\strut{} 150}}%
  \put(2911,440){\makebox(0,0){\strut{} 100}}%
  \put(2126,440){\makebox(0,0){\strut{} 50}}%
  \put(1340,440){\makebox(0,0){\strut{} 0}}%
  \put(1220,4799){\makebox(0,0)[r]{\strut{} 1}}%
  \put(1220,3967){\makebox(0,0)[r]{\strut{} 0.01}}%
  \put(1220,3135){\makebox(0,0)[r]{\strut{} 0.0001}}%
  \put(1220,2304){\makebox(0,0)[r]{\strut{} 1e-06}}%
  \put(1220,1472){\makebox(0,0)[r]{\strut{} 1e-08}}%
  \put(1220,640){\makebox(0,0)[r]{\strut{} 1e-10}}%
\includegraphics{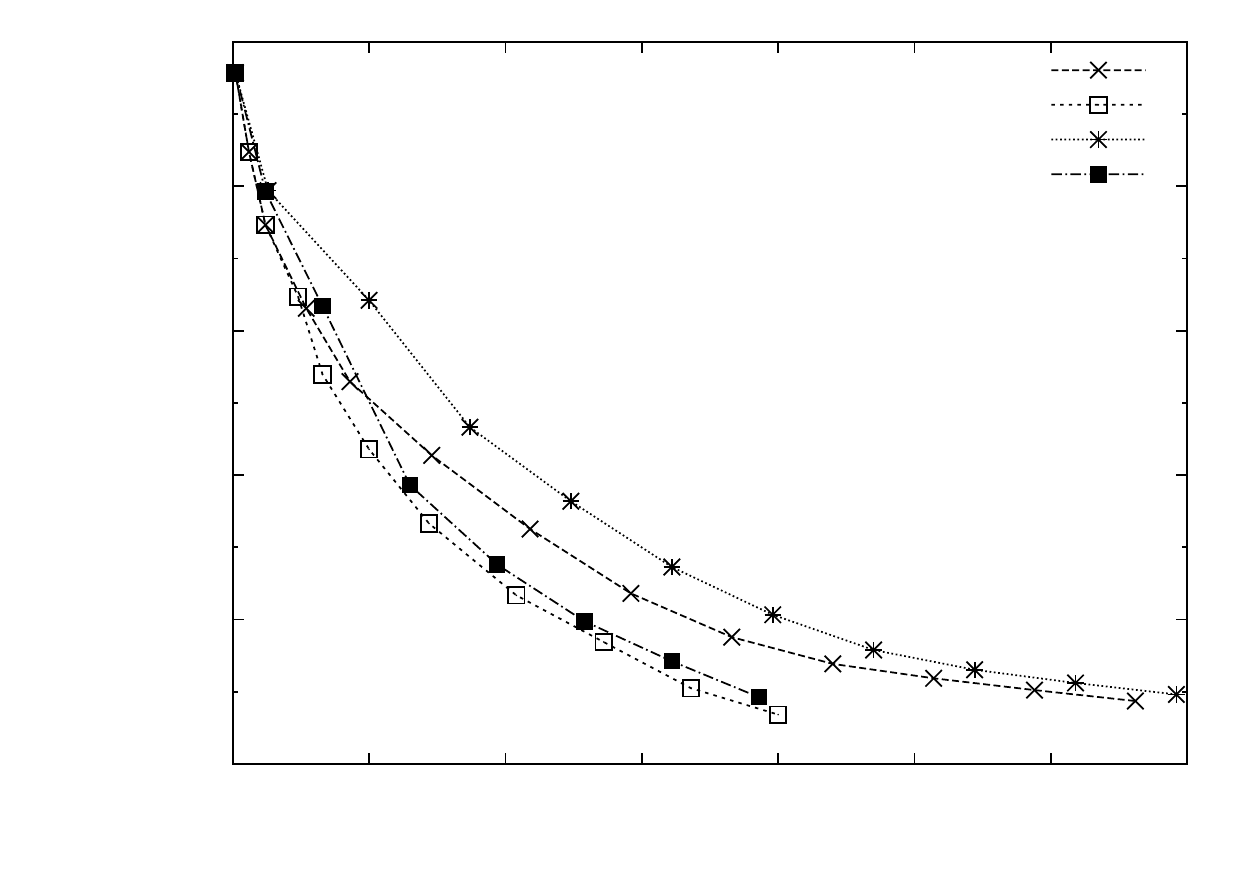}%
\end{picture}%
\endgroup

}
 \label{fig:refl_jfnk}}
\\
\multicolumn{2}{c}{
\subfloat[Comparison of the best results for each method]{
 \resizebox{80mm}{!}{
\begingroup%
\makeatletter%
\newcommand{\GNUPLOTspecial}{%
  \@sanitize\catcode`\%=14\relax\special}%
\setlength{\unitlength}{0.0500bp}%
\begin{picture}(7200,5040)(0,0)%
  \put(5936,3836){\makebox(0,0)[r]{\strut{}JFNK (GMRES(20), $\eta=0.1$)}}%
  \put(5936,4036){\makebox(0,0)[r]{\strut{}Broyden(30)}}%
  \put(5936,4236){\makebox(0,0)[r]{\strut{}NKA$_{-1}$(30)}}%
  \put(5936,4436){\makebox(0,0)[r]{\strut{}NKA(30)}}%
  \put(5936,4636){\makebox(0,0)[r]{\strut{}FPI}}%
  \put(4089,140){\makebox(0,0){\strut{}Sweeps}}%
  \put(160,2719){%
\rotatebox{-270}{%
  \makebox(0,0){\strut{}$\| {F} \|_{2,s}$}%
}}%
  \put(6839,440){\makebox(0,0){\strut{} 350}}%
  \put(6053,440){\makebox(0,0){\strut{} 300}}%
  \put(5268,440){\makebox(0,0){\strut{} 250}}%
  \put(4482,440){\makebox(0,0){\strut{} 200}}%
  \put(3697,440){\makebox(0,0){\strut{} 150}}%
  \put(2911,440){\makebox(0,0){\strut{} 100}}%
  \put(2126,440){\makebox(0,0){\strut{} 50}}%
  \put(1340,440){\makebox(0,0){\strut{} 0}}%
  \put(1220,4799){\makebox(0,0)[r]{\strut{} 1}}%
  \put(1220,3967){\makebox(0,0)[r]{\strut{} 0.01}}%
  \put(1220,3135){\makebox(0,0)[r]{\strut{} 0.0001}}%
  \put(1220,2304){\makebox(0,0)[r]{\strut{} 1e-06}}%
  \put(1220,1472){\makebox(0,0)[r]{\strut{} 1e-08}}%
  \put(1220,640){\makebox(0,0)[r]{\strut{} 1e-10}}%
\includegraphics{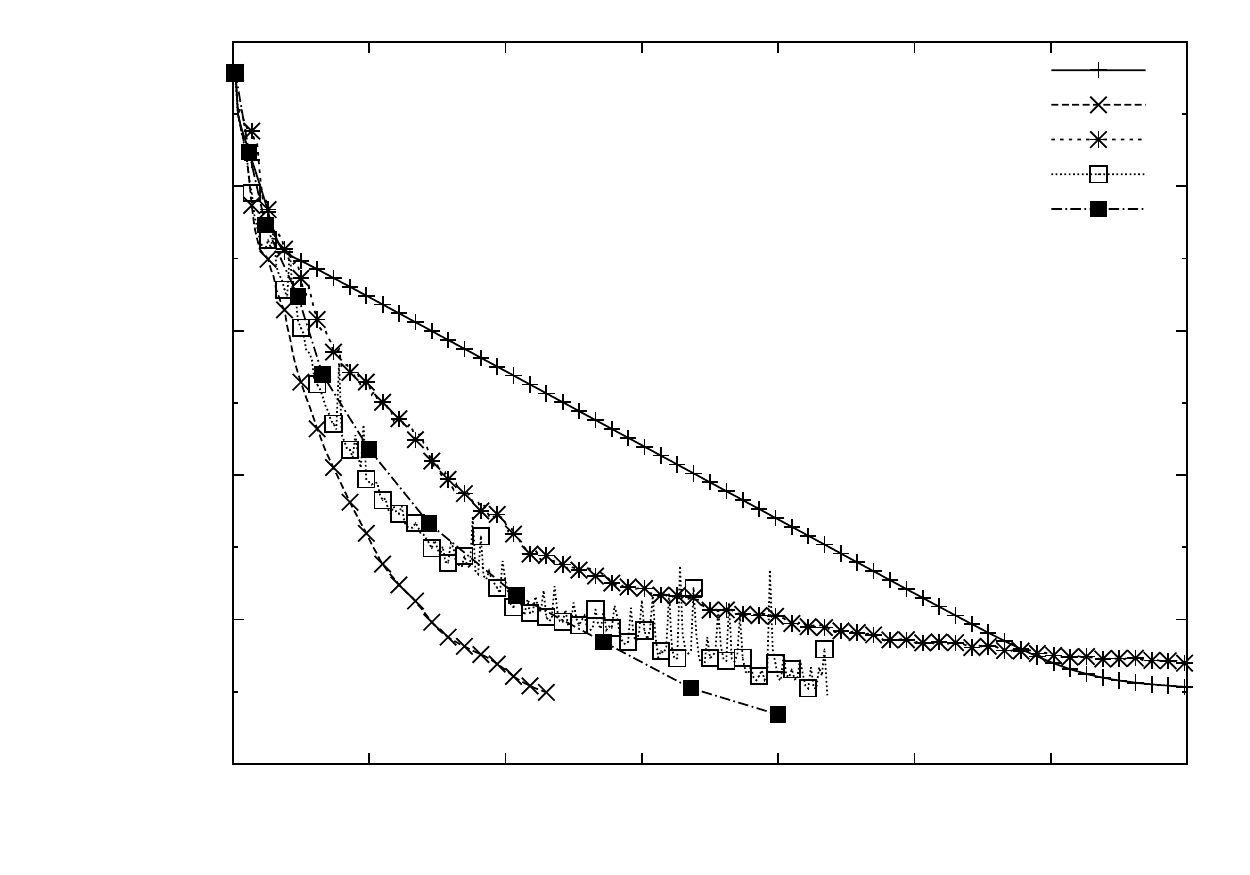}%
\end{picture}%
\endgroup

}
 \label{fig:refl_comp}}}
\\
\end{tabular}
 \label{fig:refl}
\end{figure}

\clearpage
\paragraph{Capsaicin results}

Here, we duplicate the results shown in~\ref{section:cap-no_refl},
only this time, reflection is added to the top, bottom and outer edges
of the cylinder. Tables~\ref{table:Cap-JFNK-refl}-\ref{table:Cap-time-refl}
show the number of JFNK and total inner GMRES iterations, the number
of sweeps and the CPU time that was spent computing the residual,
respectively. As can be seen in Tables~\ref{table:JFNK-refl} and
\ref{table:sweep-refl}, increasing the subspace size does not affect 
the Newton iteration count, but it does have a significant effect on
the GMRES iteration count and sweep count. In this regard, the
behavior is more similar to the unreflected case than the PARTISN
reflected case. For NKA and Broyden, there is an overall increase in
iteration count, but for Broyden it is slight compared to that seen
for the unreflected cylinder and the reflected PARTISN results. Also,
for Capsaicin, Broyden always converges whereas for PARTISN only
subspaces of 20 and 30 converged.  The only subspace size that
converged for NKA$_{-1}$ was 30, which is consistent with the PARTISN
results, and which once again requires more sweeps than any other
iterative method including FPI. When comparing the runtimes shown in
Table~\ref{table:Cap-time-refl} for all of the iterative methods, we
see that JFNK with GMRES(5) is generally slower than FPI, while JFNK
with GMRES(30) is slower than NKA with a subspace of 5. Broyden is
slower than NKA, but always more efficient than JFNK. The runtime for
NKA$_{-1}$ is twice that for FPI.

The plots in Fig.~\ref{fig:refl} show the scaled $L^2$-norm of the
residual as a function of the number of
sweeps. Fig.~\ref{fig:refl_nka_cap} shows that the behavior of NKA is
similar for all subspace sizes and much more efficient than FPI while
Fig.~\ref{fig:refl_And_cap} shows that NKA$_{-1}$ seems to stagnate
before reaching the specified convergence criterion subspaces of 5 and
10, and at 20 . Even when it does converge for a subspace size of 30,
it is slower even than FPI. Fig.~\ref{fig:refl_broy_cap} shows the
behavior of Broyden. As can be seen, the convergence is fairly
regular, in sharp constrast to the PARTISN results where the norm
fluctuated erratically and often failed to converge.
Fig.~\ref{fig:refl_jfnk_cap} shows some of the more efficient JFNK
results plotted at each Newton iteration for the current sweep
count. And finally, Fig.~\ref{fig:refl_comp} shows a comparison of the
most efficient results for each method. From these results, NKA(20) is more
efficient than the other iterative methods.  

\begin{table}[h!]
\centering
\caption{Capsaicin reflected cylinder: Number of outer and inner JFNK
iterations, sweeps and run-time spent computing the residual to an
accuracy of $\|F\|_{2,s}\leq 10^{-9}$ for the various methods.}
%
\subfloat[Outer JFNK/total inner GMRES iterations]
{
\begin{tabular}{cccccc}  \toprule[1pt]
  & \multicolumn{5}{c}{$\eta$} \\ \cline{2-6}
\raisebox{1.5ex}[0cm][0cm]{subspace}
	 &  0.1	    &	0.01   &  0.001	  &  EW1     &  EW2	\\\hline
30       &  9 (68)  &  5 (64)  &  4 (78)  &  4 (61)  &  4 (69)  \\
20       &  9 (68)  &  5 (64)  &  4 (87)  &  5 (74)  &  4 (69)  \\
10       &  9 (68)  &  5 (68)  &  4 (91)  &  4 (62)  &  4 (72)  \\
5        &  9 (76)  &  5 (137) &  4 (95)  &  4 (66)  &  4 (76)  \\ \bottomrule[1pt]
\label{table:Cap-JFNK-refl}
\end{tabular}}

\subfloat[Number of sweeps (FPI converged in 99)]
{
\begin{tabular}{ccccccccc}  \toprule[1pt]
	&		&		&		&		&		& JFNK $\eta$  &		&		\\\cline{5-9}
\raisebox{1.5ex}[0cm][0cm]{subspace}
        &	\raisebox{1.5ex}[0cm][0cm]{NKA}
                        &	\raisebox{1.5ex}[0cm][0cm]{NKA$_{-1}$}
                                                        &     \raisebox{1.5ex}[0cm][0cm]{Broyden}
                          &   0.1  &  0.01  &  0.001 &  EW1 &  EW2	  \\\hline
30       & 56 & 202 & 60  &   88   &    76  &  88    &  74  &  79  \\
20       & 57 & --  & 65  &   88   &    76  &  98    &  90  &  79  \\
10       & 56 & --  & 66  &   88   &    84  &  107   &  79  &  86  \\
5        & 62 & --  & 65  &   105  &    95  &  121   &  89  &  98  \\ \bottomrule[1pt]
\label{table:Cap-sweep-refl}
\end{tabular}}

\subfloat[CPU time (ks) (FPI converged in 6.25 ks)]
{
\begin{tabular}{ccccccccc}  \toprule[1pt]
	&		&		&		&		&		& JFNK $\eta$  &		&		\\\cline{5-9}
\raisebox{1.5ex}[0cm][0cm]{subspace}
        &	\raisebox{1.5ex}[0cm][0cm]{NKA}
                        &	\raisebox{1.5ex}[0cm][0cm]{NKA$_{-1}$}
                                                        &     \raisebox{1.5ex}[0cm][0cm]{Broyden}
                             &   0.1   &    0.01 &  0.001&  EW1  &  EW2	  \\\hline
30       & 3.5 & 12.8 & 3.9  &   5.9   &    5.2  &  5.7  &  5.0  &  5.2   \\
20       & 3.7 & --   & 4.0  &   5.8   &    5.0  &  6.4  &  6.0  &  5.2   \\
10       & 3.6 & --   & 4.2  &   5.7   &    5.3  &  6.9  &  5.2  &  5.6   \\
5        & 3.9 & --   & 4.2  &   6.7   &    6.5  &  7.8  &  5.8  &  6.6   \\ \bottomrule[1pt]
\label{table:Cap-time-refl}
\end{tabular}} 
\end{table}

\begin{figure}
\caption{Capsaicin reflected cylinder: Scaled $L^2$-norm of the residual as a function of number of
sweeps for the various methods and subspace sizes. Each of the methods is plotted on the same scale to simplify 
comparisons between panels. Note that, in panel (d),  points
plotted on the lines indicate when a JFNK iteration starts (JFNK requires
multiple sweeps per iteration). In panels (a), (b), and (c) we plot one point
per two iterations of the method. In panel (e) the convention for iterations
per plotted points is the same as in panels (a) through (d).} 
\begin{tabular}{cc}
\subfloat[NKA and FPI]{ 
 \resizebox{80mm}{!}{
\begingroup%
\makeatletter%
\newcommand{\GNUPLOTspecial}{%
  \@sanitize\catcode`\%=14\relax\special}%
\setlength{\unitlength}{0.0500bp}%
\begin{picture}(7200,5040)(0,0)%
  \put(5936,3836){\makebox(0,0)[r]{\strut{}NKA(30)}}%
  \put(5936,4036){\makebox(0,0)[r]{\strut{}NKA(20)}}%
  \put(5936,4236){\makebox(0,0)[r]{\strut{}NKA(10)}}%
  \put(5936,4436){\makebox(0,0)[r]{\strut{}NKA(5)}}%
  \put(5936,4636){\makebox(0,0)[r]{\strut{}FPI}}%
  \put(4089,140){\makebox(0,0){\strut{}Sweeps}}%
  \put(160,2719){%
\rotatebox{-270}{%
  \makebox(0,0){\strut{}$\| {F} \|_{2,s}$}%
}}%
  \put(6577,440){\makebox(0,0){\strut{} 100}}%
  \put(5530,440){\makebox(0,0){\strut{} 80}}%
  \put(4482,440){\makebox(0,0){\strut{} 60}}%
  \put(3435,440){\makebox(0,0){\strut{} 40}}%
  \put(2387,440){\makebox(0,0){\strut{} 20}}%
  \put(1340,440){\makebox(0,0){\strut{} 0}}%
  \put(1220,4799){\makebox(0,0)[r]{\strut{} 1}}%
  \put(1220,4383){\makebox(0,0)[r]{\strut{} 0.1}}%
  \put(1220,3967){\makebox(0,0)[r]{\strut{} 0.01}}%
  \put(1220,3551){\makebox(0,0)[r]{\strut{} 0.001}}%
  \put(1220,3135){\makebox(0,0)[r]{\strut{} 0.0001}}%
  \put(1220,2719){\makebox(0,0)[r]{\strut{} 1e-05}}%
  \put(1220,2304){\makebox(0,0)[r]{\strut{} 1e-06}}%
  \put(1220,1888){\makebox(0,0)[r]{\strut{} 1e-07}}%
  \put(1220,1472){\makebox(0,0)[r]{\strut{} 1e-08}}%
  \put(1220,1056){\makebox(0,0)[r]{\strut{} 1e-09}}%
  \put(1220,640){\makebox(0,0)[r]{\strut{} 1e-10}}%
\includegraphics{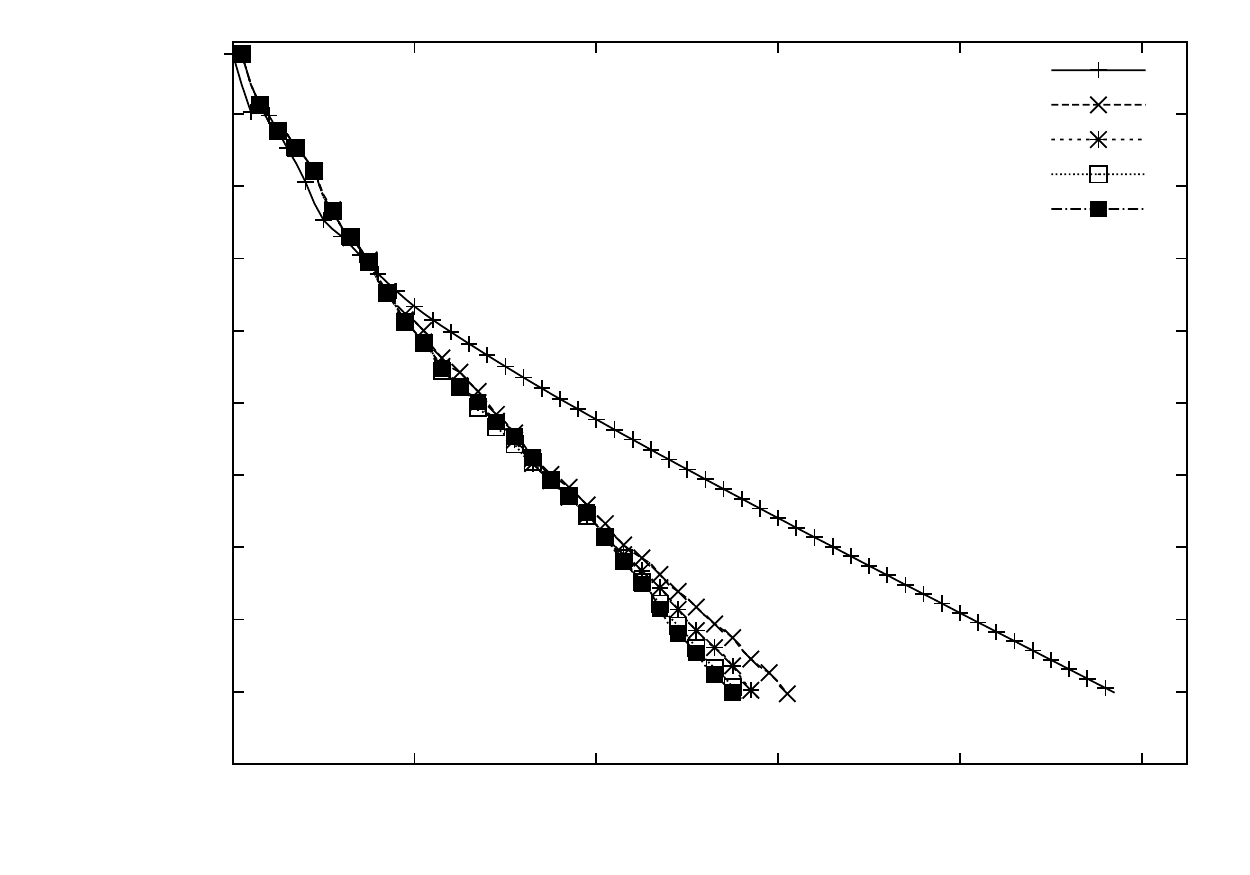}%
\end{picture}%
\endgroup

}
 \label{fig:refl_nka_cap}}
&
\subfloat[NKA$_{-1}$]{
 \resizebox{80mm}{!}{
\begingroup%
\makeatletter%
\newcommand{\GNUPLOTspecial}{%
  \@sanitize\catcode`\%=14\relax\special}%
\setlength{\unitlength}{0.0500bp}%
\begin{picture}(7200,5040)(0,0)%
  \put(5936,4036){\makebox(0,0)[r]{\strut{}NKA$_{-1}$(30)}}%
  \put(5936,4236){\makebox(0,0)[r]{\strut{}NKA$_{-1}$(20)}}%
  \put(5936,4436){\makebox(0,0)[r]{\strut{}NKA$_{-1}$(10)}}%
  \put(5936,4636){\makebox(0,0)[r]{\strut{}NKA$_{-1}$(5)}}%
  \put(4089,140){\makebox(0,0){\strut{}Sweeps}}%
  \put(160,2719){%
\rotatebox{-270}{%
  \makebox(0,0){\strut{}$\| {F} \|_{2,s}$}%
}}%
  \put(6577,440){\makebox(0,0){\strut{} 100}}%
  \put(5530,440){\makebox(0,0){\strut{} 80}}%
  \put(4482,440){\makebox(0,0){\strut{} 60}}%
  \put(3435,440){\makebox(0,0){\strut{} 40}}%
  \put(2387,440){\makebox(0,0){\strut{} 20}}%
  \put(1340,440){\makebox(0,0){\strut{} 0}}%
  \put(1220,4799){\makebox(0,0)[r]{\strut{} 1}}%
  \put(1220,3967){\makebox(0,0)[r]{\strut{} 0.01}}%
  \put(1220,3135){\makebox(0,0)[r]{\strut{} 0.0001}}%
  \put(1220,2304){\makebox(0,0)[r]{\strut{} 1e-06}}%
  \put(1220,1472){\makebox(0,0)[r]{\strut{} 1e-08}}%
  \put(1220,640){\makebox(0,0)[r]{\strut{} 1e-10}}%
\includegraphics{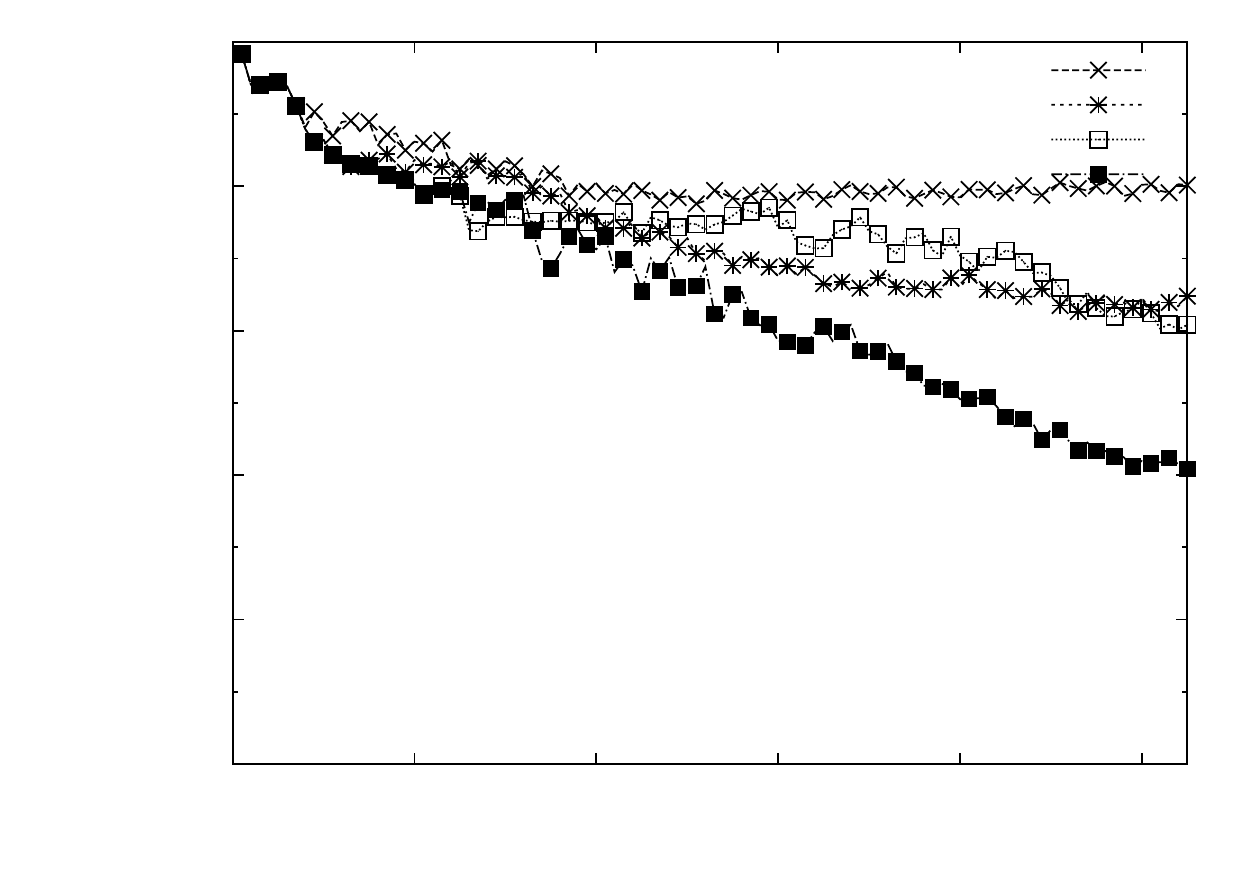}%
\end{picture}%
\endgroup

}
 \label{fig:refl_And_cap}}
\\
\subfloat[Broyden]{
 \resizebox{80mm}{!}{
\begingroup%
\makeatletter%
\newcommand{\GNUPLOTspecial}{%
  \@sanitize\catcode`\%=14\relax\special}%
\setlength{\unitlength}{0.0500bp}%
\begin{picture}(7200,5040)(0,0)%
  \put(5936,4036){\makebox(0,0)[r]{\strut{}Broyden(30)}}%
  \put(5936,4236){\makebox(0,0)[r]{\strut{}Broyden(20)}}%
  \put(5936,4436){\makebox(0,0)[r]{\strut{}Broyden(10)}}%
  \put(5936,4636){\makebox(0,0)[r]{\strut{}Broyden(5)}}%
  \put(4089,140){\makebox(0,0){\strut{}Sweeps}}%
  \put(160,2719){%
\rotatebox{-270}{%
  \makebox(0,0){\strut{}$\| {F} \|_{2,s}$}%
}}%
  \put(6577,440){\makebox(0,0){\strut{} 100}}%
  \put(5530,440){\makebox(0,0){\strut{} 80}}%
  \put(4482,440){\makebox(0,0){\strut{} 60}}%
  \put(3435,440){\makebox(0,0){\strut{} 40}}%
  \put(2387,440){\makebox(0,0){\strut{} 20}}%
  \put(1340,440){\makebox(0,0){\strut{} 0}}%
  \put(1220,4799){\makebox(0,0)[r]{\strut{} 1}}%
  \put(1220,3967){\makebox(0,0)[r]{\strut{} 0.01}}%
  \put(1220,3135){\makebox(0,0)[r]{\strut{} 0.0001}}%
  \put(1220,2304){\makebox(0,0)[r]{\strut{} 1e-06}}%
  \put(1220,1472){\makebox(0,0)[r]{\strut{} 1e-08}}%
  \put(1220,640){\makebox(0,0)[r]{\strut{} 1e-10}}%
\includegraphics{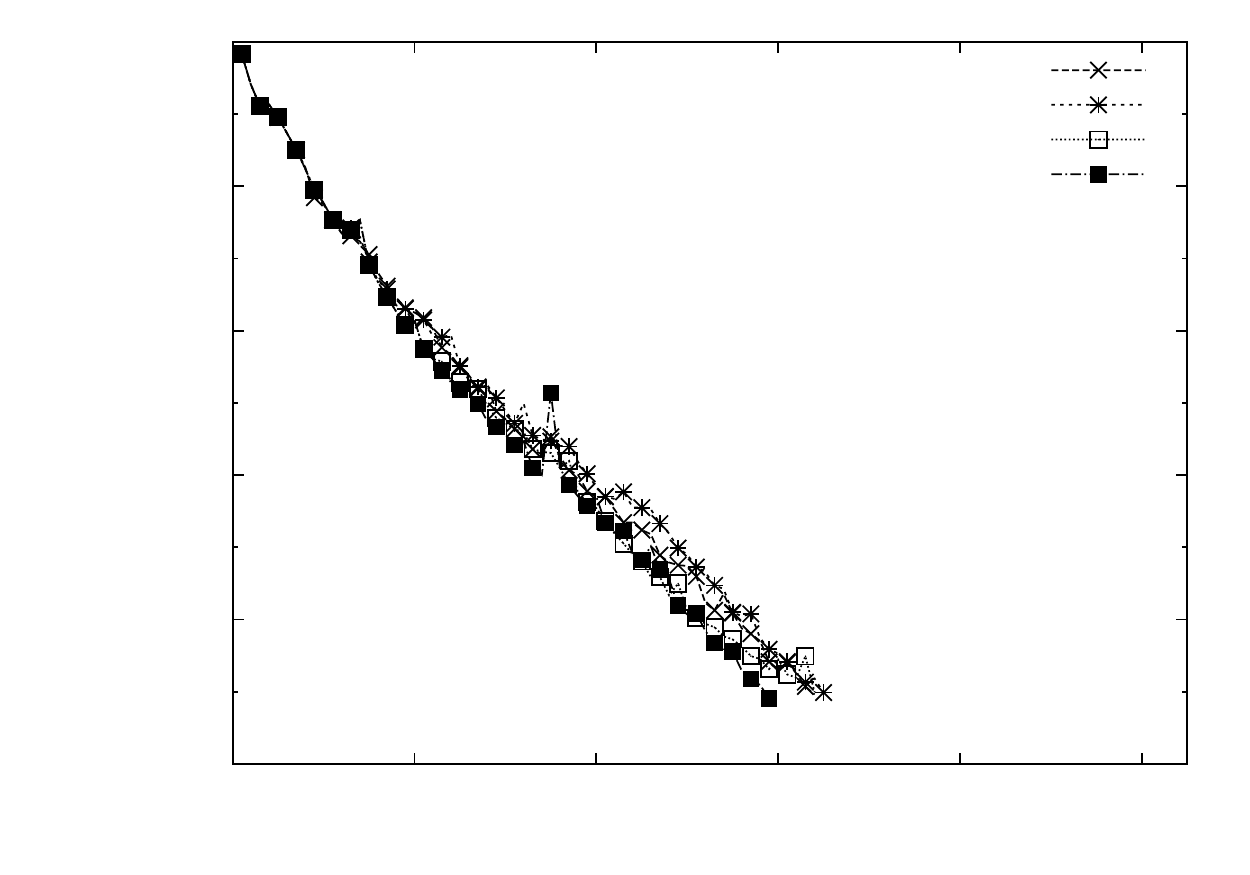}%
\end{picture}%
\endgroup

}
 \label{fig:refl_broy_cap}}
&
\subfloat[JFNK]{
 \resizebox{80mm}{!}{
\begingroup%
\makeatletter%
\newcommand{\GNUPLOTspecial}{%
  \@sanitize\catcode`\%=14\relax\special}%
\setlength{\unitlength}{0.0500bp}%
\begin{picture}(7200,5040)(0,0)%
  \put(5936,4036){\makebox(0,0)[r]{\strut{}GMRES(20), $\eta_{EW2}$}}%
  \put(5936,4236){\makebox(0,0)[r]{\strut{}GMRES(5), $\eta_{EW2}$}}%
  \put(5936,4436){\makebox(0,0)[r]{\strut{}GMRES(20), $\eta=0.1$}}%
  \put(5936,4636){\makebox(0,0)[r]{\strut{}GMRES(5), $\eta=0.1$}}%
  \put(4089,140){\makebox(0,0){\strut{}Sweeps}}%
  \put(160,2719){%
\rotatebox{-270}{%
  \makebox(0,0){\strut{}$\| {F} \|_{2,s}$}%
}}%
  \put(6577,440){\makebox(0,0){\strut{} 100}}%
  \put(5530,440){\makebox(0,0){\strut{} 80}}%
  \put(4482,440){\makebox(0,0){\strut{} 60}}%
  \put(3435,440){\makebox(0,0){\strut{} 40}}%
  \put(2387,440){\makebox(0,0){\strut{} 20}}%
  \put(1340,440){\makebox(0,0){\strut{} 0}}%
  \put(1220,4799){\makebox(0,0)[r]{\strut{} 1}}%
  \put(1220,3967){\makebox(0,0)[r]{\strut{} 0.01}}%
  \put(1220,3135){\makebox(0,0)[r]{\strut{} 0.0001}}%
  \put(1220,2304){\makebox(0,0)[r]{\strut{} 1e-06}}%
  \put(1220,1472){\makebox(0,0)[r]{\strut{} 1e-08}}%
  \put(1220,640){\makebox(0,0)[r]{\strut{} 1e-10}}%
\includegraphics{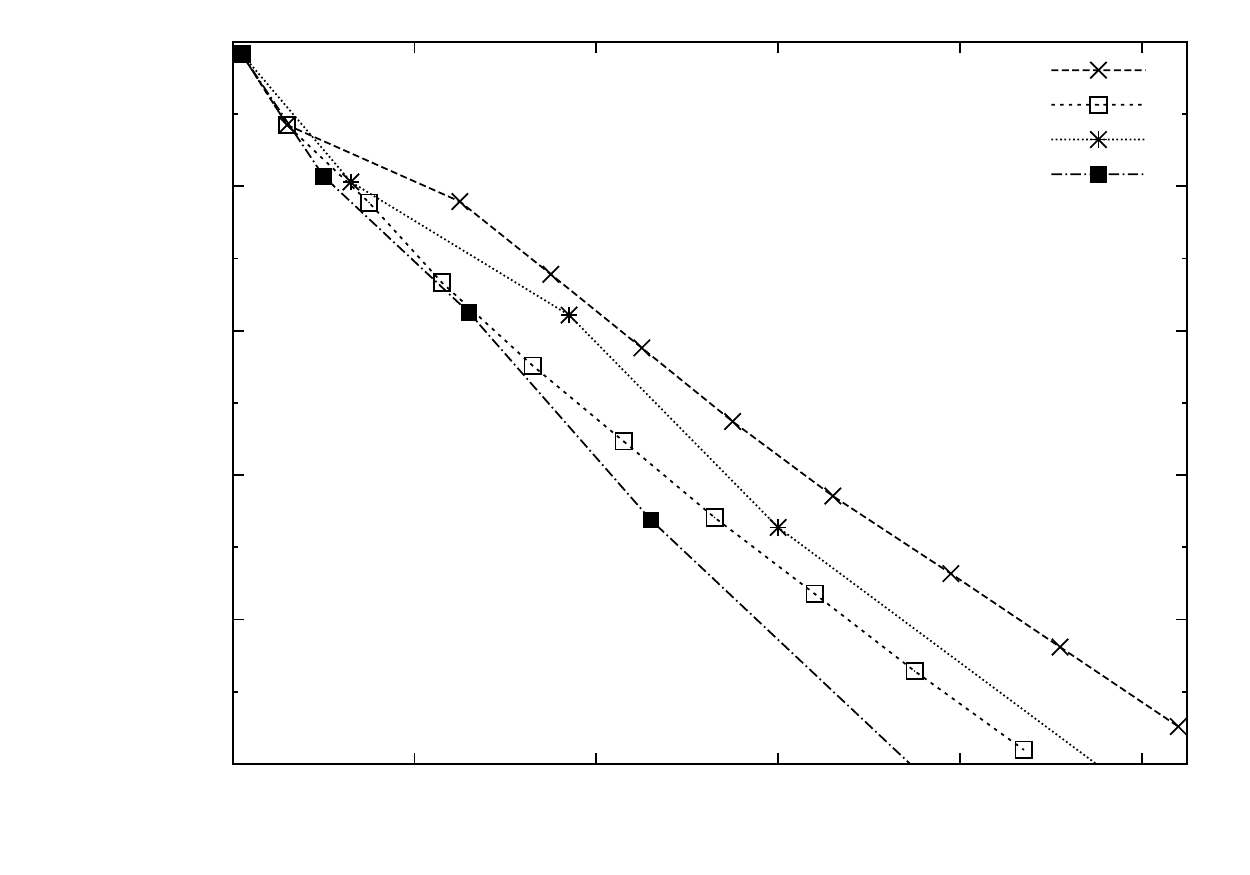}%
\end{picture}%
\endgroup

}
 \label{fig:refl_jfnk_cap}}
\\
\multicolumn{2}{c}{
\subfloat[Comparison of the best results for each method]{
 \resizebox{80mm}{!}{
\begingroup%
\makeatletter%
\newcommand{\GNUPLOTspecial}{%
  \@sanitize\catcode`\%=14\relax\special}%
\setlength{\unitlength}{0.0500bp}%
\begin{picture}(7200,5040)(0,0)%
  \put(5936,3836){\makebox(0,0)[r]{\strut{}JFNK (GMRES(20), $\eta_{EW2}$)}}%
  \put(5936,4036){\makebox(0,0)[r]{\strut{}Broyden(30)}}%
  \put(5936,4236){\makebox(0,0)[r]{\strut{}NKA$_{-1}$(30)}}%
  \put(5936,4436){\makebox(0,0)[r]{\strut{}NKA(20)}}%
  \put(5936,4636){\makebox(0,0)[r]{\strut{}FPI}}%
  \put(4089,140){\makebox(0,0){\strut{}Sweeps}}%
  \put(160,2719){%
\rotatebox{-270}{%
  \makebox(0,0){\strut{}$\| {F} \|_{2,s}$}%
}}%
  \put(6577,440){\makebox(0,0){\strut{} 100}}%
  \put(5530,440){\makebox(0,0){\strut{} 80}}%
  \put(4482,440){\makebox(0,0){\strut{} 60}}%
  \put(3435,440){\makebox(0,0){\strut{} 40}}%
  \put(2387,440){\makebox(0,0){\strut{} 20}}%
  \put(1340,440){\makebox(0,0){\strut{} 0}}%
  \put(1220,4799){\makebox(0,0)[r]{\strut{} 1}}%
  \put(1220,3967){\makebox(0,0)[r]{\strut{} 0.01}}%
  \put(1220,3135){\makebox(0,0)[r]{\strut{} 0.0001}}%
  \put(1220,2304){\makebox(0,0)[r]{\strut{} 1e-06}}%
  \put(1220,1472){\makebox(0,0)[r]{\strut{} 1e-08}}%
  \put(1220,640){\makebox(0,0)[r]{\strut{} 1e-10}}%
\includegraphics{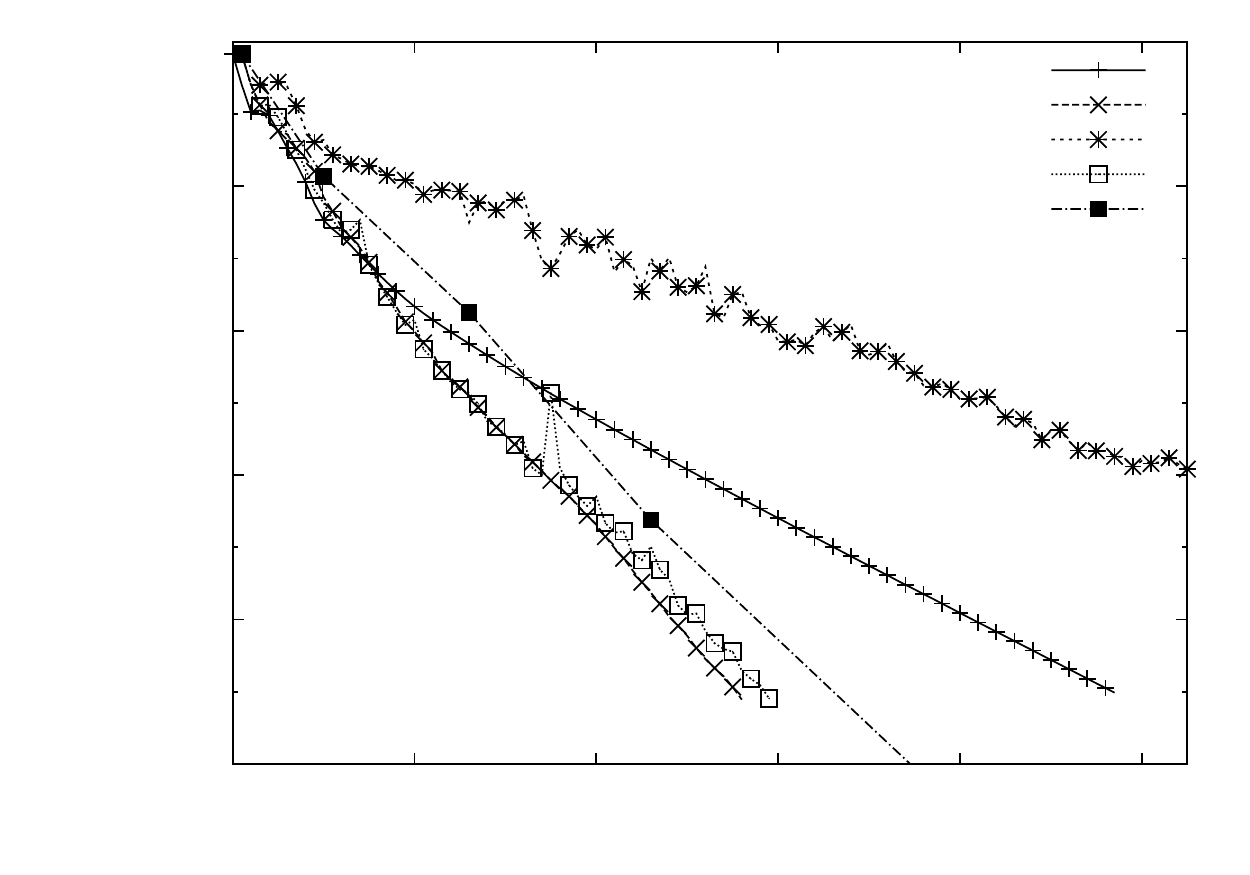}%
\end{picture}%
\endgroup
}
 \label{fig:refl_comp_cap}}}
\\
\end{tabular}
\label{fig:refl_cap}
\end{figure}

\subsection{C5G7-MOX benchmark}

\begin{figure}
\caption{\label{fig:c5g7-view}The C5G7-MOX configuration viewed from above.}
\begin{center}
\includegraphics[height=17cm]{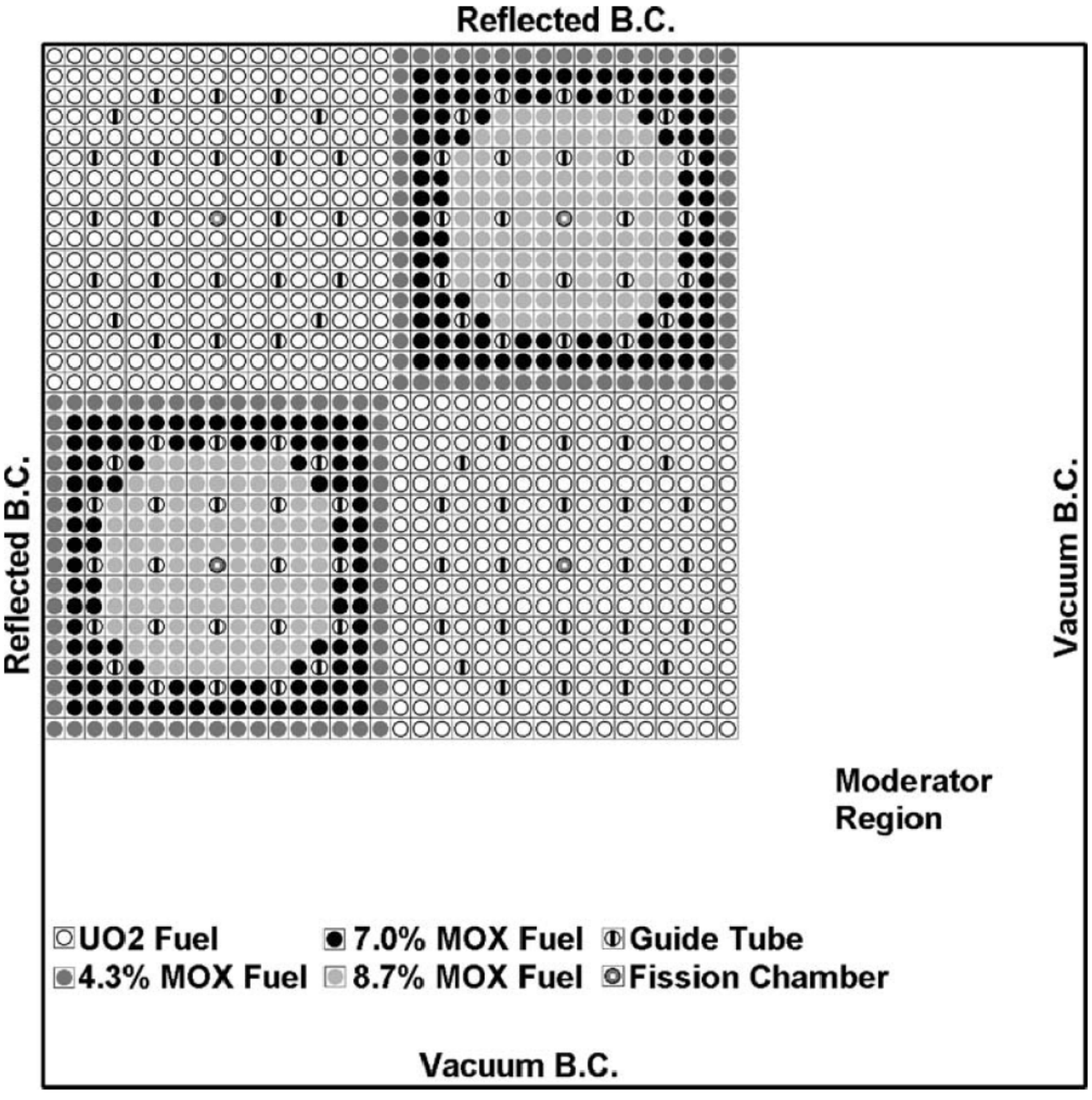}
\end{center}
\end{figure}

We chose the Rodded B variant~\cite{C5G7} of the C5G7-MOX problem and
modeled it as a quarter-core with reflecting boundary conditions on
the interior edges. Fig.~\ref{fig:c5g7-view} depicts the configuration from above. The spatial mesh was $506\times 506\times 360$ for
a total of 92.2 million spatial cells and a total of 645.2 million
degrees of freedom. The problem was run on the Typhoon cluster at Los
Alamos National Laboratory on 512 processors. We considered the system converged when the residual reached $\|F\|_{2,s} < 10^{-8}$. Despite the fact that
this problem is much larger and more complex than the cylinder
problem, all of the conclusions that were made in~\ref{cyl test} are
corroborated here, as can be seen in Table~\ref{table:c5g7} and
Fig.~\ref{fig:c5g7}. Even for a subspace size of 5, NKA still
converged in significantly fewer iterations and shorter runtimes than
any of the other methods even when the subspace size for the other
method was 30. NKA$_{-1}$ did not converge for \textit{any} of the
subspace sizes and Broyden only converged for subspace sizes of 10 and
5. This is the reverse of the PARTISN Broyden results for the
reflected cylinder problem, which did not converge for subspaces sizes
of 5 and 10, but did converged for 20 and, in even fewer iterations,
30. We believe this has to do with the fact that the eigenvalue is
negative for the first and second iterations (which are the same for
all of the methods) and the smaller subspace size allows those `bad'
iterates to be eliminated from the subspace sooner, thereby increasing
stability. The performance of JFNK is somewhat erratic, requiring
anywhere from 245 to 618 sweeps with no clear correlation to the
subspace size, and JFNK with $\eta_{EW1}$ and a subspace of 5
stagnates at $\|F\|_{2,s}\approx 10^{-7}$. The norm also fluctuates
more erratically at times for this problem than it did for the
previous problems, but NKA steadily decreases the norm at every
iteration. Fig.~\ref{fig:c5g7_comp} shows conclusively that NKA is the
most efficient nonlinear solver of the five examined for this
problem, even when NKA is run with a depth of five vectors. The percentage of CPU time spent in the residual
(Table~\ref{table:c5g7-perc}) is slightly higher in this case than the
others, which is doubtlessly due to the communications overhead that
results from doing a parallel computation.

\begin{table}[h!]
\centering
\caption{PARTISN C5G7-MOX Benchmark: Number of outer and inner JFNK
  iterations, sweeps and run-time spent computing the residual to an
  accuracy of $\|F\|_{2,s}\leq 10^{-9}$ for the various
  methods. \label{table:c5g7}}
%
\subfloat[Outer JFNK/total inner GMRES iterations]
{
\begin{tabular}{ccc}  \toprule[1pt]
  & \multicolumn{2}{c}{$\eta$} \\ \cline{2-3}
\raisebox{1.5ex}[0cm][0cm]{subspace}
	 &  EW1     &  EW2	\\\hline
30       &  13(204) &  15(556) \\
20       &  11(237) &  13(322) \\
10       &  10(321) &  6 (256) \\
5        &  --      &  11(502) \\ \bottomrule[1pt]
\label{table:c5g7-JFNK}
\end{tabular}}

\subfloat[Number of sweeps (FPI converged in 1223)]
{
\begin{tabular}{cccccc}  \toprule[1pt]
	&		&		&		& \multicolumn{2}{c}{JFNK $\eta$}  \\\cline{5-6}
\raisebox{1.5ex}[0cm][0cm]{subspace}
        &	\raisebox{1.5ex}[0cm][0cm]{NKA}
                        &	\raisebox{1.5ex}[0cm][0cm]{NKA$_{-1}$}
                                                        &     \raisebox{1.5ex}[0cm][0cm]{Broyden}
                          &   EW1 &  EW2	  \\\hline
30       & 90 & --  & --  &   245 &  599 \\
20       & 93 & --  & --  &   277 &  359 \\
10       & 137& --  & 690 &   377 &  290 \\
5        & 144& --  & 578 &   --  &  618 \\ \bottomrule[1pt]
\label{table:c5g7-sweep}
\end{tabular}}

\subfloat[CPU time (ks) (FPI converged in 10.02 ks)]
{
\begin{tabular}{cccccc}  \toprule[1pt]
	&		&		&		& \multicolumn{2}{c}{JFNK $\eta$}  \\\cline{5-6}
\raisebox{1.5ex}[0cm][0cm]{subspace}
        &	\raisebox{1.5ex}[0cm][0cm]{NKA}
                        &	\raisebox{1.5ex}[0cm][0cm]{NKA$_{-1}$}
                                                        &     \raisebox{1.5ex}[0cm][0cm]{Broyden}
                             &   EW1  &  EW2	  \\\hline
30       & 0.83 & --   & --   &   3.9  &  5.9   \\
20       & 0.94 & --   & --   &   4.4  &  3.5   \\
10       & 1.2  & --   & 6.2  &   5.9  &  2.8   \\
5        & 1.3  & --   & 5.1  &   --   &  6.0   \\ \bottomrule[1pt]
\label{table:c5g7-time}
\end{tabular}} 

\subfloat[Percentage of CPU time spent in the residual evaluation]
{
\begin{tabular}{cccccc}  \toprule[1pt]
	&		&		&		& \multicolumn{2}{c}{JFNK $\eta$}  \\\cline{5-6}
\raisebox{1.5ex}[0cm][0cm]{subspace}
        &	\raisebox{1.5ex}[0cm][0cm]{NKA}
                        &	\raisebox{1.5ex}[0cm][0cm]{NKA$_{-1}$}
                                                        &     \raisebox{1.5ex}[0cm][0cm]{Broyden}
                                                      	&	EW1	&	EW2	\\\hline
30	&	88.77	&	--   	&	--   	&	97.93	&	96.64	\\
20	&	82.12	&	--   	&	--   	&	98.13	&	97.40	\\
10	&	92.52	&	--   	&	93.56	&	98.51	&	97.91	\\
5	&	94.09	&	--	&	95.67	&	--   	&	98.39	\\\bottomrule[1pt]
\label{table:c5g7-perc}
\end{tabular}} 

\end{table}

\begin{figure}
\caption{PARTISN C5G7-MOX Benchmark: Scaled $L^2$-norm of the residual
as a function of number of sweeps for the various methods and subspace
sizes. Each of the methods is plotted on the same scale to simplify
comparisons between panels. NKA$_{-1}$ for all subspace sizes,
Broyden(20) and Broyden(30) were not included in these plots because
they did not converge and displayed very erratic behavior. Note that,
in panels (c) and (d), points plotted on the lines indicate when a
JFNK iteration starts (JFNK requires multiple sweeps per
iteration). In panels (a), and (b) we plot one point per two
iterations of the method. In panel (e) the convention for iterations
per plotted points is the same as in panels (a) through (d).}
\begin{tabular}{cc}
\subfloat[NKA and FPI]{ 
 \resizebox{80mm}{!}{
\begingroup%
\makeatletter%
\newcommand{\GNUPLOTspecial}{%
  \@sanitize\catcode`\%=14\relax\special}%
\setlength{\unitlength}{0.0500bp}%
\begin{picture}(7200,5040)(0,0)%
  \put(5996,3836){\makebox(0,0)[r]{\strut{}NKA(30)}}%
  \put(5996,4036){\makebox(0,0)[r]{\strut{}NKA(20)}}%
  \put(5996,4236){\makebox(0,0)[r]{\strut{}NKA(10)}}%
  \put(5996,4436){\makebox(0,0)[r]{\strut{}NKA(5)}}%
  \put(5996,4636){\makebox(0,0)[r]{\strut{}FPI}}%
  \put(4179,140){\makebox(0,0){\strut{}Sweeps}}%
  \put(280,2719){%
\rotatebox{-270}{%
  \makebox(0,0){\strut{}$\| {F} \|_{2,s}$}%
}}%
  \put(6899,440){\makebox(0,0){\strut{} 800}}%
  \put(6219,440){\makebox(0,0){\strut{} 700}}%
  \put(5539,440){\makebox(0,0){\strut{} 600}}%
  \put(4859,440){\makebox(0,0){\strut{} 500}}%
  \put(4180,440){\makebox(0,0){\strut{} 400}}%
  \put(3500,440){\makebox(0,0){\strut{} 300}}%
  \put(2820,440){\makebox(0,0){\strut{} 200}}%
  \put(2140,440){\makebox(0,0){\strut{} 100}}%
  \put(1460,440){\makebox(0,0){\strut{} 0}}%
  \put(1340,4799){\makebox(0,0)[r]{\strut{} 1}}%
  \put(1340,3967){\makebox(0,0)[r]{\strut{} 0.01}}%
  \put(1340,3135){\makebox(0,0)[r]{\strut{} 0.0001}}%
  \put(1340,2304){\makebox(0,0)[r]{\strut{} 1e-06}}%
  \put(1340,1472){\makebox(0,0)[r]{\strut{} 1e-08}}%
  \put(1340,640){\makebox(0,0)[r]{\strut{} 1e-10}}%
\includegraphics{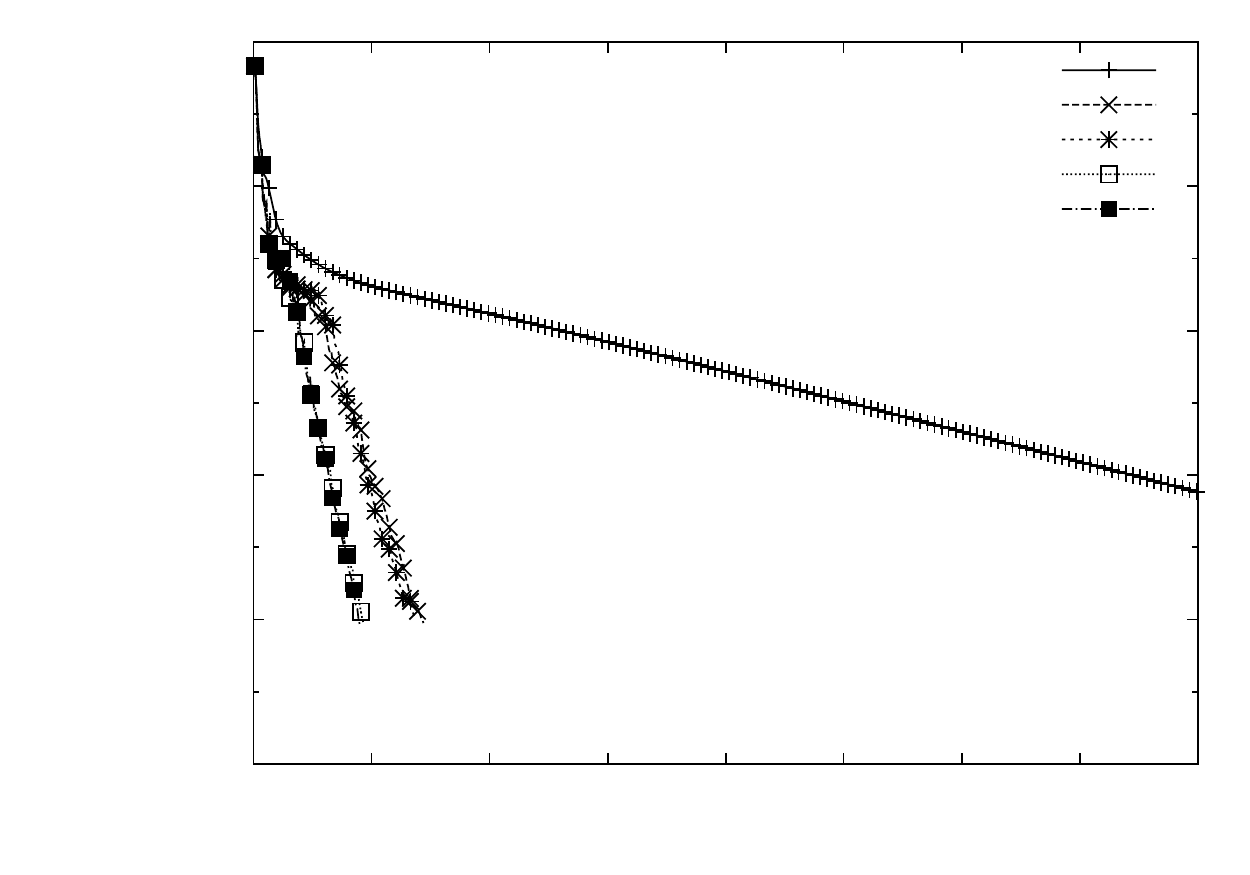}%
\end{picture}%
\endgroup
}
\label{fig:c5g7_nka}}
&
\subfloat[Broyden]{
 \resizebox{80mm}{!}{
\begingroup%
\makeatletter%
\newcommand{\GNUPLOTspecial}{%
  \@sanitize\catcode`\%=14\relax\special}%
\setlength{\unitlength}{0.0500bp}%
\begin{picture}(7200,5040)(0,0)%
  \put(5996,4436){\makebox(0,0)[r]{\strut{}Broyden(10)}}%
  \put(5996,4636){\makebox(0,0)[r]{\strut{}Broyden(5)}}%
  \put(4179,140){\makebox(0,0){\strut{}Sweeps}}%
  \put(280,2719){%
\rotatebox{-270}{%
  \makebox(0,0){\strut{}$\| {F} \|_{2,s}$}%
}}%
  \put(6899,440){\makebox(0,0){\strut{} 800}}%
  \put(6219,440){\makebox(0,0){\strut{} 700}}%
  \put(5539,440){\makebox(0,0){\strut{} 600}}%
  \put(4859,440){\makebox(0,0){\strut{} 500}}%
  \put(4180,440){\makebox(0,0){\strut{} 400}}%
  \put(3500,440){\makebox(0,0){\strut{} 300}}%
  \put(2820,440){\makebox(0,0){\strut{} 200}}%
  \put(2140,440){\makebox(0,0){\strut{} 100}}%
  \put(1460,440){\makebox(0,0){\strut{} 0}}%
  \put(1340,4799){\makebox(0,0)[r]{\strut{} 1}}%
  \put(1340,3967){\makebox(0,0)[r]{\strut{} 0.01}}%
  \put(1340,3135){\makebox(0,0)[r]{\strut{} 0.0001}}%
  \put(1340,2304){\makebox(0,0)[r]{\strut{} 1e-06}}%
  \put(1340,1472){\makebox(0,0)[r]{\strut{} 1e-08}}%
  \put(1340,640){\makebox(0,0)[r]{\strut{} 1e-10}}%
\includegraphics{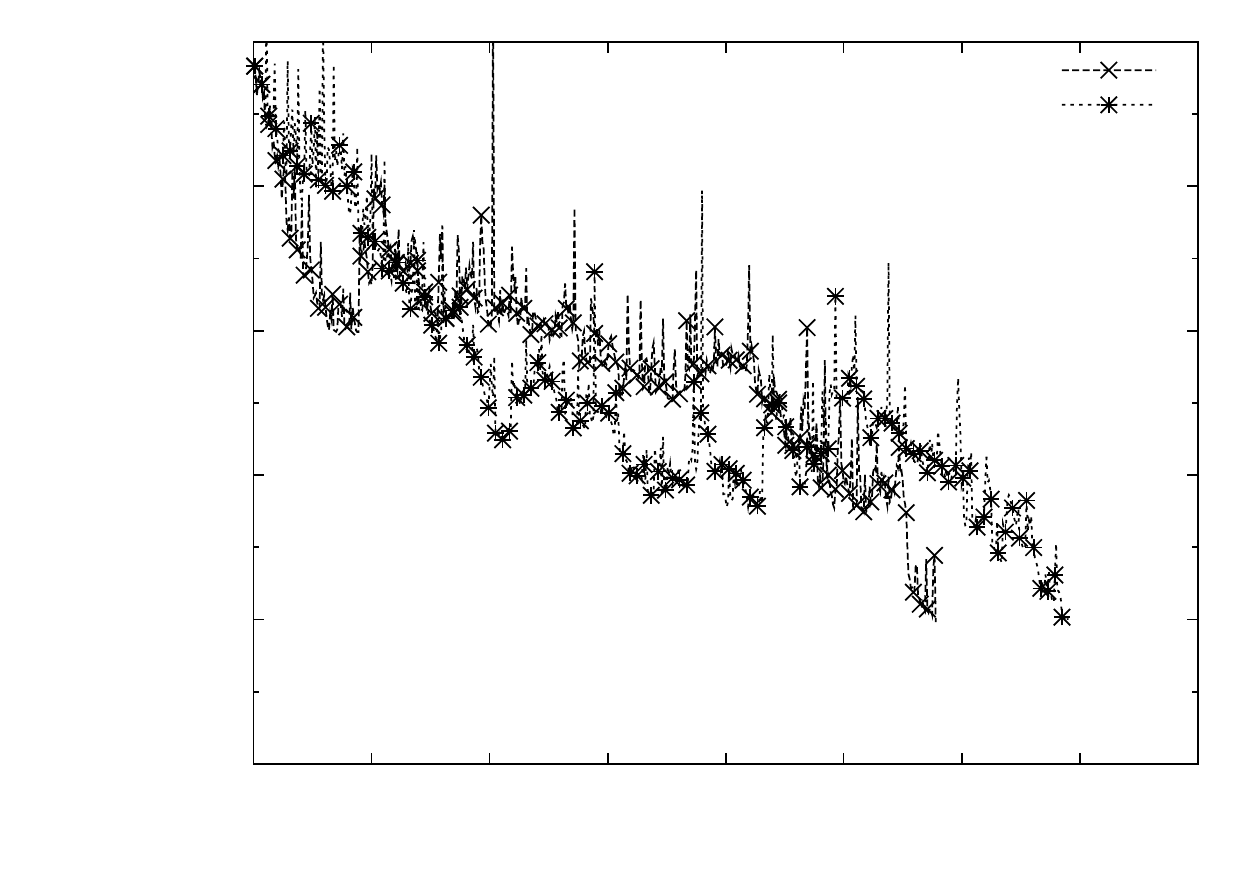}%
\end{picture}%
\endgroup
}
\label{fig:c5g7_broy}}
\\
\subfloat[JFNK $\eta_{EW1}$]{
 \resizebox{80mm}{!}{
\begingroup%
\makeatletter%
\newcommand{\GNUPLOTspecial}{%
  \@sanitize\catcode`\%=14\relax\special}%
\setlength{\unitlength}{0.0500bp}%
\begin{picture}(7200,5040)(0,0)%
  \put(5996,4036){\makebox(0,0)[r]{\strut{}GMRES(30)}}%
  \put(5996,4236){\makebox(0,0)[r]{\strut{}GMRES(20)}}%
  \put(5996,4436){\makebox(0,0)[r]{\strut{}GMRES(10)}}%
  \put(5996,4636){\makebox(0,0)[r]{\strut{}GMRES(5)}}%
  \put(4179,140){\makebox(0,0){\strut{}Sweeps}}%
  \put(280,2719){%
\rotatebox{-270}{%
  \makebox(0,0){\strut{}$\| {F} \|_{2,s}$}%
}}%
  \put(6899,440){\makebox(0,0){\strut{} 800}}%
  \put(6219,440){\makebox(0,0){\strut{} 700}}%
  \put(5539,440){\makebox(0,0){\strut{} 600}}%
  \put(4859,440){\makebox(0,0){\strut{} 500}}%
  \put(4180,440){\makebox(0,0){\strut{} 400}}%
  \put(3500,440){\makebox(0,0){\strut{} 300}}%
  \put(2820,440){\makebox(0,0){\strut{} 200}}%
  \put(2140,440){\makebox(0,0){\strut{} 100}}%
  \put(1460,440){\makebox(0,0){\strut{} 0}}%
  \put(1340,4799){\makebox(0,0)[r]{\strut{} 1}}%
  \put(1340,3967){\makebox(0,0)[r]{\strut{} 0.01}}%
  \put(1340,3135){\makebox(0,0)[r]{\strut{} 0.0001}}%
  \put(1340,2304){\makebox(0,0)[r]{\strut{} 1e-06}}%
  \put(1340,1472){\makebox(0,0)[r]{\strut{} 1e-08}}%
  \put(1340,640){\makebox(0,0)[r]{\strut{} 1e-10}}%
\includegraphics{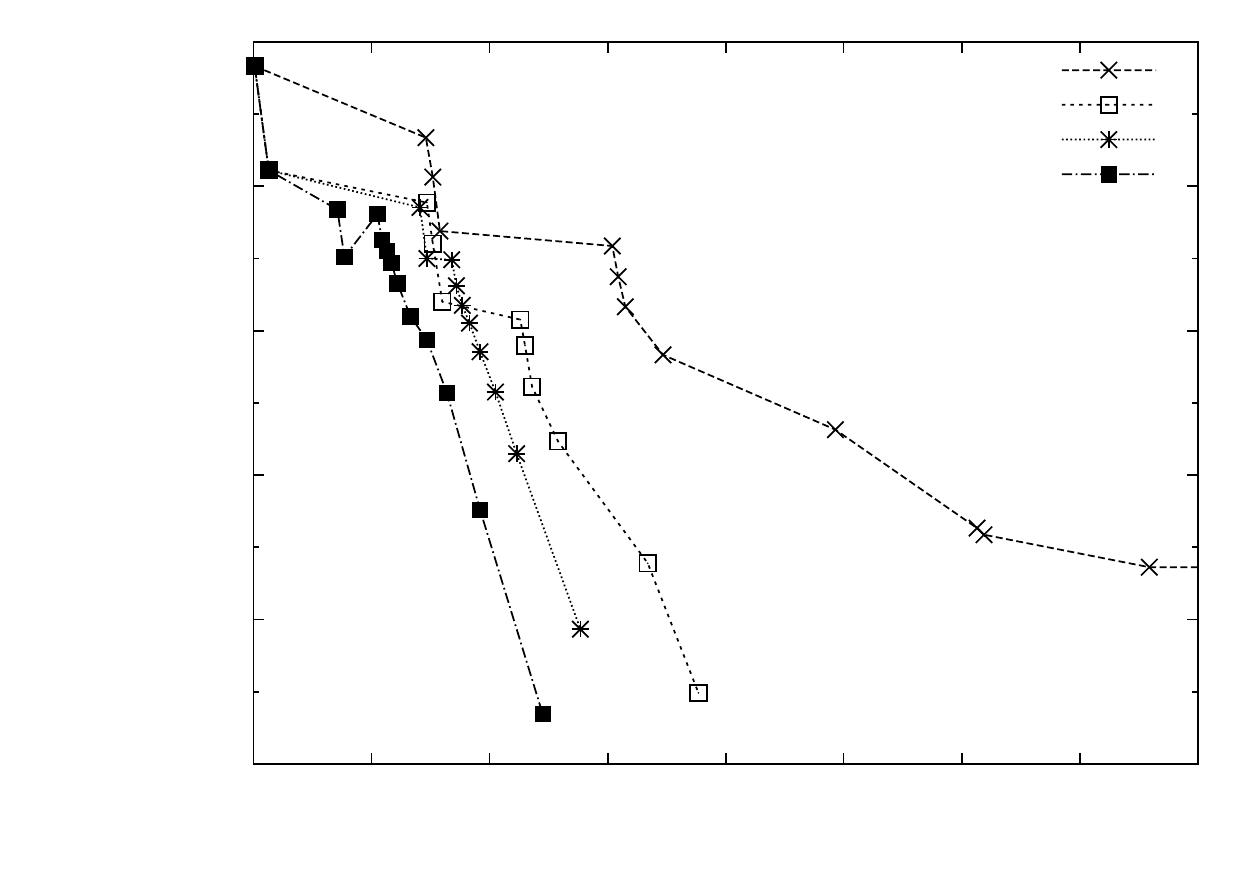}%
\end{picture}%
\endgroup
}
\label{fig:c5g7_jfnk1}}
&
\subfloat[JFNK $\eta_{EW2}$]{
 \resizebox{80mm}{!}{
\begingroup%
\makeatletter%
\newcommand{\GNUPLOTspecial}{%
  \@sanitize\catcode`\%=14\relax\special}%
\setlength{\unitlength}{0.0500bp}%
\begin{picture}(7200,5040)(0,0)%
  \put(5996,4036){\makebox(0,0)[r]{\strut{}GMRES(30)}}%
  \put(5996,4236){\makebox(0,0)[r]{\strut{}GMRES(20)}}%
  \put(5996,4436){\makebox(0,0)[r]{\strut{}GMRES(10)}}%
  \put(5996,4636){\makebox(0,0)[r]{\strut{}GMRES(5)}}%
  \put(4179,140){\makebox(0,0){\strut{}Sweeps}}%
  \put(280,2719){%
\rotatebox{-270}{%
  \makebox(0,0){\strut{}$\| {F} \|_{2,s}$}%
}}%
  \put(6899,440){\makebox(0,0){\strut{} 800}}%
  \put(6219,440){\makebox(0,0){\strut{} 700}}%
  \put(5539,440){\makebox(0,0){\strut{} 600}}%
  \put(4859,440){\makebox(0,0){\strut{} 500}}%
  \put(4180,440){\makebox(0,0){\strut{} 400}}%
  \put(3500,440){\makebox(0,0){\strut{} 300}}%
  \put(2820,440){\makebox(0,0){\strut{} 200}}%
  \put(2140,440){\makebox(0,0){\strut{} 100}}%
  \put(1460,440){\makebox(0,0){\strut{} 0}}%
  \put(1340,4799){\makebox(0,0)[r]{\strut{} 1}}%
  \put(1340,3967){\makebox(0,0)[r]{\strut{} 0.01}}%
  \put(1340,3135){\makebox(0,0)[r]{\strut{} 0.0001}}%
  \put(1340,2304){\makebox(0,0)[r]{\strut{} 1e-06}}%
  \put(1340,1472){\makebox(0,0)[r]{\strut{} 1e-08}}%
  \put(1340,640){\makebox(0,0)[r]{\strut{} 1e-10}}%
\includegraphics{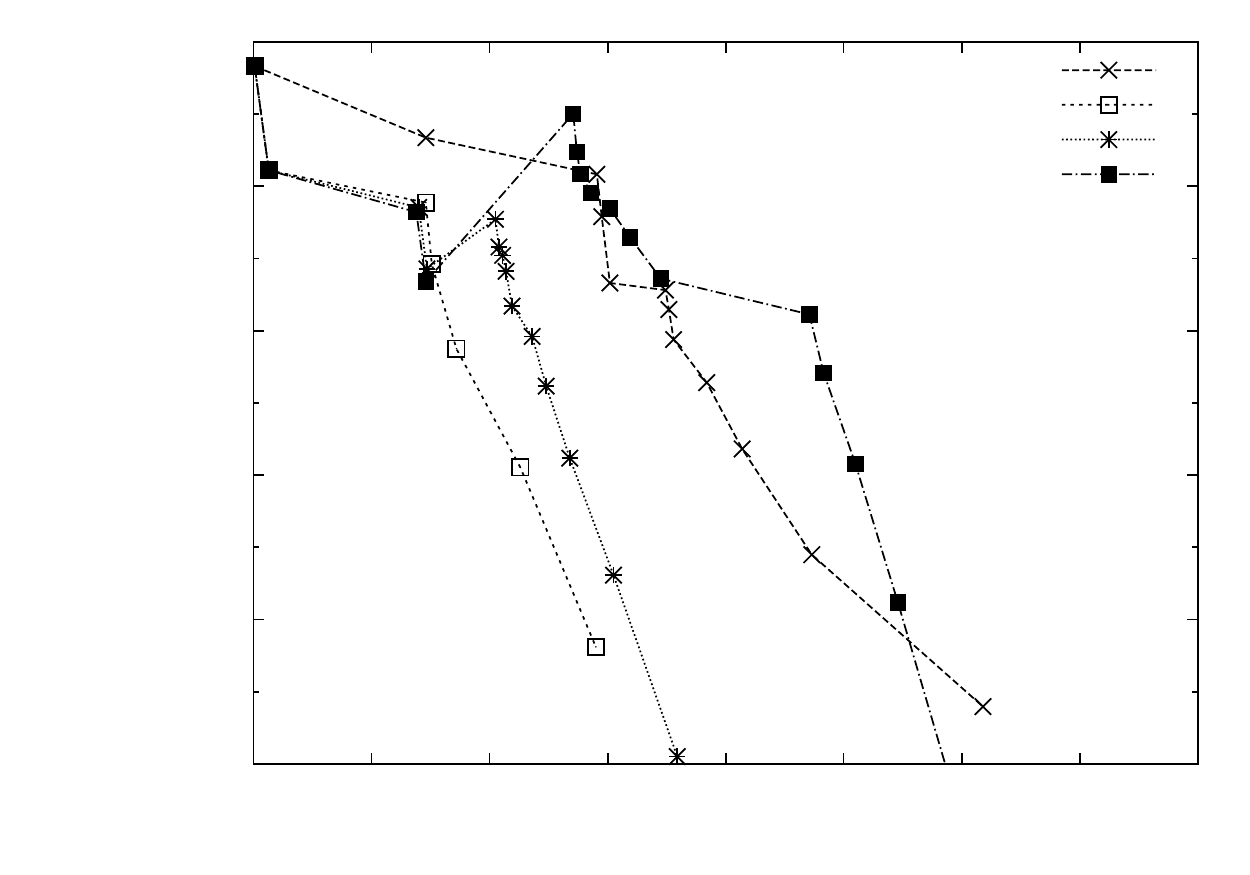}%
\end{picture}%
\endgroup
}
\label{fig:c5g7_jfnk2}}
\\
\multicolumn{2}{c}{
\subfloat[Comparison of the best results for Broyden, JFNK and NKA(5)]{
 \resizebox{80mm}{!}{
\begingroup%
\makeatletter%
\newcommand{\GNUPLOTspecial}{%
  \@sanitize\catcode`\%=14\relax\special}%
\setlength{\unitlength}{0.0500bp}%
\begin{picture}(7200,5040)(0,0)%
  \put(5996,3836){\makebox(0,0)[r]{\strut{}JFNK (GMRES(10), $\eta_{EW2}$)}}%
  \put(5996,4036){\makebox(0,0)[r]{\strut{}JFNK (GMRES(30), $\eta_{EW1}$)}}%
  \put(5996,4236){\makebox(0,0)[r]{\strut{}Broyden(5)}}%
  \put(5996,4436){\makebox(0,0)[r]{\strut{}NKA(5)}}%
  \put(5996,4636){\makebox(0,0)[r]{\strut{}FPI}}%
  \put(4179,140){\makebox(0,0){\strut{}Sweeps}}%
  \put(280,2719){%
\rotatebox{-270}{%
  \makebox(0,0){\strut{}$\| {F} \|_{2,s}$}%
}}%
  \put(6899,440){\makebox(0,0){\strut{} 800}}%
  \put(6219,440){\makebox(0,0){\strut{} 700}}%
  \put(5539,440){\makebox(0,0){\strut{} 600}}%
  \put(4859,440){\makebox(0,0){\strut{} 500}}%
  \put(4180,440){\makebox(0,0){\strut{} 400}}%
  \put(3500,440){\makebox(0,0){\strut{} 300}}%
  \put(2820,440){\makebox(0,0){\strut{} 200}}%
  \put(2140,440){\makebox(0,0){\strut{} 100}}%
  \put(1460,440){\makebox(0,0){\strut{} 0}}%
  \put(1340,4799){\makebox(0,0)[r]{\strut{} 1}}%
  \put(1340,3967){\makebox(0,0)[r]{\strut{} 0.01}}%
  \put(1340,3135){\makebox(0,0)[r]{\strut{} 0.0001}}%
  \put(1340,2304){\makebox(0,0)[r]{\strut{} 1e-06}}%
  \put(1340,1472){\makebox(0,0)[r]{\strut{} 1e-08}}%
  \put(1340,640){\makebox(0,0)[r]{\strut{} 1e-10}}%
\includegraphics{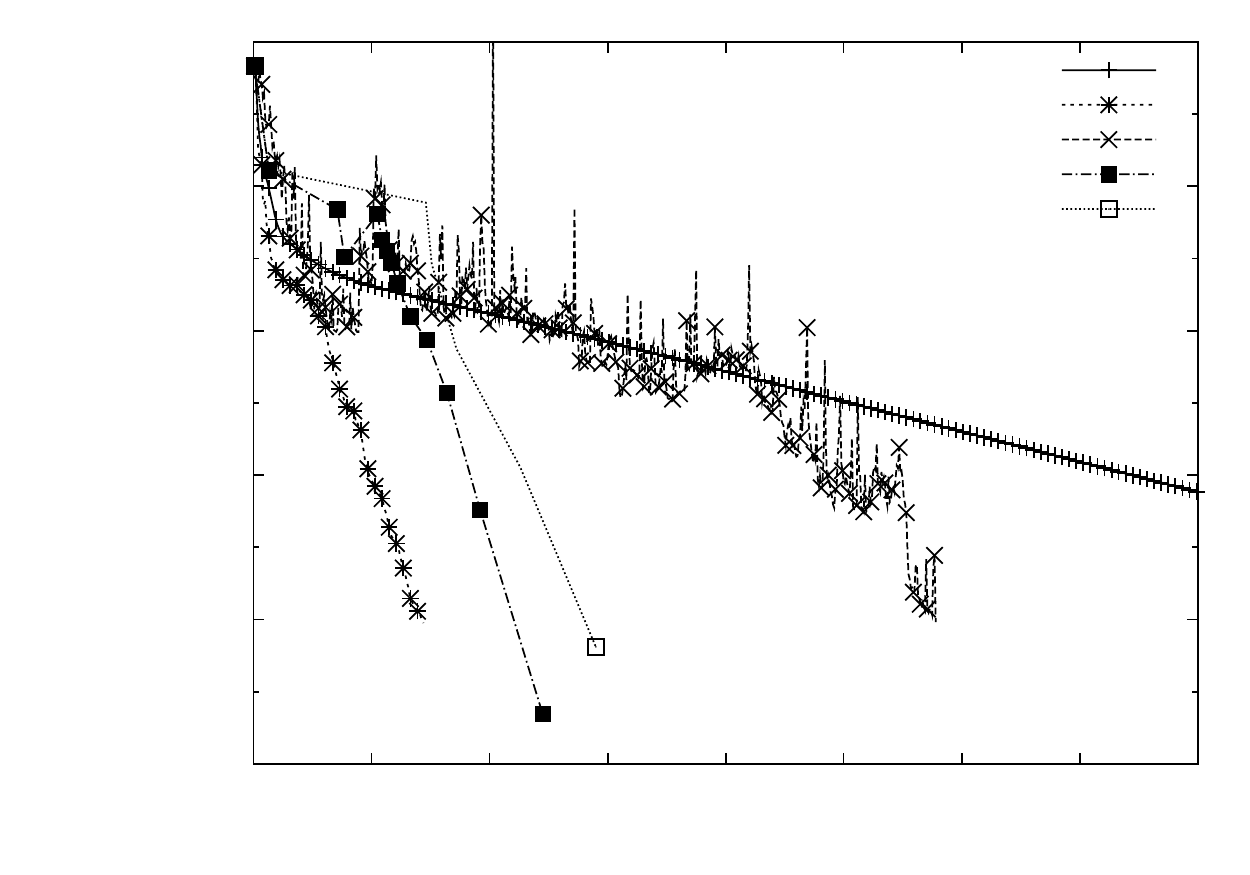}%
\end{picture}%
\endgroup
}
\label{fig:c5g7_comp}}}
\\
\end{tabular}
\label{fig:c5g7}
\end{figure}

\newpage
\clearpage
\section{Conclusion}\label{sec:concl}

The benefit of JFNK is its ``Newton-like'' convergence achieved by
developing an arbitrarily accurate approximation of the inverse of the
Jacobian at each `outer' iteration, but the cost is repeated function
evaluations in each of these `outer' iterations and the wasteful
discarding of potentially useful subspace information.  In contrast
Broyden and Anderson Mixing only perform a single function evaluation
at each iteration, but continue to use `old' information from previous
iterations to improve their estimate of the Jacobian. The drawback for
these methods is that the approximate Jacobian or its inverse is based
on an amalgam of new and old information, so they are unlikely to
converge in fewer iterations than Newton's method.  Performance of all
these methods will clearly depend on how the Jacobian is changing from
iteration to iteration, and how information collected at each function
evaluation is used.  Memory requirements for Broyden are half that of
NKA making it an attractive alternative.

Analytic examinations in~\cite{WalkerNi:2010} show that variants of
Anderson Mixing, including NKA, behave like a Krylov method (such as
GMRES) when applied to most linear problems.  Further, the hypotheses
of non-stagnation of GMRES can be removed at the cost of a modest
performance penalty.  It is worth noting that the authors of this
paper know of no theoretical results regarding Anderson Acceleration
or NKA for nonlinear problems.  This is in sharp contrast to Broyden
and JFNK.

Our numerical results indicate that Anderson Mixing in the form of NKA
found solutions in the PARTISN and Capsaicin codes for all problems
examined in the fewest number of function evaluations and the shortest
runtimes.  These are very large-scale problems involving approximately
one and eight million degrees of freedom for PARTISN and Capsaicin,
respectively, for the cylinder problem and 645 million degrees of
freedom in PARTISN for the C5G7 problem.  Our results highlight the
strength of this method for the problems at hand: regularity,
consistency and efficiency. In our results, NKA was shown to bring the
norm down to zero smoothly, much as FPI and JFNK do, but with greater
efficiency than those methods. Broyden and NKA$_{-1}$, while they at
times achieved excellent performance, did not always demonstrate this
same smooth convergence behavior and often diverged. Based on these
results we feel that NKA may be well-suited to other computational
physics problems beyond neutron transport.

\newpage
\clearpage

\section{Acknowledgment}
\noindent
The authors are grateful to Dana Knoll for pointing out the work of
Walker and Ni and the work of Fang and Saad and for discussions on
nonlinear solvers and their applications. The authors also gratefully
acknowledge Jon Dahl for providing the C5G7-MOX PARTISN input files.

\vspace{2mm} 
\noindent
This work was performed under the auspices of the
National Nuclear Security Administration of the US Department of
Energy at Los Alamos National Laboratory under Contract
No. DE-AC52-06NA25396. LA-UR 12-24529

\appendix
\section{Proof of Theorem~\ref{mtc:th:2}}\label{apdx:1}

\begin{proof}[Proof of Theorem~\ref{mtc:th:2}]
From the linearity of the problem ${\bf w}_i = {\bf A}{\bf v}_i$, and
from the update rule ${\bf x}_{n+1} = {\bf x}_n - {\bf v}_{n+1}$, we
have ${\bf x}_n = {\bf x}_0 - \sum_{i=1}^n {\bf v}_i,$ and then
\begin{equation*}
{\bf r}_n := f({\bf x}_n) = {\bf A}{\bf x}_n - {\bf b} = {\bf r}_0 - {\bf A}\sum_{i=1}^n {\bf v}_i,
\end{equation*}
from which we have
\begin{equation*}
f({\bf x}_n) - \sum_{i=1}^n z_i^{(n)}{\bf w}_i = {\bf r}_0 - {\bf A}\sum_{i=1}^n (1+z_i^{(n)}){\bf v}_i.
\end{equation*}
We use the above to compute the $n+1$ residual using the unmodified NKA update
\begin{eqnarray*}
{\bf r}_{n+1} =& {\bf A}{\bf x}_{n+1} - {\bf b}\\
=& {\bf A}{\bf x}_n - {\bf b} - {\bf A}{\bf v}_{n+1}\\
=& {\bf A}{\bf x}_n - {\bf b} - {\bf A}
\left[
\sum_{i=1}^n {\bf v}_i z_i^{(n)} + f({\bf x}_n) - \sum_{i=1}^n {\bf w}_i z_i^{(n)}
\right]\\
=& {\bf r}_0 - {\bf A}\sum_{i=1}^n {\bf v}_i 
- {\bf A}
\left[
\sum_{i=1}^n {\bf v}_i z_i^{(n)} + {\bf r}_0 - {\bf A}\sum_{i=1}^n (1+z_i^{(n)}){\bf v}_i.
\right]\\
=& ({\bf I} - {\bf A})
\left({\bf r}_0 - {\bf A}\sum_{i=1}^n (1+z_i^{(n)}){\bf v}_i\right).
\end{eqnarray*}

If $\|{\bf w}_n\|_2 = \|{\bf Av}_n\|_2 \ne 0$, then the
modified NKA update changes the $n^\text{th}$ coefficient as follows
\begin{equation*}
z_n^{(n)} \leftarrow z_n^{(n)} \pm \varepsilon 
\frac{\left \|f({\bf x}_n) - \sum_{i=1}^n {\bf w}_i z_i^{(n)}\right\|_2}
{\|{\bf w}_n\|_2} = 
z_n^{(n)} \pm \varepsilon 
\frac{\left \|{\bf r}_0 - {\bf A}\sum_{i=1}^n (1+z_i^{(n)}){\bf v}_i \right\|_2}
{\|{\bf A}{\bf v}_n\|_2}. 
\end{equation*}
The $n^\text{th}$ residual becomes
\begin{equation*}
{\bf r}_{n+1}=
({\bf I} - {\bf A})
\left({\bf r}_0 - {\bf A}\sum_{i=1}^n (1+z_i^{(n)}){\bf v}_i
\pm \varepsilon 
\frac{\left \|{\bf r}_0 - {\bf A}\sum_{i=1}^n (1+z_i^{(n)}){\bf v}_i \right\|_2}
{\|{\bf A}{\bf v}_n\|_2}{\bf A}{\bf v}_n
\right),
\end{equation*}
from which we have 
\begin{eqnarray*}
\|{\bf r}_{n+1}\|_2 \le&
\|{\bf I} - {\bf A}\|
\left\|{\bf r}_0 - {\bf A}\sum_{i=1}^n (1+z_i^{(n)}){\bf v}_i
\pm \varepsilon 
\frac{\left \|{\bf r}_0 - {\bf A}\sum_{i=1}^n (1+z_i^{(n)}){\bf v}_i \right\|_2}
{\|{\bf A}{\bf v}_n\|_2}{\bf A}{\bf v}_n
\right\|_2\\
\le &
(1+ \varepsilon )
\|{\bf I} - {\bf A}\|
\left\|{\bf r}_0 - {\bf A}\sum_{i=1}^n (1+z_i^{(n)}){\bf v}_i\right\|_2,
\end{eqnarray*}
where $\|\cdot\|$ denotes the operator 2-norm. Let ${\bf z}^{(n)} = (z_1^{(n)},\ldots,z_n^{(n)})$ denote the solution to the minimization problem
\begin{equation*}
{\bf z}^{(n)} = 
\argmin_{{\bf y}\in \RealVect{n}} \left\|f({\bf x}_n) - \sum_{i=1}^n {\bf w}_i y_i\right\|_2 = 
\argmin_{{\bf y}\in \RealVect{n}} \left\| {\bf r}_0 - {\bf A}\sum_{i=1}^n (1+y_i){\bf v}_i\right\|_2
\end{equation*}
and let ${\mathcal V}_n = \spn\{ {\bf v}_1,\ldots,{\bf v}_n\}$, then the above gives
\begin{equation}\label{mtc:eq:ap1}
\|{\bf r}_{n+1}\|_2 \le 
(1 + \varepsilon)\| {\bf I} - {\bf A}\|
\min_{{\bf v} \in {\mathcal V}_n}\| {\bf r}_0 - {\bf A}{\bf v}\|_2.
\end{equation}
If the modification of $z^{(n)}_n$ is not performed, then the factor
of $(1+\varepsilon)$ would simply be one.  In either case the above
bound on $\|{\bf r}_{n+1}\|_2$ given in \req{mtc:eq:ap1} holds.  

We now turn our attention to the question of when the Krylov spaces
(${\mathcal K}_n := \spn\{{\bf r}_0, {\bf A}{\bf r}_0,\ldots, {\bf
A}^{n-1}{\bf r}_0\}$) are expanding with $n$ and when ${\bf w}_n \ne 0$.  In what follows we assume ${\bf
r}_0 \ne 0$, i.e. that our initial guess ${\bf x}_0$ is not a
solution. Let $M$ be the lowest natural number so that ${\mathcal
K}_{M+1} \subseteq {\mathcal K}_M$. Then the vectors ${\bf r}_0, {\bf
A}{\bf r}_0,\ldots,{\bf A}^{M-1}{\bf r}_0$ are linearly independent
and there is a unique set of coefficients
$\alpha_0,\ldots,\alpha_{M-1}$ such that
\begin{equation*}
{\bf A}^M{\bf r}_0 = \sum_{i=0}^{M-1}\alpha_i {\bf A}^i{\bf r}_0.
\end{equation*}
If $\alpha_0 = 0$, then we would have
\begin{equation*}
{\bf A}\left[
{\bf A}^{M-1}{\bf r}_0 - \sum_{i=0}^{M-2}\alpha_{i+1} {\bf A}^i{\bf r}_0\right] = 0,
\end{equation*}
and, because ${\bf A}$ is non-singular, this would imply that
${\mathcal K}_{M} \subseteq {\mathcal K}_{M-1}$ contradicting the
definition of $M$. As such we have
\begin{equation*}
{\bf r}_0 = \frac 1 \alpha_0 \left[
{\bf A}^M {\bf r}_0 - \sum_{i=1}^{M-1}\alpha_i {\bf A}^i {\bf r}_0
\right] = {\bf A}\left(
\frac{1}{\alpha_0} \left[
{\bf A}^{M-1}{\bf r}_0 - \sum_{i=0}^{M-2} \alpha_{i+1}{\bf A}^i {\bf r}_0
\right]
\right),
\end{equation*}
and hence a solution ${\bf x}^* \in {\mathcal K}_M$ to ${\bf Ax} = {\bf r}_0$, from which we have ${\bf A}({\bf x}_0 - {\bf x}^*) = {\bf b}$.  We assume NKA will be stopped at this point.  This allows us to conclude if we have not found
a solution to the problem within ${\mathcal K}_M$, then ${\mathcal
K}_M\subsetneq {\mathcal K}_{M+1}$.

It now suffices to show that ${\mathcal V}_n = {\mathcal K}_n$.  We
begin by noting ${\mathcal V}_1 = \spn\{{\bf v}_1\} = \spn\{{\bf
r}_0\} = {\mathcal K}_1$, and, because ${\bf A}$ is of full rank,
${\bf w}_1 = {\bf A}{\bf v}_1 = {\bf A}{\bf r}_0 \ne 0$.  Our
inductive hypothesis is that there is some natural number $n$,
e.g. $n=1$, so that ${\bf w}_n \ne 0$ and that for all natural numbers
$j\le n$, ${\mathcal V}_j = {\mathcal K}_j$ and ${\bf A}{\bf x}_j -
{\bf b}\ne 0$.  That we have not found a solution in the first $n$
iterations is sufficient to conclude from above arguments that
${\mathcal K}_1\subsetneq\ldots\subsetneq{\mathcal K}_{n+1}$.

The NKA update rule, where the modification to $z^{(n)}_n$ has already
been applied (the modification is possible because ${\bf w}_n \ne 0$), may be written as
\begin{eqnarray*}
{\bf v}_{n+1} = &
\sum_{i=1}^n {\bf v}_i z_i^{(n)} + f({\bf x}_n) - \sum_{i=1}^n {\bf w}_i z_i^{(n)}\\
=&
\sum_{i=1}^n {\bf v}_i z_i^{(n)} + {\bf r}_0 - 
{\bf A}\sum_{i=1}^n {\bf v}_i (1 + z_i^{(n)})\\
=&
\left\{
\begin{array}{ccc}
\left[(1 + z^{(1)}_1){\bf r}_0 \right]
- {\bf A}{\bf r}_0(1 + z_n^{(1)}) & \quad\text{for}\quad n=1&\quad\text{(recall ${\bf v}_1 = {\bf r}_0$})\\
\left[\sum_{i=1}^n {\bf v}_i z_i^{(n)} + {\bf r}_0\right] - 
\left[{\bf A}\sum_{i=1}^{n-1} {\bf v}_i (1 + z_i^{(n)}) \right]
- {\bf A}{\bf v}_n(1 + z_n^{(n)}) & \quad\text{for}\quad n\ne1 &\\.
\end{array}\right.
\end{eqnarray*}

In the case $n=1$, ${\bf v}_2 = (1+z^{(1)}_1)({\bf r}_0 - {\bf A}{\bf
r}_0)$.  Our modification ensures that $1+z^{(1)}_1 \ne 0$. If ${\bf
r}_0 = {\bf A}{\bf r}_0$, then NKA has found the solution ${\bf x}_1 =
{\bf x}_0 - {\bf r}_0$.  We have excluded this by assumption and so
${\mathcal V}_2 = \spn\{{\bf v}_1,{\bf v}_2\} = \spn\{{\bf r}_0,{\bf
A}{\bf r}_0\} = {\mathcal K}_2$ and ${\bf w}_2 = {\bf A}{\bf v}_2 \ne
0$. Our inductive hypothesis holds for $n=2$

For the cases $n\ne 1$, we know that ${\mathcal V}_{n} = {\mathcal
K}_{n}$, which is sufficient to conclude that the first bracketed term in the last of the preceding lines is in ${\mathcal K}_n$.  We also know that ${\mathcal V}_{n-1} = {\mathcal
K}_{n-1}$.  This implies that any linear combination of ${\bf
v}_1,\ldots,{\bf v}_{n-1}$ lies in ${\mathcal K}_{n-1}$, and so any
linear combination of ${\bf A}{\bf v}_1,\ldots,{\bf A}{\bf v}_{n-1}$
lies within ${\mathcal K}_n$.  We conclude that the quantities in
brackets in the last of the preceding lines lies within ${\mathcal
K}_n$.  Since ${\bf x}_n$ is not the solution, $\spn\{{\bf
v}_1,\ldots,{\bf v}_{n-1}\} = {\mathcal K}_{n-1} \subsetneq {\mathcal
K}_n = \spn\{{\bf v}_1,\ldots,{\bf v}_{n}\}$.  This is sufficient to
conclude that ${\bf v}_n$ may be written as the sum $\tilde {\bf v}
+\alpha{\bf A}^{n-1}{\bf r}_0$, where $\tilde {\bf v} \in {\mathcal V}_{n-1}
= {\mathcal K}_{n-1}$ and $\alpha \ne 0$.  By construction $(1+z^{(n)}_n) \ne 0$ and so
if ${\bf v}_{n+1} = 0$, then ${\mathcal K}_{n+1} \subseteq{\mathcal
K}_n$ and hence ${\bf x}_{n+1}$ will be a solution.  Since we assume
${\bf x}_{n+1}$ is not the solution, our inductive hypothesis holds
for $n+1$. \end{proof}

\bibliographystyle{elsarticle-num}
\bibliography{biblio}

\end{document}